\documentclass[journal,twoside,web]{ieeecolor}

\usepackage{generic}
\usepackage{cite}
\usepackage{amsfonts}
\usepackage{graphicx}
\usepackage{textcomp}

\let\labelindent\relax
\usepackage{xcolor}
\usepackage{tikz}
\usetikzlibrary{positioning}
\usetikzlibrary{arrows}
\usetikzlibrary{circuits.ee.IEC}
\usepackage{booktabs}
\usepackage[font=footnotesize]{caption}
\usepackage[font=footnotesize]{subcaption}
\usepackage[shortlabels]{enumitem}

\definecolor{backgroundblue}{rgb}{0.6, 0.8, 1}
\definecolor{backgroundorange}{rgb}{1, 0.752, 0}

\let\proof\relax
\let\endproof\relax
\usepackage{amsthm}
\usepackage[cmex10]{amsmath}
\newtheorem{theorem}{Theorem}
\newtheorem{lemma}{Lemma}
\newtheorem{definition}{Definition}

\newtheorem{corollary}{Corollary}
\newtheorem{remark}{Remark}

  \newenvironment{proofThm}{%
  \proof}{\endproof}

\usepackage{amssymb} 
\usepackage{amsfonts}
\usepackage{mathtools}

\allowdisplaybreaks
\usepackage[bottom]{footmisc}

\usepackage[capitalize,nameinlink,compress]{cleveref}
\crefname{figure}{Fig.}{Figures} 
\crefname{line}{line}{lines} 
\crefname{claim}{Claim}{Claims} 
\crefname{equation}{}{} 
\crefname{problem}{Problem}{Problems}
\crefname{assumption}{Assumption}{Assumptions}

\usepackage{soul}
\usepackage{framed}
\colorlet{shadecolor}{yellow}
\soulregister\cite7
\soulregister\cref7
\soulregister\ref7
\soulregister\pageref7
\soulregister\eqref7

\usepackage{makecell}

\renewcommand{\baselinestretch}{1}

\def\BibTeX{{\rm B\kern-.05em{\sc i\kern-.025em b}\kern-.08em
   T\kern-.1667em\lower.7ex\hbox{E}\kern-.125emX}}
\begin{document}
\title{Decentralized Parametric Stability Certificates for Grid-Forming Converter Control}
\author{Verena~Häberle,~\IEEEmembership{Graduate Student Member,~IEEE,}
        Xiuqiang~He,~\IEEEmembership{Member,~IEEE,}
        Linbin~Huang,~\IEEEmembership{Member,~IEEE,}
        Florian Dörfler,~\IEEEmembership{Senior Member,~IEEE,}
        and~Steven Low,~\IEEEmembership{Fellow,~IEEE}
\thanks{This work was supported by the EU Horizon 2023 research and innovation program (Grant Agreement Number 101096197). Verena Häberle and Florian Dörfler are with the Automatic Control Laboratory, ETH Zurich, 8092 Zurich, Switzerland. Email: \{verenhae,dorfler\}@ethz.ch. Xiuqiang He is with the Department of Automation, Tsinghua University, Beijing 100084, China. Email: hxq19@tsinghua.org.cn. Linbin Huang is with the College of Electrical Engineering, Zhejiang University, Hangzhou 310027, China. Email: hlinbin@zju.edu.cn. Steven Low is with the Department of Computing and Mathematical Sciences and the Department of Electrical Engineering, Caltech, Pasadena, CA 91125 USA. Email: slow@caltech.edu.}
}

\maketitle

\begin{abstract}
We propose a decentralized framework to analytically guarantee the small-signal stability of future power systems with grid-forming converters. Our approach leverages dynamic loop-shifting techniques to compensate for the lack of passivity in the network dynamics and establishes decentralized parametric stability certificates, depending on the local device-level controls and incorporating the effects of the network. By following practical tuning rules, we are able to ensure plug-and-play operation without centralized coordination. Unlike prior works, our approach accommodates coupled frequency and voltage dynamics, incorporates network dynamics, and does not rely on specific network configurations or operating points, offering a general and scalable solution for the integration of power-electronics-based devices into future power systems. We validate our theoretical stability results through numerical case studies in a high-fidelity simulation model.
\end{abstract}


\section{Introduction}
\label{sec:introduction}

\IEEEPARstart{T}{he} transition to future power systems is characterized by a substantial increase in the share of power electronics (PE)-based generation devices \cite{milano2018foundations}. This shift introduces significant changes in system dynamics, where the interactions between PE-based generation and the remainder of the power network are not fully understood yet \cite{wang2022identifying,NERC900,cheng2022real}. In particular, the faster dynamics of the control loops and filters in PE-based devices, compared to conventional generators, may induce unforeseen phenomena (e.g., overvoltages\cite{NERC900}, subsynchronous oscillations\cite{cheng2022real}), thereby posing significant challenges to system stability. Consequently, there is a pressing need for a deeper understanding of how PE-based devices interact with one another or with the grid and for the development of stability frameworks to ensure stable operation. 

One of the key aspects of future power systems involves the integration of PE-based devices with either grid-following or grid-forming control. Grid-forming (GFM) control is particularly promising as it enhances stability by establishing a well-defined ac voltage at the grid connection \cite{milano2018foundations,matevosyan2019grid}. However, it has been observed that network dynamics can significantly influence the stability of PE-dominated systems \cite{gross2019effect,markovic2021understanding}, unlike traditional synchronous generator-dominated grids, which exhibit relatively slow dynamics compared to the fast network dynamics. This creates a need for analysis methods that account for the interplay between network dynamics and the control characteristics of PE-based devices.  

In this paper, we address the destabilizing effects of network dynamics on the small-signal stability of heterogeneous interconnected GFM converters. Our approach seeks to compensate for these effects by proposing \textit{parametric} decentralized stability certificates to ensure system-wide stability. Specifically, we introduce a decentralized framework to analytically certify stability at the individual device level through \textit{local tuning rules} for each GFM controller. This enables plug-and-play operation without centralized coordination. Our approach leverages dynamic loop-shifting techniques to compensate for the lack of passivity in the network dynamics. By formulating device-level specifications that are sufficiently passive, we ensure overall system stability in a flexible and scalable way. 

Our contribution improves significantly over prior works on analytic stability certification of GFM converter systems in the small-signal regime. Unlike \cite{pates2019robust,gross2022compensating}, which are limited to single-input single-output (SISO) dynamics, our approach accommodates multiple-input multiple-output (MIMO) dynamics coupling frequency and voltage, thus providing a more comprehensive analysis of the system behavior. Additionally, we incorporate the network dynamics, extending beyond the quasi-stationary or zero-power flow approximations in \cite{pates2019robust,siahaan2024decentralized,watson2020control}. Our decentralized stability conditions enable the use of heterogeneous GFM controls, improving over the homogeneity assumptions in \cite{gross2022compensating,subotic2021lyapunov}. Moreover, unlike \cite{huang2024gain,decentralized2019vorobev}, our analytic approach does not require detailed knowledge of the network configuration. Collectively, these attributes represent a step toward small-signal stability certification methods that explicitly account for network dynamics while accommodating heterogeneous GFM controls and limited network information. While further theoretical developments remain necessary, the proposed conditions provide analytically tractable and explicitly parametric criteria that can support stability assessment and controller design, and thereby contribute to the reliable integration of power-electronics-based generation.

The paper is structured as follows. Section \ref{sec:interconnected_systems} introduces preliminary concepts of feedback stability and passivity for linear time-invariant (LTI) systems. In Section \ref{sec:power_system_model}, we describe the dynamic power system model utilized for stability analysis. Section \ref{sec:decentralized_stability_conditions} presents our main results: decentralized parametric stability certificates for GFM converters that account for network dynamics. The results are validated through numerical case studies in Section \ref{sec:numerical_case_studies}. Section \ref{sec:conclusion} concludes the paper.

\section{Foundations of Interconnected Systems}\label{sec:interconnected_systems}
\subsection{Preliminaries}\label{sec:preliminaries}
Let $\mathbb{R}$ denote the set of real numbers, $\mathbb{N}$ the set of positive integers, and $\mathbb{C}$ the set of complex numbers with imaginary unit $\mathrm{j}$. We use $I_n$ to denote the $n$-by-$n$ identity matrix (abbreviated as $I$ when the dimensions can be inferred from the context). We use $\otimes$ to denote the Kronecker product. We use $A=\mathrm{diag}(A_1,A_2,...,A_k),\, k\in\mathbb{Z}$ to denote the block-diagonal matrix with blocks $A_1,A_2,...,A_k$. The Euclidean norm of a matrix $A$ is defined as $||A||_F=\sqrt{\textstyle\sum_{i,j}|a_{ij}|^2}$, where $a_{ij}$ is the entry in the $i$th row and $j$th column of $A$. 

For a complex matrix $A\in\mathbb{C}^{n\times n}$, we use $A^\star$ to denote its conjugate transpose. A matrix $A\in\mathbb{C}^{n\times n}$ satisfying $A=A^\star$ is called \textit{Hermitian}. A Hermitian matrix $A\in\mathbb{C}^{n\times n}$ is said to be \textit{positive definite (semi-definite)}, denoted by $A \succ 0$ $(\succeq 0)$, if $x^\star Ax > 0$ $(\geq 0)$ for all $x\neq 0$. A matrix $A\in\mathbb{C}^{n\times n}$ is called \textit{diagonally dominant}, if $|a_{ii}|\geq \textstyle\sum_{j\neq i}|a_{ij}|,\,\forall i$.

To establish our main stability results in \cref{sec:decentralized_stability_conditions}, we require the following lemma which follows from the Gershgorin's Circle Theorem \cite[Thm. 6.1.10]{horn2012matrix}:
\begin{lemma}\label{lemma:gershgorin}
A Hermitian diagonally dominant matrix with real non-negative diagonal entries is positive semi-definite.
\end{lemma}
 
Further, we review the recent concept of {phases of complex matrices} based on the matrix's numerical range \cite{chen2021phase}. Namely, the \textit{numerical range} of a complex matrix $A\in\mathbb{C}^{n\times n}$ is defined as $W(A)=\{x^\star A x: x\in\mathbb{C}^n, ||x|| = 1 \}$. If $0\notin W(A)$, then $A$ is said to be a \textit{sectorial matrix}. For a sectorial $A$, there exists a nonsingular matrix $T$ and a diagonal unitary matrix $D$ such that $A=T^\star DT$ \cite{chen2021phase}. The diagonal elements of $D$ are distributed in an arc on the unit circle with length smaller than $\pi$. Then, the \textit{phases} of $A$, denoted by 
\begin{align}
    \overline{\phi}(A) = \phi_1(A)\geq \dots \geq \phi_n(A) = \underline{\phi}(A),
\end{align}
are defined as the phases of the diagonal entries of $D$ so that $\overline{\phi}(A)-\underline{\phi}(A) < \pi$. The definition of phases can be generalized to so-called \textit{semi-sectorial} matrices whose numerical ranges contain the origin on their boundaries and $\overline{\phi}(A)-\underline{\phi}(A) \leq \pi$.

\subsection{Transfer Functions \& Stability}
Let $u(t)\in\mathbb{R}^n$ be the input and $y(t)\in\mathbb{R}^n$ the output signal of a MIMO LTI system. The $n\times n$ transfer matrix $H(s)$ describes the input-output system response in the frequency domain as $y(s)=H(s)u(s)$, where $y(s)$ and $u(s)$ are the Laplace transformations of the output and input, respectively.
 
We denote the space of $n$-by-$n$ real rational proper transfer matrices of stable LTI systems by $\mathcal{RH}_\infty^{n\times n}$. An LTI system is called \textit{stable} if all poles are in the open left half plane (LHP). It is called \textit{semi-stable} if it may have poles on the imaginary axis but no poles in the open right half plane (RHP).

\begin{definition}[Internal Feedback Stability \cite{skogestad2005multivariable}]
Let $H_1(s)$ and $H_2(s)$ be $n\times n$ real rational proper transfer matrices. The feedback system in \cref{fig:internal_feedback_stability} is \textit{internally feedback stable} if and only if the following four closed-loop transfer matrices 
    \begin{align}\label{eq:gangof4}
        \hspace{-2.5mm}\begin{bmatrix}
            y_1\\y_2
        \end{bmatrix} \hspace{-1.2mm}= \hspace{-1.2mm}\underset{\eqqcolon \,H_1\#H_2(s)}{\underbrace{\begin{bmatrix}
            (I\hspace{-0.6mm}+\hspace{-0.6mm}H_1H_2)^{-1}H_1&-(I\hspace{-0.6mm}+\hspace{-0.6mm}H_1H_2)^{-1}H_1H_2 \\ H_2(I\hspace{-0.6mm}+\hspace{-0.6mm}H_1H_2)^{-1}H_1& H_2(I\hspace{-0.6mm}+\hspace{-0.6mm}H_1H_2)^{-1} 
        \end{bmatrix}}}\hspace{-1.5mm}\begin{bmatrix}
            w_1\\ w_2
        \end{bmatrix}\hspace{-1.6mm}\\[-0.7cm]
        \nonumber
    \end{align}
are stable, compactly referred to as $H_1\#H_2 \in \mathcal{RH}_\infty^{2n\times 2n}$. 
\end{definition}
If there are no RHP pole-zero cancellations between $H_1(s)$ and $H_2(s)$, then stability of one closed-loop transfer matrix implies stability of the others \cite[Thm 4.5]{skogestad2005multivariable}:

\begin{lemma}\label{lemma:rhp_cancellations}
Assume there are no RHP pole-zero cancellations between $H_1(s)$ and $H_2(s)$, i.e., all RHP poles in $H_1(s)$ and $H_2(s)$ are contained in the minimal realization of $H_1(s)H_2(s)$ and $H_2(s)H_1(s)$. Then, the feedback system in \cref{fig:internal_feedback_stability} is internally feedback stable if and only if one of the four closed-loop transfer matrices in \cref{eq:gangof4} is stable.
\end{lemma}

\begin{figure}[t!]
    \centering
\begin{subfigure}{0.21\textwidth}
    \centering
    \vspace{-1mm}
     \resizebox{1\textwidth}{!}{
\tikzstyle{roundnode}=[circle,draw=black!60,fill=black!5,scale=0.75]
\begin{tikzpicture}
\draw [rounded corners = 3,fill=backgroundblue!30] (-0.3,4.3) rectangle (1.3,3.5);
\node at (0.5,3.9) {$H_1(s)$};
\draw [rounded corners = 3,fill=backgroundblue!30] (-0.3,2.7) rectangle (1.3,1.9);
\node at (0.5,2.3) {$H_2(s)$};
\draw [-latex](1.3,3.9) -- (2.5,3.9);
\draw [-latex](-0.3,2.3) -- (-1.5,2.3);

\draw[-latex] (1.7,2.3) -- (1.3,2.3);
\draw [-latex](2.5,2.3) -- (1.9,2.3); 
\draw [-latex](1.8,3.9) -- (1.8,2.4);
\node [roundnode] at (1.8,2.3) {};

\draw [-latex](-0.7,3.9) -- (-0.3,3.9);
\draw [-latex](-1.5,3.9) -- (-0.9,3.9); 
\draw [-latex](-0.8,2.3) -- (-0.8,3.8);
\node [roundnode] at (-0.8,3.9) {};
\node at (-0.6,3.6) {$-$};
\node at (-1.3,4.1) {$w_1$};
\node at (2.3,4.1) {$y_1$};
\node at (-1.3,2.5) {$y_2$};
\node at (2.3,2.5) {$w_2$};
\end{tikzpicture}
}
        \vspace{-8.5mm}
    \caption{Original feedback system.}
    \vspace{-1mm}
    \label{fig:internal_feedback_stability}
\end{subfigure}
\hspace{0.05cm}
\begin{subfigure}{0.24\textwidth}
    \centering
     \vspace{-1mm}
   \resizebox{1\textwidth}{!}{
\tikzstyle{roundnode}=[circle,draw=black!60,fill=black!5,scale=0.75]
\begin{tikzpicture}
\draw [rounded corners = 3, dashed,fill=backgroundblue!10] (2,3.85) rectangle (-1.1,1.8);
\draw [rounded corners = 3, dashed, fill=backgroundblue!10] (2,4.1) rectangle (-1.1,6.15);
\draw [rounded corners = 3,fill=backgroundblue!30] (-0.3,5.9) rectangle (1.3,5.1);
\draw [rounded corners = 3,fill=backgroundblue!30] (-0.3,2.7) rectangle (1.3,1.9);
\node at (0.5,5.5) {$H_1(s)$};
\node [backgroundblue!200] at (1.65,5.85) {$H'_1$};
\node at (0.5,2.3) {$H_2(s)$};
\node [backgroundblue!200] at (1.65,3.55) {$H'_2$};
\draw [-latex](1.3,5.5) -- (3,5.5);
\draw (-0.9,2.3) -- (-1.4,2.3);

\draw[-latex] (2.2,2.3) -- (1.3,2.3);

\draw [-latex](2.3,5.5) -- (2.3,2.4);

\draw [-latex](-1.3,5.5) -- (-0.9,5.5);
\draw [-latex](-2.1,5.5) -- (-1.5,5.5); 
\draw [-latex](-1.4,2.3) node (v1) {} -- (-1.4,5.4);
\node [roundnode] at (-1.4,5.5) {};
\node at (-1.25,5.2) {-};
\node at (-1.9,5.7) {$w_1$};
\node at (2.8,5.7) {$y_1$};
\draw [rounded corners = 3, fill = black!20] (-0.3,3.6) rectangle (1.3,2.8);
\node at (0.5,3.2) {$\Gamma(s)$};
\draw [rounded corners = 3, fill = black!20] (-0.3,5) rectangle (1.3,4.2);
\node at (0.5,4.6) {$\Gamma(s)$};
\draw [-latex](1.8,2.3) -- (1.8,3.2) -- (1.3,3.2);
\draw [-latex](-0.3,2.3) -- (-0.7,2.3); 
\draw [-latex](-0.3,3.2) -- (-0.8,3.2) -- (-0.8,2.4);
\node [roundnode] at (-0.8,2.3) {}; 

\node[roundnode] at (-0.8,5.5) {}; 
\draw [-latex](-0.7,5.5) -- (-0.3,5.5);
\draw [-latex](-0.3,4.6) -- (-0.8,4.6) -- (-0.8,5.4);
\draw[-latex] (1.8,5.5) -- (1.8,4.6) -- (1.3,4.6);
\draw [-latex](-1.4,2.3) -- (-2.1,2.3);
\node at (-1.9,2.5) {$y_2$};
\node [roundnode]at (2.3,2.3) {};
\draw [-latex](3,2.3) -- (2.4,2.3);
\node at (2.8,2.5) {$w_2$};
\end{tikzpicture}
}
        \vspace{-8.5mm}
    \caption{Loop shifting with $\Gamma(s)$.}
        \vspace{-1mm}
    \label{fig:loop_shited_system}
\end{subfigure}
    \caption{\footnotesize Closed-loop feedback interconnection of two LTI systems.}
    \label{fig:loop_shifting}
    \vspace{-4mm}
\end{figure}

\subsection{Small-Phase \& Passivity Theory}\label{sec:small_phase_passivity}
For LTI systems, the property of passivity is equivalent to the property of positive realness \cite{willems1972dissipative,bao2007process,hassan2002nonlinear}.
\begin{definition}[Passivity]\label{def:passivity}
A $n\times n$ real rational proper transfer matrix $H(s)$ is said to be passive if 
\begin{enumerate}[label=(\roman*)]
    \item poles of all elements of $H(s)$ are in $\text{Re}(s) \leq 0$,
    \item $H(\mathrm{j}\omega) + H^\star(\mathrm{j}\omega)\succeq0$ for any $\omega$ for which $\mathrm{j}\omega$ is not a pole of any element of $H(s)$,
\item any purely imaginary pole $\mathrm{j}\omega$ of any element of $H(s)$ is a simple pole and $\text{lim}_{s\rightarrow \mathrm{j}\omega}(s\hspace{-0.7mm}-\hspace{-0.7mm}\mathrm{j}\omega)H(s)$ is positive semi-definite Hermitian.
\end{enumerate}
\end{definition}

\begin{definition}[Strict Passivity]\label{def:strict_passivity}
A $n\times n$ real rational proper transfer matrix $H(s)$ is said to be strictly passive if 
\begin{enumerate}[label=(\roman*)]
    \item poles of all elements of $H(s)$ are in $\text{Re}(s) < 0$,
    \item $H(\mathrm{j}\omega) + H^\star(\mathrm{j}\omega)\succ0$ for any $\omega\in(-\infty,\infty)$.
\end{enumerate}
\end{definition}

We now present a variant of the passivity theorem \cite{hassan2002nonlinear} for LTI systems as follows:
\begin{theorem}[Passivity Theorem]\label{thm:stabilty_passivity}
    Consider two LTI systems $H_1(s)$ and $H_2(s)$ in negative feedback configuration, as in \cref{fig:internal_feedback_stability}. The feedback system is internally feedback stable if 
    \begin{itemize}
        \item $H_1(s)$ is strictly passive,
        \item $H_2(s)$ is passive, and
        \item $\bar{\sigma}(H_1(\mathrm{j}\infty))\bar{\sigma}(H_2(\mathrm{j}\infty))<1$.
    \end{itemize}
\end{theorem}
\noindent Several variants of the above theorem exist, often for the SISO case, with subtle distinctions arising in the special case $\omega = \infty$. To avoid ambiguity, we provide a self-contained proof below, relying on the Mixed Gain–Phase Theorem \cite{zhao2022small}.

\begin{lemma}[Mixed Gain-Phase Theorem with Cutoff Frequency]\label{lemma:mixed_gain_phase}
    Let $\omega_\mathrm{c}\in(0,\infty]$, $H_1(s)\in\mathcal{RH}_\infty^{n\times n}$ be frequency-wise sectorial, and $H_2(s)$ be semi-stable frequency-wise semi-sectorial over $(-\omega_\mathrm{c},\omega_\mathrm{c})$ with $\mathrm{j}\Omega$ being the set of poles on the imaginary axis satisfying $\text{max}_{\omega\in\Omega}|\omega|< \omega_\mathrm{c}$. Then, $H_1\#H_2$ in \cref{fig:internal_feedback_stability} is internally feedback stable if
\begin{enumerate}[label=(\roman*)]
        \item for each $\omega\in[0,\omega_\mathrm{c})\backslash\Omega$, it holds
        \begin{align}\label{eq:mixed_gain_phase_cond1}
    \begin{split}
    &\overline{\phi}(H_1(\mathrm{j}\omega))+\overline{\phi}(H_2(\mathrm{j}\omega)) < \pi \quad \text{and}\\
 &\underline{\phi}(H_1(\mathrm{j}\omega))+\underline{\phi}(H_2(\mathrm{j}\omega)) > -\pi,
 \end{split}
    \end{align}
    \item and for each $\omega\in[\omega_\mathrm{c},\infty]$, it holds
    \begin{align}\label{eq:mixed_gain_phase_cond2}
        \bar{\sigma}(H_1(\mathrm{j\omega}))\bar{\sigma}(H_2(\mathrm{j\omega})) < 1.
    \end{align}
        \end{enumerate}
\end{lemma}

\begin{proofThm}
We establish \cref{thm:stabilty_passivity} by invoking \cref{lemma:mixed_gain_phase} with $\omega_{\mathrm{c}} = \infty$. Since \cref{lemma:mixed_gain_phase} holds for arbitrarily large $\omega_{\mathrm{c}}$, continuity ensures its validity also at $\omega_{\mathrm{c}} = \infty$. Accordingly, we proceed as follows. Because $H_1(s)$ is strictly passive, it follows that $H_1(s)\in\mathcal{RH}_\infty^{n\times n}$ is frequency-wise sectorial with phases $\left[\underline{\phi}(H_1(\mathrm{j}\omega)), \overline{\phi}(H_1(\mathrm{j}\omega))\right]\subset(-\tfrac{\pi}{2}, \tfrac{\pi}{2}), \forall \omega\in[0,\infty)$. Likewise, since $H_2(s)$ is passive, it follows that $H_2(s)$ is semi-stable with poles at $\mathrm{j}\Omega$ on the imaginary axis and frequency-wise semi-sectorial with phases $\left[\underline{\phi}(H_2(\mathrm{j}\omega)), \overline{\phi}(H_2(\mathrm{j}\omega))\right]\subset[-\tfrac{\pi}{2}, \tfrac{\pi}{2}], \forall \omega\in[0,\infty]\backslash \Omega$. Hence, the conditions in \cref{eq:mixed_gain_phase_cond1} in \cref{lemma:mixed_gain_phase} are satisfied. For $\omega=\infty$, it becomes immediate how \cref{eq:mixed_gain_phase_cond2} in \cref{lemma:mixed_gain_phase} is satisfied. This implies $H_1\#H_2 \in \mathcal{RH}_\infty^{2n\times 2n}$.
\end{proofThm}

To extend the passivity-based stability conditions to more general feedback interconnections with both passive and nonpassive subsystems, we employ the concepts of excess and shortage of passivity \cite{bao2007process,hassan2002nonlinear}. The basic idea is that the excess of passivity of one subsystem can compensate for the passivity deficit in the other subsystem, such that their feedback interconnection remains stable. This can be achieved by performing loop-shifting techniques as shown in \cref{fig:loop_shifting}, where a transfer matrix $\Gamma(s)\in\mathbb{C}^{n\times n}$ is added as a positive feedforward to $H_2(s)$, and as a positive feedback to $H_1(s)$. Now the idea is that if the original subsystems $H_1(s)$ and $H_2(s)$ in \cref{fig:internal_feedback_stability} are not satisfying the passivity conditions in \cref{thm:stabilty_passivity}, the two subsystems ${H}'_2(s)$ and ${H}'_1(s)$ in \cref{fig:loop_shited_system} might do so for a suitable $\Gamma(s)$, thus resulting in a stable feedback interconnection ${H}'_1\#{H}'_2$. Notice that this dynamic loop-shifting approach based on the dynamic transfer function $\Gamma(s)$ effectively corresponds to the concept of frequency-dependent (dynamic) passivity indices\cite{bao2007process,hassan2002nonlinear}, which, are generally less conservative than static passivity indices.

Among the four closed-loop transfer matrices in \cref{eq:gangof4}, it is immediately evident that only the upper-left transfer matrix mapping from $w_1$ to $y_1$ remains equivalent between the original system in \cref{fig:internal_feedback_stability} and the loop-shifted system in \cref{fig:loop_shited_system}. Specifically, $(I + {H}'_1{H}'_2)^{-1}{H}'_1 = (I + H_1H_2)^{-1}H_1$, where $H'_1=(I-H_1\Gamma)^{-1}H_1$ and $H'_2= H_2+\Gamma$. Hence, stability of ${H}'_1\#{H}'_2$ directly implies stability of $(I + H_1H_2)^{-1}H_1$. If there are no RHP pole-zero cancellations between $H_1(s)$ and $H_2(s)$, by \cref{lemma:rhp_cancellations}, stability of the upper-left transfer matrix $(I + H_1H_2)^{-1}H_1$ also guarantees stability of $H_1\#H_2$. 

\begin{remark} By swapping the feedforward and feedback of $\Gamma(s)$ in \cref{fig:loop_shited_system} and comparing the closed-loop transfer function from $w_2$ to $y_2$, similar stability conclusions can be made.
\end{remark}

\begin{remark}
Previous dynamic loop-shifting techniques based on the transfer matrix $\Gamma(s)$ are purely mathematical tools for establishing passivity properties and analytically proving closed-loop stability. They do not involve any hardware changes or modifications to the physical control system.
\end{remark}
\begin{figure}[t!]
\vspace{-1mm}
    \centering
    \resizebox{0.42\textwidth}{!}{

\tikzstyle{roundnode}=[circle,draw=black!60, fill=black!5,scale=0.5]
\begin{tikzpicture}[scale=1,every node/.style={scale=0.65}]
\draw (-2.2,1.4) node (v4) {} -- (0.2,0.4) node (v1) {} -- (-1.7,2.2) node (v3) {} --  (1.6,2)  -- (0.2,0.4) node (v5) {} -- (1.6,1) node (v2) {} -- (1.6,2) node (v6) {} -- (1.6,1) -- (1.6,1) node (v7) {};

\node [scale=1.8] at (0,1.6) {$\ddots$};
\node at (-1.7,2.4) {$1$};
\draw (-1.7,2.2) -- (-2.7,2.2); 
\draw (-2.2,1.4) -- (-3,1.4); 
\draw (0.2,0.4) -- cycle;
\draw (0.2,0.4) -- (-1.7,0.4); 
\draw (1.6,2) -- (2.1,2); 
\draw (1.6,1) -- (2.5,1);
\draw  [rounded corners = 3](-3.4,2.5) rectangle (-2.7,1.9);
\draw  [rounded corners = 3](-3,1.7) rectangle (-3.7,1.1);
\draw  [rounded corners = 3](-1.7,0.7) rectangle (-2.4,0.1);
\node  [scale=1.8] at (-2.6,1.1) {$\ddots$};
\draw  [rounded corners = 3](2.1,2.3) rectangle (2.8,1.7);
\draw  [rounded corners = 3](2.5,1.3) rectangle (3.2,0.7);

\node at (2.85,0.55) {converter};
\node at (-3.35,0.96) {converter};
\node at (-2.3,1.6) {2};
\node at (0.3,0.2) {$i$};
\node at (1.7,0.8) {$n-1$};
\node at (1.7,2.2) {$n$};

\draw (-3.57,1.39) -- (-3.47,1.39); 
\draw (-3.47,1.54) -- (-3.47,1.24); 
\draw (-3.42,1.54) -- (-3.42,1.24); 
\draw (-3.42,1.44) -- (-3.32,1.49) -- (-3.32,1.64); 
\draw (-3.42,1.34) -- (-3.32,1.29); 
\draw (-3.32,1.29) node (v9) {} -- (-3.32,1.14); 
\draw (-3.32,1.19) -- (-3.22,1.19) --  (-3.22,1.34);
\draw  (-3.22,1.44) -- (-3.22,1.59) -- (-3.32,1.59);
\draw (-3.27,1.44) -- (-3.17,1.44); 
\draw (-3.22,1.44) node (v8) {} -- (-3.27,1.34) -- (-3.17,1.34) -- (-3.22,1.44);
\draw (-3.35,1.33) --  (-3.32,1.29) -- (-3.37,1.29);

\draw (-2.28,0.39) -- (-2.18,0.39); 
\draw (-2.18,0.54) -- (-2.18,0.24); 
\draw (-2.13,0.54) -- (-2.13,0.24); 
\draw (-2.13,0.44) -- (-2.03,0.49) -- (-2.03,0.64); 
\draw (-2.13,0.34) -- (-2.03,0.29); 
\draw (-2.03,0.29) node (v9) {} -- (-2.03,0.14); 
\draw (-2.03,0.19) -- (-1.93,0.19) --  (-1.93,0.34);
\draw  (-1.93,0.44) -- (-1.93,0.59) -- (-2.03,0.59);
\draw (-1.98,0.44) -- (-1.88,0.44); 
\draw (-1.93,0.44) node (v8) {} -- (-1.98,0.34) -- (-1.88,0.34) -- (-1.93,0.44);
\draw (-2.06,0.33) --  (-2.03,0.29) -- (-2.08,0.29);

\draw (2.22,2.01) -- (2.32,2.01); 
\draw (2.32,2.16) -- (2.32,1.86); 
\draw (2.37,2.16) -- (2.37,1.86); 
\draw (2.37,2.06) -- (2.47,2.11) -- (2.47,2.26); 
\draw (2.37,1.96) -- (2.47,1.91); 
\draw (2.47,1.91) node (v9) {} -- (2.47,1.76); 
\draw (2.47,1.81) -- (2.57,1.81) --  (2.57,1.96);
\draw  (2.57,2.06) -- (2.57,2.21) -- (2.47,2.21);
\draw (2.52,2.06) -- (2.62,2.06); 
\draw (2.57,2.06) node (v8) {} -- (2.52,1.96) -- (2.62,1.96) -- (2.57,2.06);
\draw (2.44,1.95) --  (2.47,1.91) -- (2.42,1.91);

\draw (2.62,1.01) -- (2.72,1.01); 
\draw (2.72,1.16) -- (2.72,0.86); 
\draw (2.77,1.16) -- (2.77,0.86); 
\draw (2.77,1.06) -- (2.87,1.11) -- (2.87,1.26); 
\draw (2.77,0.96) -- (2.87,0.91); 
\draw (2.87,0.91) node (v9) {} -- (2.87,0.76); 
\draw (2.87,0.81) -- (2.97,0.81) --  (2.97,0.96);
\draw  (2.97,1.06) -- (2.97,1.21) -- (2.87,1.21);
\draw (2.92,1.06) -- (3.02,1.06); 
\draw (2.97,1.06) node (v8) {} -- (2.92,0.96) -- (3.02,0.96) -- (2.97,1.06);
\draw (2.84,0.95) --  (2.87,0.91) -- (2.82,0.91);

\draw (-3.26,2.2) -- (-3.16,2.2); 
\draw (-3.16,2.35) -- (-3.16,2.05); 
\draw (-3.11,2.35) -- (-3.11,2.05); 
\draw (-3.11,2.25) -- (-3.01,2.3) -- (-3.01,2.45); 
\draw (-3.11,2.15) -- (-3.01,2.1); 
\draw (-3.01,2.1) node (v9) {} -- (-3.01,1.95); 
\draw (-3.01,2) -- (-2.91,2) --  (-2.91,2.15);
\draw  (-2.91,2.25) -- (-2.91,2.4) -- (-3.01,2.4);
\draw (-2.96,2.25) -- (-2.86,2.25); 
\draw (-2.91,2.25) node (v8) {} -- (-2.96,2.15) -- (-2.86,2.15) -- (-2.91,2.25);
\draw (-3.04,2.14) --  (-3.01,2.1) -- (-3.06,2.1);

\draw[fill=backgroundblue!10,color=backgroundblue!10,opacity=0.5]  (-0.2,1.5) ellipse (2 and 1.1);
\draw[-latex,backgroundblue!180] (-2.55,2.3) -- (-2.05,2.3);
\node at (-2.35,2.5) {$\Delta i_{\mathrm{dq},i}$};
\draw[dashed,backgroundblue!200]  (-0.2,1.5) ellipse (2 and 1.1);
\node [color=backgroundblue!200] at (-0.2,2.4) {power network};
\node [roundnode] at (-1.7,2.2) {};
\node [roundnode] at (-2.2,1.4) {};
\node [roundnode] at (0.2,0.4) {};
\node [roundnode] at (1.6,1) {};
\node [roundnode] at (1.6,2) {};
\draw[-latex,backgroundblue!180] (-1.9,2.2) -- (-1.9,2.6);
\node at (-1.9,2.7) {$\Delta v_{\mathrm{dq},i}$};
\fill[backgroundblue!180] (-1.9,2.2) circle(0.2mm);
\end{tikzpicture}

}
    \vspace{-9mm}
    \caption{Illustration of the multi-converter transmission system.}
    \label{fig:multi_conv_system}
    \vspace{-4mm}
\end{figure}

\section{Power System Model}\label{sec:power_system_model}
\subsection{Small-Signal Network Dynamics}\label{sec:network}
We study the stabilization of multiple three-phase GFM voltage source converters (VSCs) connected through a dynamic transmission network modeled by resistive-inductive lines (\cref{fig:multi_conv_system}). We consider a Kron-reduced \cite{dorfler2012kron}, balanced network with $n\in\mathbb{N}$ converter nodes, denoted by $\{1,...,n\}$, where the dynamic \textit{small-signal} model (in the global per unit system) of the line between node $i$ and node $j$ ($i,j\in\{1,...,n\}$) is given in the frequency domain as 
\begin{align}\label{eq:dynamic_small_signal_model_IV}
    \begin{bmatrix}
        \Delta i_{\mathrm{d},ij}(s)\\\Delta i_{\mathrm{q},ij}(s)
    \end{bmatrix}=y_{ij}(s)
    \left(
    \begin{bmatrix}
        \Delta v_{\mathrm{d},i}(s)\\ \Delta v_{\mathrm{q},i}(s)
    \end{bmatrix}-  \begin{bmatrix}
        \Delta v_{\mathrm{d},j}(s)\\ \Delta v_{\mathrm{q},j}(s)
    \end{bmatrix}
    \right),
\end{align}
where $\Delta i_{\mathrm{dq},ij}=\left[\Delta i_{\mathrm{d},ij}\,\,\Delta i_{\mathrm{q},ij} \right]^\top$ is the current vector from node $i$ to node $j$ in the global $\mathrm{dq}$ coordinates with a constant nominal rotating frequency $\omega_0$, $\Delta v_{\mathrm{dq},i}=\left[\Delta v_{\mathrm{d},i}\,\,\Delta v_{\mathrm{q},i} \right]^\top$ is the voltage vector of node $i$, and $y_{ij}(s)$ is a $2\times 2$ transfer matrix encoding the small-signal dynamics of the line, i.e.,
\begin{align}\label{eq:line_dynamics}
    y_{ij}(s) = b_{ij} \underset{\eqqcolon\, f_\rho(s)}{\underbrace{\begin{bmatrix}
        \rho+\tfrac{s}{\omega_0}&1\\-1&\rho+\tfrac{s}{\omega_0}
    \end{bmatrix}\tfrac{1}{1+\left( \rho+\tfrac{s}{\omega_0} \right)^2}}},
\end{align}
where $b_{ij}=\tfrac{1}{l_{ij}}$ is the line susceptance, and $\rho_{ij}=\tfrac{r_{ij}}{l_{ij}}$ is the resistance-inductance ratio of the line $ij$ which is assumed to be \emph{small} and \emph{uniform} (i.e., $\rho_{ij}=\rho\ll 1,\,\forall\,i,j\in\{1,...,n\}$) throughout the dominantly inductive transmission network. If there is no line, we consider $b_{ij}=0$ and $\rho_{ij} =0$. As is standard in small-signal analysis, the network model assumes fixed RL-line parameters and a constant $\omega_0$ in the applied $\mathrm{dq}$ transformation, and is therefore valid only in the small-signal regime around the nominal operating frequency.

The bus current injection of each node $i$ is defined as $\Delta i_{\mathrm{dq},i}(s) \coloneqq \textstyle\sum_{j\neq i}^n\Delta i_{\mathrm{dq},ij}(s)$, based on which we can construct the network dynamics for all nodes in the form of the $2n\times 2n$ Laplacian transfer matrix $Y(s)$, i.e., 
\begin{align}\label{eq:reduced_dynamics_IV}
   \begin{bmatrix}
        \Delta i_{\mathrm{d},1}(s)\\\Delta i_{\mathrm{q},1}(s)\\ \vdots \\ \Delta i_{\mathrm{d},n}(s)\\\Delta i_{\mathrm{q},n}(s)
    \end{bmatrix} = \underset{\eqqcolon\ Y(s)}{\underbrace{B\otimes f_\rho(s)}}
    \begin{bmatrix}
        \Delta v_{\mathrm{d},1}(s)\\\Delta v_{\mathrm{q},1}(s)\\ \vdots \\ \Delta v_{\mathrm{d},n}(s)\\\Delta v_{\mathrm{q},n}(s)
    \end{bmatrix},
\end{align}
where $Y_{ij}(s)=-y_{ij}(s)$ if $i\neq j$, $Y_{ii}(s) = \textstyle\sum_{j\neq i}^n y_{ij}(s)$, and $B\in\mathbb{R}^{n\times n}$ is a Laplacian matrix  
\begin{align}
    B = \begin{bmatrix}
        \textstyle\sum_{j\neq1}^{n}b_{1j} & \dots & -b_{1n} \\ \vdots & \ddots & \vdots \\
        -b_{n1} & \cdots &  \textstyle\sum_{j\neq n}^{n}b_{nj}
    \end{bmatrix},
\end{align}
that encodes the network topology and line susceptances, where $B_{ij} = -b_{ij}$ if $i\neq j$, and $B_{ii}=\textstyle\sum_{j\neq i}^n b_{ij}$. 
 
The time-domain bus voltages are modeled in polar coordinates as $v_{\mathrm{d},i}(t) \coloneqq |v|_i(t)\cos{\delta_i(t)}$ and $v_{\mathrm{q},i}(t) \coloneqq |v|_i(t)\sin{\delta_i(t)} $ with magnitude $|v|_i(t)$ and relative angle $\delta_i(t) = \theta_i(t) - \theta_0(t)$ in SI units, where $\theta_i(t)$ is the voltage angle at bus $i\in\{1,...,n \}$, and $\theta_0(t)$ the Park transformation angle of the global $\mathrm{dq}$ coordinates, which can be expressed as
\begin{align}\label{eq:bus_voltage}
    \begin{split}
        \hspace{-2mm}|v|_i(t) \hspace{-0.7mm}= \hspace{-0.7mm}\sqrt{\hspace{-0.7mm}v_{\mathrm{d},i}(t)^2\hspace{-0.7mm}+\hspace{-0.7mm}v_{\mathrm{q},i}(t)^2},\quad \delta_i(t) \hspace{-0.7mm}= \arctan\hspace{-0.7mm}\left(\tfrac{v_{\mathrm{q},i}(t)}{v_{\mathrm{d},i}(t)}\right)\hspace{-0.5mm}.\hspace{-2mm}
    \end{split}
\end{align}
We consider the active and reactive branch powers 
\begin{align}
\begin{split}
    p_{ij}(t) &= v_{\mathrm{d},i}(t)i_{\mathrm{d},ij}(t)+v_{\mathrm{q},i}(t)i_{\mathrm{q},ij}(t)\\
    q_{ij}(t) &= -v_{\mathrm{d},i}(t)i_{\mathrm{q},ij}(t)+v_{\mathrm{q},i}(t)i_{\mathrm{d},ij}(t),
    \end{split}
\end{align}
and define the associated bus power injections $p_{i}(t)\coloneqq\textstyle\sum_{j\neq i}^np_{ij}(t)$ and $q_{i}(t)\coloneqq\textstyle\sum_{j\neq i}^nq_{ij}(t)$, which are obtained as
\begin{align}\label{eq:branch_power_injections}
\begin{split}
    p_{i}(t) &= v_{\mathrm{d},i}(t)i_{\mathrm{d},i}(t)+v_{\mathrm{q},i}(t)i_{\mathrm{q},i}(t)\\
    q_{i}(t) &= -v_{\mathrm{d},i}(t)i_{\mathrm{q},i}(t)+v_{\mathrm{q},i}(t)i_{\mathrm{d},i}(t).
\end{split}
\end{align}
By linearizing \cref{eq:bus_voltage} and \cref{eq:branch_power_injections}, transforming them into the frequency domain, and performing some analytical computations (see Appendix \ref{appendix1} for details), we obtain
\begin{align}\label{eq:full_network_polar}
   \hspace{-1mm} \underset{\eqqcolon\,{\tiny\begin{bmatrix}
       \Delta p(s)\\\Delta q(s)
   \end{bmatrix}}}{\underbrace{\begin{bmatrix}
       \Delta p_1(s)\\ \Delta q_1(s)\\\vdots\\ \Delta p_n(s)\\ \Delta q_n(s)
   \end{bmatrix}}}
  \hspace{-0.5mm} = \hspace{-0.5mm}\underset{\eqqcolon \, {N}(s)}{\underbrace{
   \begin{bmatrix}
       N_{11}(s)&\cdots&N_{1n}(s)\\\vdots & \ddots& \vdots\\ N_{n1}(s)& \cdots & N_{nn}(s)
   \end{bmatrix}}}
   \underset{\eqqcolon\,{\tiny\begin{bmatrix}
       \Delta \delta(s)\\\Delta |v|_\mathrm{n}(s)
   \end{bmatrix}}}{\underbrace{
   \begin{bmatrix}
       \Delta \delta_1(s)\\\Delta {|v|}_{\mathrm{n},1} (s)\\ \vdots\\ \Delta \delta_n(s)\\\Delta {|v|}_{\mathrm{n},n} (s)
   \end{bmatrix}}}\hspace{-1mm} 
\end{align}
with the $2\times 2$ transfer matrix blocks $N_{ii}(s)$ and $N_{ij}(s)$
\begin{align}\label{eq:full_network_polar_matrix_blocks}
\begin{split}
    N_{ii}(s) &= \textstyle\sum_{j\neq i}^n b_{ij}\tfrac{|v|_{0,i}^2}{1+(\rho+\tfrac{s}{\omega_0})^2}\begin{bmatrix}
        1&\rho+\tfrac{s}{\omega_0}\\-(\rho+\tfrac{s}{\omega_0}) & 1
    \end{bmatrix}\\
    &\quad\quad\quad+\textstyle\sum_{j\neq i}^n b_{ij}\tfrac{|v|_{0,i}^2-|v|_{0,i}|v|_{0,j}}{1+\rho^2}\begin{bmatrix}
        -1& \rho \\ \rho &1
    \end{bmatrix}\\
    N_{ij}(s) &= b_{ij}\tfrac{|v|_{0,i}|v|_{0,j}}{1+(\rho+\tfrac{s}{\omega_0})^2}\begin{bmatrix}
        -1&-(\rho+\tfrac{s}{\omega_0})\\(\rho+\tfrac{s}{\omega_0})&-1
    \end{bmatrix},
\end{split}
\end{align}
where $\Delta {|v|}_{\mathrm{n},i} (s) \coloneqq \tfrac{\Delta |v|_i(s)}{|v|_{0,i}}$ is the voltage magnitude at bus $i$ normalized at the steady-state $|v|_{0,i}$. Moreover, to derive \cref{eq:full_network_polar_matrix_blocks}, as in any power system small-signal model, we have assumed a small steady-state angle difference $\delta_{0,i}\approx \delta_{0,j}$, which is standard and justified since thermal limitations for transmission lines preclude large angle differences \cite{kundur2007power}. 

For a dominantly inductive transmission network with $\rho \ll 1$, the standard approximation of the line dynamics matrix
\begin{align}
\begin{split}
    M(s)&=\tfrac{1}{1+(\rho+\tfrac{s}{\omega_0})^2}\begin{bmatrix}
        1&\rho+\tfrac{s}{\omega_0}\\-(\rho+\tfrac{s}{\omega_0}) & 1
    \end{bmatrix}
\end{split}
\end{align}
appearing in $N_{ii}(s)$ and $N_{ij}(s)$ in \cref{eq:full_network_polar_matrix_blocks} is 
\begin{align}\label{eq:standard_approx}
\begin{split}
    M_1(s)&=\tfrac{1}{1+\tfrac{s^2}{\omega_0^2}}\begin{bmatrix}
        1&\tfrac{s}{\omega_0}\\-\tfrac{s}{\omega_0} & 1
    \end{bmatrix},
\end{split}
\end{align}
i.e., losses are entirely neglected ($\rho =0$). Here we put forward a novel and more accurate approximation for $M(s)$, namely
\begin{align}\label{eq:novel_approx}
     M_2(s)&=\tfrac{1}{1+(\rho+\tfrac{s}{\omega_0})^2}\begin{bmatrix}
        1&\tfrac{s}{\omega_0}\\-\tfrac{s}{\omega_0} & 1
    \end{bmatrix}.
\end{align}
Indeed, a straightforward calculation comparing the Euclidean norms of the matrix distances for $s=\mathrm{j}\omega$ reveals that $\forall \omega \geq 0$
\begin{align}
    \tfrac{||M(\mathrm{j}\omega)-M_2(\mathrm{j}\omega)||_F}{||M(\mathrm{j}\omega)-M_1(\mathrm{j}\omega)||_F} = \sqrt{\tfrac{1+\tfrac{\omega^4}{\omega_0^4}-2\tfrac{\omega^2}{\omega_0^2}}{1+\tfrac{\omega^4}{\omega_0^4}+6\tfrac{\omega^2}{\omega_0^2}+\rho^2\left(1+\tfrac{\omega^2}{\omega_0^2}\right)}}< 1,
\end{align}
that is, $M_2$ is a strictly better approximation for $\rho >0$. For $\rho \rightarrow 0$, it can be shown that both approximations are consistent in terms of recovering $M$ asymptotically. Using the approximation \cref{eq:novel_approx}, we can eliminate the less dominant terms of the off-diagonal matrix elements in \cref{eq:full_network_polar_matrix_blocks} as stated below, which allows to exploit symmetry properties of the network dynamics in the stability proof in \cref{sec:proof_main_thm}. However, since \cref{eq:novel_approx} is valid only for $\rho \ll 1$, extending our results to more general R/X ratios remains a topic for future work.

\begin{definition}[Dynamic Small-Signal Network Model]
The transfer matrix blocks for the \emph{dynamic (i.e., $s\neq0$, $|v|_{0,i}\neq |v|_{0,j}$)} small-signal network model in \cref{eq:full_network_polar} used to derive the stability certificates in \cref{sec:decentralized_stability_conditions} are given by:
\begin{align}\label{eq:full_network_polar_matrix_blocks_level0}
    \begin{split}
    N_{ii}(s) &= \textstyle\sum_{j\neq i}^n b_{ij}\tfrac{|v|_{0,i}^2}{1+(\rho+\tfrac{s}{\omega_0})^2}\begin{bmatrix}
        1&\tfrac{s}{\omega_0}\\-\tfrac{s}{\omega_0} & 1
    \end{bmatrix}\\
    &\quad\quad\quad+\textstyle\sum_{j\neq i}^n b_{ij}\tfrac{|v|_{0,i}^2-|v|_{0,i}|v|_{0,j}}{1+\rho^2}\begin{bmatrix}
        -1& 0\\ 0 &1
    \end{bmatrix}\\
    N_{ij}(s) &= b_{ij}\tfrac{|v|_{0,i}|v|_{0,j}}{1+(\rho+\tfrac{s}{\omega_0})^2}\begin{bmatrix}
        -1&-\tfrac{s}{\omega_0}\\\tfrac{s}{\omega_0}&-1
    \end{bmatrix}.
\end{split}
\end{align}
\end{definition}

By setting $s = 0$ in \cref{eq:full_network_polar_matrix_blocks_level0}, we derive the quasi-stationary network model, which is employed in a similar form in \cite{siahaan2024decentralized}:
\begin{definition}[Network-Simplification Level 1]\label{def:level1}
The transfer matrix blocks for the \emph{quasi-stationary (i.e., $s=0$, $|v|_{0,i}\neq |v|_{0,j}$)} small-signal network model in \cref{eq:full_network_polar} are given by:
\begin{align}\label{eq:full_network_polar_matrix_blocks_level1}
    \begin{split}
    \hspace{-2mm}N_{ii}(s) \hspace{-0.5mm}&=\hspace{-0.5mm} \textstyle\sum_{j\neq i}^n\hspace{-0.5mm} b_{ij}\tfrac{1}{1+\rho^2}\hspace{-0.8mm}\begin{bmatrix}
        |v|_{0,i}|v|_{0,j}\hspace{-2mm}&0\\0 \hspace{-2mm}& 2|v|_{0,i}^2\hspace{-0.5mm}-\hspace{-0.5mm}|v|_{0,i}|v|_{0,j}
    \end{bmatrix}\hspace{-2mm}\\
    \hspace{-2mm}N_{ij}(s)\hspace{-0.5mm} &=\hspace{-0.5mm} b_{ij}\tfrac{|v|_{0,i}|v|_{0,j}}{1+\rho^2}\hspace{-0.5mm}\begin{bmatrix}
        -1&0\\0&-1
    \end{bmatrix}\hspace{-0.5mm}.\hspace{-2mm}
\end{split}
\end{align}
\end{definition}

To further simplify \cref{eq:full_network_polar_matrix_blocks_level1}, we set $|v|_{0,i} = |v|_{0,j} = |v|_0$, yielding the zero-power flow network model similar to \cite{watson2020control}:

\begin{definition}[Network-Simplification Level 2]\label{def:level2}
The transfer matrix blocks for the \emph{zero-power flow (i.e., $s=0$, $|v|_{0,i}\hspace{-0.5mm}=\hspace{-0.5mm}|v|_{0,j}\hspace{-0.5mm}=\hspace{-0.5mm}|v|_0$)} small-signal network model in \cref{eq:full_network_polar} are given by:
\begin{align}\label{eq:full_network_polar_matrix_blocks_level2}
    \begin{split}
    N_{ii}(s)&=\textstyle\sum_{j\neq i}^n b_{ij}\tfrac{|v|_0^2}{1+\rho^2}\begin{bmatrix}
        1&0\\0 & 1
    \end{bmatrix}\\
    N_{ij}(s) &= b_{ij}\tfrac{|v|_{0}^2}{1+\rho^2}\hspace{-0.5mm}\begin{bmatrix}
        -1&0\\0&-1
    \end{bmatrix}\hspace{-0.5mm}.\hspace{-2mm}
\end{split}
\end{align}
\end{definition}

\begin{remark}
For $\rho = 0$, the network model in \cite{siahaan2024decentralized} corresponds to \cref{eq:full_network_polar_matrix_blocks_level1} in \cref{def:level1}, while the model in \cite{watson2020control} aligns with \cref{eq:full_network_polar_matrix_blocks_level2} in \cref{def:level2}. Moreover, other works \cite{gross2022compensating,pates2019robust} focus solely on SISO frequency dynamics. Hence, our dynamic network model in \cref{eq:full_network_polar_matrix_blocks_level0} offers the most detailed representation in literature which is compatible with theoretical stability certificates.
\end{remark}
\begin{figure}[t!]
    \centering
    \vspace{-1mm}
    \resizebox{0.39\textwidth}{!}{

\usetikzlibrary{circuits.ee.IEC}
\tikzstyle{roundnode}=[circle,draw=black!60, fill=black!5,scale=0.5]
\begin{tikzpicture}[circuit ee IEC, scale=1,every node/.style={scale=0.6}]
\draw (-2.76,2.7) -- (-2.66,2.7); 
\draw (-2.66,2.85) -- (-2.66,2.55); 
\draw (-2.61,2.85) -- (-2.61,2.55); 
\draw (-2.61,2.75) -- (-2.51,2.8) -- (-2.51,2.95); 
\draw (-2.61,2.65) -- (-2.51,2.6); 
\draw (-2.51,2.6) node (v9) {} -- (-2.51,2.45); 
\draw (-2.51,2.5) -- (-2.41,2.5) --  (-2.41,2.65);
\draw  (-2.41,2.75) -- (-2.41,2.9) -- (-2.51,2.9);
\draw (-2.46,2.75) -- (-2.36,2.75); 
\draw (-2.41,2.75) node (v8) {} -- (-2.46,2.65) -- (-2.36,2.65) -- (-2.41,2.75);
\draw (-2.54,2.64) --  (-2.51,2.6) -- (-2.56,2.6);

\draw  [rounded corners = 3](-2.9,3) rectangle (-2.2,2.4);
\draw [-latex](-2.55,1.7) -- (-2.55,2.4);
\draw  (-2.9,1.7) rectangle (-2.2,1.3);
\node at (-2.55,1.5) {PWM};

\draw  (-2.9,1) rectangle (-2.2,0.5);
\draw (-2.9,0.5) -- (-2.2,1);
\node at (-2.7,0.85) {abc};
\node at (-2.4,0.65) {dq};
\draw  (-1.2,1) rectangle (0.1,0.5);
\draw[-latex] (-2.55,1) -- (-2.55,1.3);
\draw (-1.2,0.5) -- (0.1,1); 
\node at (-0.9,0.85) {abc};
\node at (-0.2,0.65) {dq};

\draw[-latex] (-0.95,0.5) -- (-0.95,0.2); 
\draw [-latex](-0.55,0.5) -- (-0.55,0.2); 
\draw [-latex](-0.15,0.5) -- (-0.15,0.2);

\draw[-latex] (-1.05,1.3) -- (-1.05,1); 
\draw [-latex](-0.55,1.3) -- (-0.55,1); 
\draw [-latex](-0.05,1.3) -- (-0.05,1);

\draw [-latex](-1.7,0.75) node (v1) {} -- (-2.2,0.75); 

\draw [-latex](-2.3,-0.7) -- (-2.55,-0.7) -- (-2.55,0.5);

\node at (-0.8,1.3) {$i_{\mathrm{abc},i}$};
\node at (-0.3,1.3) {$v_{\mathrm{abc},i}$};
\node at (-1.35,1.3) {$i_{\mathrm{cabc},i}$};

\node at (-1.05,0.12) {$i_{\mathrm{cdq},i}$};
\node at (-0.55,0.12) {$i_{\mathrm{dq},i}$};
\node at (-0.05,0.12) {$v_{\mathrm{dq},i}$};

\draw [thick](-2.3,-0.5) -- (-2.3,-0.9); 

\draw [-latex](-1.8,-0.6) -- (-2.3,-0.6); 
\draw[-latex] (-1.8,-0.8) -- (-2.3,-0.8);
\draw [rounded corners = 3] (-1.8,-0.4) rectangle (-0.8,-1);
\node at (-1.3,-0.6) {current};
\node at (-1.3,-0.8) {control};
\node at (-2.05,-0.45) {$v_{\mathrm{cd},i}^\star$};
\node at (-2.05,-0.95) {$v_{\mathrm{cq},i}^\star$};
\draw [-latex](-1.5,-0.1) -- (-1.5,-0.4); 
\draw [-latex] (-1.1,-0.1) -- (-1.1,-0.4);
\node at (-1.75,-0.2) {$i_{\mathrm{cd},i}$};
\node at (-0.85,-0.2) {$i_{\mathrm{cq},i}$};
\draw [-latex](-1.5,-1.3) -- (-1.5,-1); 
\draw [-latex](-1.1,-1.3) -- (-1.1,-1); 
\node at (-1.7,-1.2) {$v_{\mathrm{d},i}$};
\node at (-0.9,-1.2) {$v_{\mathrm{q},i}$};
\draw [-latex](-0.3,-0.6) -- (-0.8,-0.6); 
\draw [-latex](-0.3,-0.8) -- (-0.8,-0.8);

\draw [-latex](0,-1.3) -- (0,-1); 
\draw [-latex](0.4,-1.3) -- (0.4,-1); 

\node at (-0.2,-1.2) {$i_{\mathrm{d},i}$};
\node at (0.6,-1.2) {$i_{\mathrm{q},i}$};

\draw [rounded corners = 3] (-0.3,-0.4) rectangle (0.7,-1);
\node at (0.2,-0.6) {voltage};
\node at (0.2,-0.8) {control}; 
\draw  [rounded corners=3,fill=backgroundblue!30](1.2,-0.4) rectangle (2.3,-1);
\node at (1.75,-0.6) {$\mathrm{q}$-$\mathrm{v}$ droop};
\node at (1.75,-0.8) {control};
\draw[-latex] (1.2,-0.6) -- (0.7,-0.6);
\draw [-latex](1.2,-0.8) -- (0.7,-0.8);
\node at (-0.5,-0.45) {$i_{\mathrm{cd},i}^\star$};
\node at (-0.5,-0.95) {$i_{\mathrm{cq},i}^\star$};
\node at (0.95,-0.45) {$v_{\mathrm{d},i}^\star$};
\node at (0.95,-0.95) {$v_{\mathrm{q},i}^\star$};
\draw [-latex](0,-0.1) -- (0,-0.4); 
\draw [-latex](0.4,-0.1) -- (0.4,-0.4);
\node at (-0.2,-0.2) {$v_{\mathrm{d},i}$};
\node at (0.6,-0.2) {$v_{\mathrm{q},i}$};
\draw [-latex](2.6,-0.7) -- (2.3,-0.7);

\node at (2.5,-0.55) {$q_i$};
\draw  (-1.2,1.9) rectangle (-0.8,1.5);
\node at (-1,1.7) {$\tfrac{1}{s}$};
\draw (-1.2,1.7) -- (-1.7,1.7) -- (-1.7,0.75) node (v2) {};
\draw [-latex] (-1.7,0.75) -- (-1.2,0.75);
\draw  [rounded corners = 3,fill=backgroundblue!30](-0.4,2) rectangle (0.7,1.4);
\draw[-latex] (-0.4,1.7) -- (-0.8,1.7);
\fill (-1.7,0.75) circle(0.2mm);
\node at (0.15,1.8) {$\mathrm{p}$-$\mathrm{f}$ droop};
\node at (0.15,1.6) {control};
\draw[-latex] (1,1.7) -- (0.7,1.7);
\node at (-1.4,1.85) {$\theta_i$};
\node at (-0.6,1.85) {$\omega_i$};
\node at (0.9,1.85) {$p_i$};
\draw [backgroundblue!200,rounded corners = 3,dashed] (-3.1,2.1) rectangle (2.8,-1.4);
\node [backgroundblue!200] at (2.1,1.9) {GFM control};
\draw  [rounded corners = 3,dashed](-3.1,2.2) rectangle (2.8,3.1);
\node at (2.2,2.9) {power part};
\draw (-2.2,2.7) to [inductor={yscale=1.2,xscale=0.8}](-0.4,2.7) node (v3) {};
\node at (-1.85,2.85) {$v_{\mathrm{cabc},i}$};
\fill (-1.85,2.7) circle(0.2mm);
\node at (-1.25,2.925) {$l_{\mathrm{f},i}$};
\node at (-0.55,2.85) {$v_{\mathrm{abc},i}$};

\draw [-latex,backgroundblue!200](-1.55,2.6) -- (-1,2.6);
\node at (-1.3,2.4) {$i_{\mathrm{cabc},i}$};
\draw[dashed]  (0.5,2.7) -- (1.65,2.7);
\node at (0.5,2.5) {node $i$}; 

\draw [-latex,backgroundblue!200](-0.15,2.8) -- (0.4,2.8);
\node at (0.05,2.93) {$i_{\mathrm{abc},i}$};

\draw (-0.45,2.7) -- (0.55,2.7);
\fill (-0.55,2.7) circle(0.2mm);
\draw (-0.55,2.7) -- (-0.55,2.52); 
\draw (-0.55,2.45) -- (-0.55,2.3); 
\draw (-0.6,2.3) -- (-0.5,2.3); 
\draw (-0.65,2.52) -- (-0.45,2.52); 
\draw (-0.65,2.45) -- (-0.45,2.45); 
\node at (-0.25,2.5) {$c_{\mathrm{f},i}$};
\node [roundnode] at (0.5,2.7) {};
\end{tikzpicture}

}
    \vspace{-9mm}
    \caption{Basic implementation of a three-phase GFM VSC.}
    \label{fig:converter_model}
    \vspace{-4mm}
\end{figure}

\subsection{Grid-Forming Converter Dynamics}
\cref{fig:converter_model} shows the implementation of a GFM three-phase VSC\footnote{A simple VSC model is used for illustration; the results extend to other architectures under the time-scale separation in \cref{rem:timescale}.} connected to the power network in \cref{fig:multi_conv_system}. The linearized small-signal dynamics (in the global per unit system) of the $i$-th VSC ($i\in \{1,...,n\}$) can be represented by the $2\times 2$ transfer matrix $D_i(s)$, mapping from active and reactive power injections $\Delta p_i(s)$ and $\Delta q_i(s)$ to the imposed frequency and voltage magnitude $\Delta \omega_i(s)=\Delta \delta_i(s)s$ and $\Delta |v|_i(s)$, i.e., 
\begin{align}\label{eq:local_conv_model}
    -\begin{bmatrix}
        \Delta \omega_i(s)\\ \Delta |v|_i(s)
    \end{bmatrix}= D_i(s)\begin{bmatrix}
        \Delta p_i(s)\\ \Delta q_i(s)
    \end{bmatrix}.
\end{align}

The most prevalent GFM control strategy is filtered droop control \cite{rocabert2012control}. Assuming time-scale separation as elaborated in Remark \ref{rem:timescale}, we can neglect the internal dynamics of the VSC (see Appendix \ref{appendix2} for a derivation) and consider $D_i(s)$ as in the following definition. Of course, all our developments will be evaluated on the full VSC model in \cref{sec:numerical_case_studies}.

\begin{definition}[Dynamic Small-Signal Converter Model]
The small-signal converter dynamics \cref{eq:local_conv_model} used to derive the stability certificates in \cref{sec:decentralized_stability_conditions} are given by the droop controller
    \begin{align}\label{eq:converter_model_final}
        D_i(s) = \begin{bmatrix}
            \tfrac{d_{\mathrm{p},i}}{\tau_{\mathrm{p},i}s+1}&0\\0&\tfrac{d_{\mathrm{q},i}}{\tau_{\mathrm{q},i}s+1}
        \end{bmatrix},
    \end{align}
    where $d_{\mathrm{p},i}\in\mathbb{R}_{>0}$ and $d_{\mathrm{q},i}\in\mathbb{R}_{>0}$ are the active and reactive power droop gains, and $\tau_{\mathrm{p},i}\in\mathbb{R}_{\geq0}$ and $\tau_{\mathrm{q},i}\in\mathbb{R}_{\geq0}$ are the active and reactive power low-pass filter time constants.
\end{definition}

Finally, if we extend \cref{eq:local_conv_model} to include the dynamics of all $n$ converters, we obtain the following $2n\times 2n$ transfer matrix
\begin{align}\label{eq:full_converter_polar}
     \hspace{-1.8mm}-\hspace{-0.3mm}\underset{\eqqcolon\,{\tiny\begin{bmatrix}
       \Delta \omega(s)\\\Delta |v|(s)
   \end{bmatrix}}}{\underbrace{
   \begin{bmatrix}
       \Delta \omega_1(s)\\\Delta {|v|}_{1} (s)\\ \vdots\\ \Delta \omega_n(s)\\\Delta {|v|}_{n} (s)
   \end{bmatrix}}} \hspace{-1.4mm}= \hspace{-1.4mm} \underset{\eqqcolon\,D(s)}{\underbrace{\begin{bmatrix}
       D_1(s) \hspace{-0.8mm}& \hspace{-0.8mm}0_{2\times 2} \hspace{-0.8mm}& \hspace{-0.8mm} \dots  \hspace{-0.8mm}& \hspace{-0.8mm} 0_{2\times 2}\\0_{2\times2}  \hspace{-0.8mm}&  \hspace{-0.8mm}D_2(s) \hspace{-0.8mm} &  \hspace{-0.8mm}\dots  \hspace{-0.8mm}&  \hspace{-0.8mm}0_{2\times 2}\\
       \vdots  \hspace{-0.8mm}&  \hspace{-0.8mm}\vdots
        \hspace{-0.8mm}& \hspace{-0.8mm}\ddots \hspace{-0.8mm}& \hspace{-0.8mm}\vdots\\0_{2\times 2} \hspace{-0.8mm}& \hspace{-0.8mm}0_{2\times 2} \hspace{-0.8mm}& \hspace{-0.8mm}\dots \hspace{-0.8mm}& \hspace{-0.8mm}D_n(s)
   \end{bmatrix}}}
    \underset{\eqqcolon\,{\tiny\begin{bmatrix}
       \Delta p(s)\\\Delta q(s)
   \end{bmatrix}}}{\underbrace{\begin{bmatrix}
       \Delta p_1(s)\\ \Delta q_1(s)\\\vdots\\ \Delta p_n(s)\\ \Delta q_n(s)
   \end{bmatrix}}} \hspace{-0.6mm}. \hspace{-1.5mm}
\end{align}

\begin{remark}\label{rem:timescale}
From a power electronics perspective, three distinct dynamic time scales are typically considered (see \cref{fig:timescales}):
\begin{itemize}
\item[(i)] fast inner VSC control loops,
\item[(ii)] intermediate network dynamics, and
\item[(iii)] slow outer VSC control loops.
\end{itemize}
These layers must be sufficiently separated in this order; otherwise, standard assumptions for control design and model reduction break down. Consistent with \cite{gross2022compensating,subotic2021lyapunov}, in this work, we assume that inner-loop and network dynamics are well separated, while a possible overlap between network and outer-loop dynamics is treated as a separate problem.
\end{remark}

\begin{remark}
In the considered small-signal regime, load sensitivities are neglected, and the Kron-reduced network is approximated solely by the RL transmission-network dynamics. Incorporating load sensitivities is left for future work.
\end{remark}

\begin{figure}[t!]
\vspace{-1mm}
    \centering
   \resizebox{0.43\textwidth}{!}{
\begin{tikzpicture}[scale=1, every node/.style={scale=0.5}]
\draw[fill=gray!10,color=gray!10] (0.7,1.65) rectangle (-4.3,2.6) node (v2) {};
\draw [fill=blue!10,color=backgroundblue!10](-4.3,2.6) node (v3) {} rectangle (0.7,3.75);
\draw[-latex] (-4.3,1.65) -- (0.7,1.65) node (v1) {};

\draw (-3.9,1.55) -- (-3.9,1.75);
\node at (-3.9,1.35) {1ms};
\draw (-2.2,1.55) -- (-2.2,1.75);
\draw(-0.5,1.55) -- (-0.5,1.75);
\node[rotate=90] at (-4.15,2.15) {physics};
\node [rotate=90] at (-4.15,3.1) {control};
\draw  [rounded corners=2, color=backgroundblue!20,fill=backgroundblue!20](-3.8,3.4) rectangle (-2.3,2.7);
\node at (-3.05,3.25) {inner converter};
\node at (-3.05,3.05) {control (current};
\node at (-3.05,2.85) {+ voltage loop)};
\draw [rounded corners=2, color=backgroundblue!20,fill=backgroundblue!20] (-1.5,3.4) rectangle (0.5,2.7);
\node at (-0.5,3.25) {outer converter};
\node at (-0.5,3.05) {control (active \&};
\node at (-0.5,2.85) {reactive power loops};
\draw [rounded corners=2, color=gray!20,fill=gray!20] (-2.1,2.5) rectangle (-0.7,1.9);
\node at (-1.4,2.3) {network};
\node at (-1.4,2.1) {dynamics};
\draw [dashed,rounded corners = 2] (-2.2,3.5) rectangle (0.6,1.8);
\node at (-0.8,3.65) {timescales of interest};
\draw[ultra thin] (-4.3,2.6) -- (0.7,2.6);
\node at (-2.2,1.35) {10ms};
\node at (-0.5,1.35) {100ms};
\end{tikzpicture}
}
    \vspace{-9mm}
    \caption{Time scales of interest for our stability study, adapted from \cite{subotic2021lyapunov}.}
        \label{fig:timescales}
            \vspace{-4mm}
\end{figure}

\subsection{Dynamic Power System Model}
The closed-loop power system dynamics are modeled as the feedback interconnection of the converter device dynamics in \cref{eq:full_converter_polar} and the network dynamics in \cref{eq:full_network_polar} as illustrated in \cref{fig:power_sys_model_orig}\footnote{An extension to heterogeneous device allocations with synchronous generators, grid-forming and -following converters will be part of future work.}. We focus on the stability of the bus frequency and voltage magnitudes, given by $\Delta \omega_i(s)=\Delta \delta_i(s)s$ and $\Delta |v|_i(s)$ for $i\in\{1,...,n\}$ and consider these quantities as interconnection signals between the device and network dynamics in \cref{fig:power_sys_model_orig}. Accordingly, the interconnected subsystems are defined as:
\begin{align}\label{eq:N_0_D_0}
\begin{split}
    \mathcal{D}_0(s) &\coloneqq D(s)\\
    \mathcal{N}_0(s) &\coloneqq N(s)\cdot\mathrm{diag}(\tfrac{1}{s},\tfrac{1}{|v|_{0,1}},\dots,\tfrac{1}{s},\tfrac{1}{|v|_{0,n}}).
\end{split}
\end{align}
In the next section, we establish internal feedback stability of $\mathcal{D}_0\#\mathcal{N}_0$ under certain decentralized parametric conditions. 

\begin{remark}
The proposed conditions apply in the small-signal regime around the nominal operating point and are therefore only valid for sufficiently small perturbations.
\end{remark}
\begin{figure}[t!]
\vspace{-1mm}
    \centering
   \resizebox{0.3\textwidth}{!}{
\tikzstyle{roundnode}=[circle,draw=black!60,fill=black!5,scale=0.7]
\begin{tikzpicture}[scale=1,every node/.style={scale=0.8}]
\draw [rounded corners = 3,fill=backgroundblue!30] (-0.3,4.3) rectangle (1.3,3.5);
\node [scale=1.2] at (0.5,3.9) {$\mathcal{D}_0(s)$};
\draw [rounded corners = 3,fill=backgroundblue!30] (-0.3,2.7) rectangle (1.3,1.9);
\node [scale=1.2] at (0.5,2.3) {$\mathcal{N}_0(s)$};
\draw [-latex](1.3,3.9) -- (2.5,3.9);
\draw (-0.3,2.3) -- (-1.1,2.3);

\draw[-latex] (1.8,2.3) -- (1.3,2.3);
\draw [-latex](2.6,2.3) -- (2,2.3); 
\draw [-latex](1.9,3.9) -- (1.9,2.4);
\node [roundnode] at (1.9,2.3) {};

\draw [-latex](-1,3.9) -- (-0.3,3.9);
\draw [-latex](-1.8,3.9) -- (-1.2,3.9); 
\draw [-latex](-1.1,2.3) node (v1) {} -- (-1.1,3.8);
\node [roundnode] at (-1.1,3.9) {};
\node at (-0.9,3.6) {$-$};
\node at (-1.8,4.35) {$\begin{bmatrix}\Delta p_\mathrm{d}\\ \Delta q_\mathrm{d}\end{bmatrix}$};
\node at (2.5,4.35) {$\begin{bmatrix} \Delta \omega\\\Delta |v|\end{bmatrix}$};
\node at (-0.65,4.35) {$\begin{bmatrix}\Delta p\\ \Delta q\end{bmatrix}$};
\node at (2.6,2.75) {$\begin{bmatrix} \Delta\omega_\mathrm{d} \\\Delta |v|_\mathrm{d}\end{bmatrix}$};
\node at (0.5,3.3) {device dynamics};
\node at (0.5,1.7) {network dynamics};
\node at (-1.05,4.35) {$-$};
\draw[-latex] (-1.1,2.3) -- (-1.8,2.3);
\node at (-1.8,2.75) {$\begin{bmatrix}\Delta p_\mathrm{e}\\ \Delta q_\mathrm{e}\end{bmatrix}$};
\end{tikzpicture}
}
    \vspace{-9mm}
    \caption{Closed-loop power system dynamics where $\Delta p_\mathrm{d}$ and  $\Delta q_\mathrm{d}$ are the active and reactive power disturbances, $\Delta \omega_\mathrm{d}$ is the frequency disturbance, and $\Delta |{v}|_{\mathrm{d}}$ is the voltage magnitude disturbance.}
    \vspace{-4mm}
        \label{fig:power_sys_model_orig}
\end{figure}

\renewcommand{\arraystretch}{1}
\begin{table*}[t!]\scriptsize
    \centering
    \vspace{-2mm}
    \begin{tabular}{c|c||c|c}
    \toprule
         $\hspace{-1mm}c_{i,\rho}\hspace{-1mm}$& Definition & $\hspace{-1mm}c_{i,\rho}\hspace{-1mm}$ & Definition\\ \hline
         $\hspace{-1mm}c_{1,\rho}\hspace{-1mm}$& $({1+\rho^2})/({5+2\rho^2})$&$c_{6,\rho}\hspace{-1mm}$& ${(1+\rho^2)}/({2\rho^2|v|_{\mathrm{max}}^2+5|v|_{\mathrm{max}}^2)}$ \\
         $\hspace{-1mm}c_{2,\rho}\hspace{-1mm}$& ${(1+\rho^2)^2}/({6\rho})$&$c_{7,\rho}$&${(\rho^2+1)^2}/({6\rho|v|_\mathrm{max}^2})$ \\
        $\hspace{-1mm}c_{3,\rho}\hspace{-1mm}$& $2\rho(1+\rho^2)$&$c_{8,\rho}$&$-2\rho|v|_\mathrm{max}^4$\\
        $\hspace{-1mm}c_{4,\rho}\hspace{-1mm}$& $((2\rho^2-5)+\sqrt{(5-2\rho^2)^2+16\rho^2})/4\rho$&$c_{9,\rho}$&$|v|_\mathrm{max}^2(2\rho^2-5)$\\$\hspace{-1mm}c_{5,\rho}\hspace{-1mm}$&${5(1+\rho^2)}/{4}$&\\
        \bottomrule
    \end{tabular}
    \caption{Definition of the quantities $c_{l,\rho}$ for $l=1,...,9$ in the conditions \cref{eq:p_conditions,eq:q_conditions} depending on the parameter $\rho$.}
    \vspace{-4mm}
    \label{tab:rho_parameters}
\end{table*}
\renewcommand{\arraystretch}{1}

\section{Decentralized Stability Conditions}\label{sec:decentralized_stability_conditions} In this section, we first introduce the decentralized stability conditions, which provide practical guidelines for local converter tuning and decentralized stability assessment. We then present our main result, which guarantees internal feedback stability of the closed-loop system, assuming the previously stated stability conditions are satisfied. Notably, these conditions naturally emerge from the proof of the main result.

\subsection{Decentralized Stability Conditions}
By applying dynamic loop-shifting techniques and passivity theory (see \cref{sec:proof_main_thm} for details), we derive parametric \emph{decentralized stability conditions} that serve as \emph{local tuning rules} for the dynamic droop control \cref{eq:converter_model_final} of each VSC $i \in \{1, \dots, n\}$, ensuring the internal feedback stability of $\mathcal{D}_0\#\mathcal{N}_0$ in \cref{fig:power_sys_model_orig}. The stability conditions for each VSC $i$ depend on
\begin{itemize}
    \item its own tunable local droop control parameters, namely, the droop gains $d_{\mathrm{p},i},\, d_{\mathrm{q},i}$ and the time constants $\tau_{\mathrm{p},i},\, \tau_{\mathrm{q},i}$,
    \item certain network parameters, including the susceptances of the neighboring lines $b_{ij}$, the resistance-to-inductance ratio $\rho$, and the maximum steady-state bus voltage magnitude $|v|_{\mathrm{max}}$. These parameters are either locally accessible or globally agreed upon, e.g., in grid codes.
\end{itemize}
Crucially, the conditions are local and entirely independent of the control parameters of other VSCs, making them directly applicable for scalable and decentralized stability assessment, as well as for local controller tuning and grid-code design. Beyond that, the decentralized stability certificates support \emph{heterogeneous} device-level controllers by allowing each converter to be designed with independent droop gains and time constants, without requiring uniform parameters across the system, as is the case in, e.g., \cite{gross2022compensating}.

The decentralized stability conditions for each VSC $i$ can be categorized into decoupled constraints on the active power-frequency droop control and the reactive power-voltage droop control. For each control scheme, these conditions can be visualized in either a two-dimensional (for fixed  $\rho$) or a three-dimensional (for variable $\rho$) coordinate system, as shown in \cref{fig:p_conditions_3D,fig:q_conditions_3D}. The coordinate axes represent scaled versions of the local droop gains and time constants (in 2D) and the global network parameter $\rho$ (in 3D), with typical value ranges. The resulting feasible parameter sets are indicated by green dots. As illustrated in \cref{fig:p_conditions_3D}, for small values of $\rho$, closed-loop stability is ensured when the droop gain $d_{\mathrm{p},i}$ and/or the self-susceptance $\textstyle\sum_{j\ne i}^n b_{ij}$ are sufficiently small, irrespective of the time constant $\tau_{\mathrm{p},i}$. Notably, this suggests potential instability issues in future transmission systems with a \emph{high} density of devices interconnected by \emph{short lines}, which correspond to large values of $\textstyle\sum_{j\ne i}^n b_{ij}$. Further, interpreting the filter time constant $\tau_{\mathrm{p},i}$ as virtual inertia reveals that increasing virtual inertia (i.e., larger $\tau_{\mathrm{p},i}$) does not necessarily enhance system stability. These findings align with the small-signal stability conditions derived for active power droop control in conjunction with SISO frequency dynamics in \cite{gross2022compensating}. 
 
Similarly, \cref{fig:q_conditions_3D} shows that closed-loop stability is guaranteed when the droop gain $d_{\mathrm{q},i}$ and/or the self-susceptance $\textstyle\sum_{j\ne i}^n b_{ij}$ are sufficiently small, provided that $\tau_{\mathrm{q},i}$ is nonzero. In particular, in contrast to $\tau_{\mathrm{p},i}$, a larger $\tau_{\mathrm{q},i}$ can be stabilizing. 

Finally, from the 3D plots in \cref{fig:p_conditions_3D,fig:q_conditions_3D}, it becomes apparent how an increasing $\rho$ allows for larger local droop gains of both the active and reactive power droop control.  

\begin{figure}[t!]
    \centering
\begin{subfigure}{0.2\textwidth}
    \centering
    \vspace{-1mm}
\scalebox{0.45}{\includegraphics[]{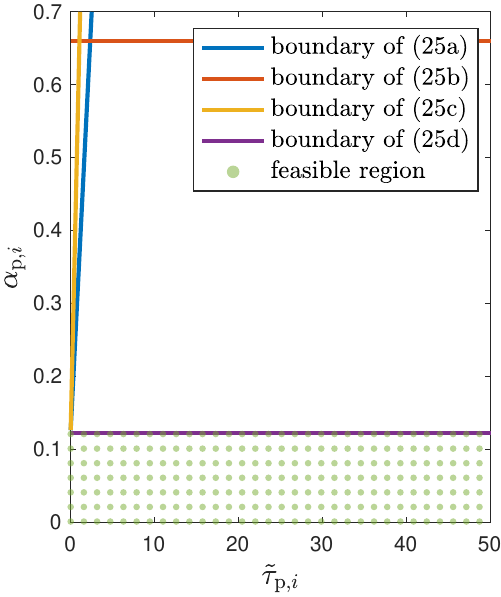}}
\vspace{-5mm}
    \caption{2D plot with $\rho = 0.3$.}
    \vspace{-1mm}
    \label{fig:2d_p}
\end{subfigure}
\hspace{0.2cm}
\begin{subfigure}{0.26\textwidth}
    \centering
    \vspace{-1mm}
\scalebox{0.45}{\includegraphics[]{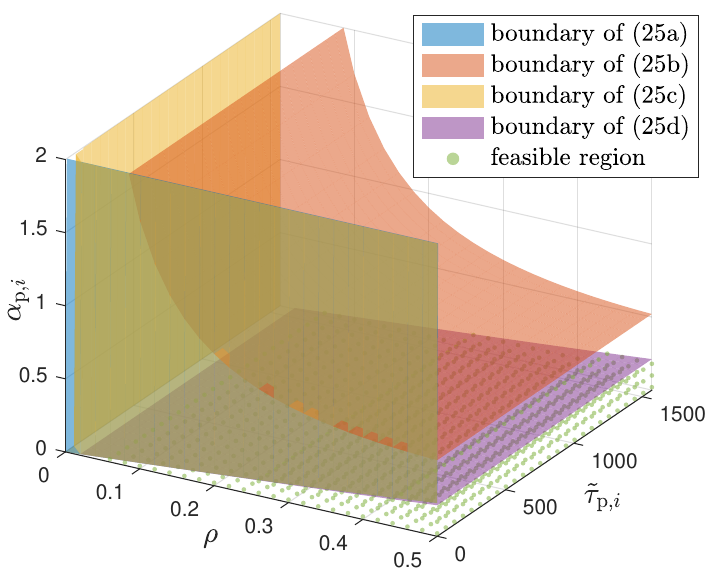}}
\vspace{-5mm}
    \caption{3D plot with variable $\rho$.}
    \vspace{-1mm}
    \label{fig:3d_p}
\end{subfigure}
\caption{Graphical illustration of the stability conditions in \cref{eq:p_conditions}, where $\alpha_{\mathrm{p},i}\hspace{-0.5mm}=\hspace{-0.5mm}\tfrac{d_{\mathrm{p},i}}{\omega_0}|v|_{\mathrm{max}}^2\sum_{j\ne i}^nb_{ij}$ and  $\tilde{\tau}_{\mathrm{p},i}\hspace{-0.5mm}=\hspace{-0.5mm}\tau_{\mathrm{p},i}\omega_0$ (in rad).}
\label{fig:p_conditions_3D}
\vspace{-2mm}
\end{figure}

\begin{figure}[t!]
    \centering
\begin{subfigure}{0.2\textwidth}
    \centering
\scalebox{0.43}{\includegraphics[]{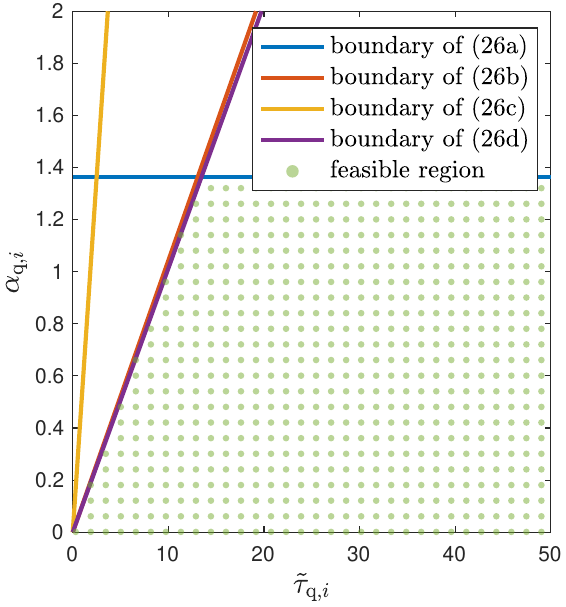}}
\vspace{-5mm}
    \caption{2D plot with $\rho = 0.3$.}
    \vspace{-1mm}
    \label{fig:2d_q}
\end{subfigure}
\hspace{0.25cm}
\begin{subfigure}{0.26\textwidth}
    \centering
\scalebox{0.44}{\includegraphics[]{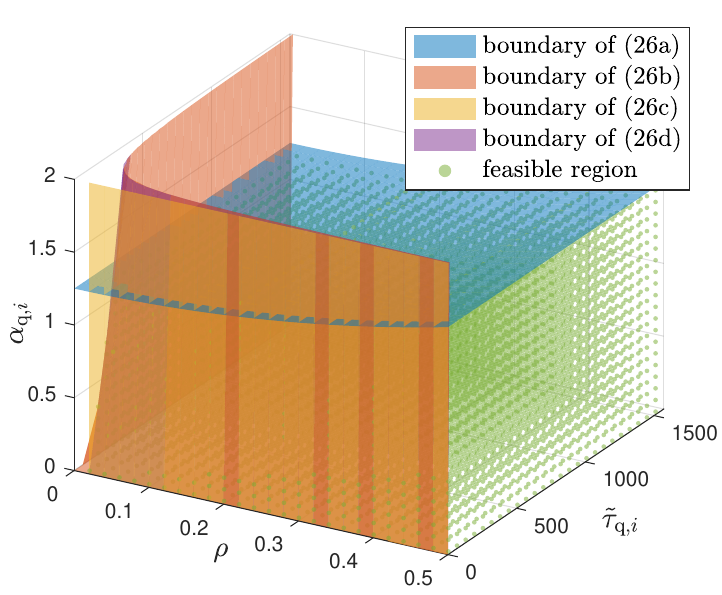}}
\vspace{-5mm}
    \caption{3D plot with variable $\rho$.}
    \vspace{-1mm}
    \label{fig:3d_q}
\end{subfigure}
\caption{Graphical illustration of the stability conditions in \cref{eq:q_conditions}, where $\alpha_{\mathrm{q},i}=\tfrac{d_{\mathrm{q},i}}{|v|_{0,i}}\sum_{j\ne i}^nb_{ij}$ and $\tilde{\tau}_{\mathrm{q},i}=\tau_{\mathrm{q},i}\omega_0$ (in rad).}
\label{fig:q_conditions_3D}
\vspace{-4mm}
\end{figure}

In what follows, we provide an algebraic parametrization of the conditions in \cref{fig:p_conditions_3D,fig:q_conditions_3D} (a detailed derivation is provided in \cref{sec:proof_main_thm}). Namely, for the active power-frequency droop control, we require
    \begin{subequations}\label{eq:p_conditions}
    \begin{align}
    \begin{split}\label{eq:condition1_p}
        \alpha_{\mathrm{p},i} &< c_{1\rho}(2\rho+\tilde{\tau}_{\mathrm{p},i}(1+\rho^2))
    \end{split}\\
     \begin{split}\label{eq:condition2_p}
        \alpha_{\mathrm{p},i} &< c_{2\rho}
         \end{split}\\
          \begin{split}\label{eq:condition3_p}
        \alpha_{\mathrm{p},i} &< c_{3\rho} \tfrac{\tilde{\tau}_{\mathrm{p},i}(\tilde{\tau}_{\mathrm{p},i}(1+\rho^2)+2\rho)+1}{4\tilde{\tau}_{\mathrm{p},i}\rho(1+\rho^2)+2\rho^2+5}
         \end{split}\\
          \begin{split}\label{eq:condition4_p}
        \alpha_{\mathrm{p},i} &< c_{4\rho},
        \end{split}
        \end{align}
    \end{subequations}
    where $\alpha_{\mathrm{p},i}\hspace{-0.5mm}=\hspace{-0.5mm}\tfrac{d_{\mathrm{p},i}}{\omega_0}|v|_{\mathrm{max}}^2\sum_{j\ne i}^nb_{ij}$, $\tilde{\tau}_{\mathrm{p},i}\hspace{-0.5mm}=\hspace{-0.5mm}\tau_{\mathrm{p},i}\omega_0$, and $c_{l\rho}$ for $l\hspace{-0.5mm}=\hspace{-0.5mm}\{1,2,3,4\}$ are quantities depending on $\rho$ as defined in \cref{tab:rho_parameters}. For the reactive power-voltage droop control, we require
\begin{subequations}\label{eq:q_conditions}
        \begin{align}
         \begin{split} \label{eq:c5_rho}
\alpha_{\mathrm{q},i} &< c_{5,\rho}
        \end{split}\\
\begin{split}\alpha_{\mathrm{q},i} &< c_{6,\rho}2\rho\tilde{\tau}_{\mathrm{q},i}
\end{split}\\
\begin{split}\alpha_{\mathrm{q},i}&< c_{7,\rho} \tilde{\tau}_{\mathrm{q},i}
        \end{split}\\
\begin{split}0&<\alpha_{\mathrm{q},i}^2c_{8,\rho}+\alpha_{\mathrm{q},i}\tilde{\tau}_{\mathrm{q},i}c_{9,\rho}+2\tilde{\tau}_{\mathrm{q},i}^2\rho
\end{split}
        \end{align}
    \end{subequations}
    where $\alpha_{\mathrm{q},i}=\tfrac{d_{\mathrm{q},i}}{|v|_{0,i}}\sum_{j\ne i}^nb_{ij}$, $\tilde{\tau}_{\mathrm{q},i}=\tau_{\mathrm{q},i}\omega_0$, and $c_{l\rho}$ for $l=\{5,...,14\}$ are quantities depending on $\rho$ as defined in \cref{tab:rho_parameters}.

    \begin{remark}[Plug-and-Play Operation] The conditions in \cref{eq:p_conditions,eq:q_conditions} are designed for the scenario where all $n$ converter nodes are populated, ensuring system-level stability during disconnection and re-connection of VSCs. Specifically, removing any VSC decreases the self susceptance $\textstyle\sum_{j\ne i}^n b_{ij}$ for each connected VSC $i$ after Kron reduction and reindexing, i.e., guaranteeing that \cref{eq:p_conditions,eq:q_conditions} remain locally satisfied. However, the guarantees hold only when the number of connected converters is at most $n$.
    \end{remark}

\subsection{Main Result}
The conditions \cref{eq:p_conditions} and \cref{eq:q_conditions} are derived in \cref{sec:proof_main_thm}, which proves our main result stated in the theorem below:
\begin{theorem}[Internal Feedback Stability of the Closed-Loop System]\label{thm:main_theorem}
Consider the device and network dynamics $\mathcal{D}_0(s)$ and $\mathcal{N}_0(s)$ in \cref{eq:N_0_D_0}, where the network model is given by \cref{eq:full_network_polar_matrix_blocks_level0}. Let the conditions in \cref{eq:p_conditions} and \cref{eq:q_conditions} hold for each VSC $i\in\{1,...,n\}$. Then, the closed-loop system $\mathcal{D}_0\#\mathcal{N}_0$ in \cref{fig:power_sys_model_orig} is internally feedback stable, i.e., $\mathcal{D}_0\#\mathcal{N}_0\in\mathcal{RH}_\infty^{4n\times 4n}$.
\end{theorem}
\begin{corollary}[Simplified Network Models]\label{corr:special_cases}
For the quasi-stationary network dynamics $\mathcal{N}_0(s)$ with the network model in \cref{eq:full_network_polar_matrix_blocks_level1}, where $s = 0$ and $|v|_{0,i} \neq |v|_{0,j}$, the stability result of \cref{thm:main_theorem} holds under the relaxed algebraic conditions  
\begin{align}\label{eq:level1_conditions}
\begin{split}
    \alpha_{\mathrm{q},i} < \tfrac{5(1+\rho^2)}{4} = c_{5,\rho}\quad\text{and}\quad 0<\tau_{\mathrm{p},i}
\end{split}
\end{align}
for each VSC $i\in\{1,...,n\}$. Further, for the zero-power-flow network dynamics $\mathcal{N}_0(s)$ with the network model in \cref{eq:full_network_polar_matrix_blocks_level2}, where $s = 0$ and $|v|_{0,i} = |v|_{0,j}$, internal feedback stability is guaranteed without imposing any additional VSC conditions.
\end{corollary}
The proof of \cref{corr:special_cases} is provided in Appendix \ref{appendix3} and follows a similar reasoning as the proof of \cref{thm:main_theorem} in \cref{sec:proof_main_thm}. It becomes apparent how the conditions in \cref{eq:p_conditions,eq:q_conditions} are a subset of the conditions in \cref{eq:level1_conditions}. The latter structurally align with the conditions derived in \cite{siahaan2024decentralized}.

\subsection{Proof of \cref{thm:main_theorem}}\label{sec:proof_main_thm}
A structural overview of the proof of \cref{thm:main_theorem} is presented in the flowchart in \cref{fig:proof_structure} and consists of four main steps:

\subsubsection*{I. Coordinate Transformation}
To preserve the symmetry of the network, we use the angle and normalized voltage derivatives, defined as $\Delta \omega_i(s)\coloneqq\Delta\delta_i(s)s$ and $\Delta|\tilde{v}|_{\mathrm{n},i}(s)\coloneqq\Delta |v|_{\mathrm{n},i}(s)s$, as interconnection signals. We then analyze the stability of the closed-loop system as in \cref{fig:power_sys_model}, where
\begin{align}\label{eq:N_D}
\begin{split}
    \mathcal{D}(s)&\coloneqq D(s)\cdot\mathrm{diag}(1,\tfrac{s}{|v|_{0,1}},\dots,1,\tfrac{s}{|v|_{0,n}})\\
    \mathcal{N}(s)&\coloneqq N(s)\cdot\mathrm{diag}(\tfrac{1}{s},\tfrac{1}{s},\dots,\tfrac{1}{s},\tfrac{1}{s}).
    \end{split}
\end{align}

\subsubsection*{II. Loop Shifting}
Given that $\mathcal{D}(s)$ and $\mathcal{N}(s)$ in \cref{fig:power_sys_model} do not satisfy the passivity conditions in \cref{thm:stabilty_passivity} (i.e., it can be shown that $\mathcal{N}(s)$ is not passive \cite{dey2022passivity,dey2022passivity2}), we resort to dynamic loop-shifting techniques as presented in \cref{sec:small_phase_passivity}. More specifically, we consider the loop-shifted system $\mathcal{D}'\#\mathcal{N}'$ in \cref{fig:loop_shifted_ND} with the subsystems $\mathcal{D}'(s)$ and $\mathcal{N}'(s)$ given as 
  \begin{align}\label{eq:N_dash_D_dash}
  \begin{split}
      \mathcal{D}'(s)&\coloneqq \mathcal{D}(s)(I-\Gamma(s)\mathcal{D}(s))^{-1}\\
      \mathcal{N}'(s)&\coloneqq \mathcal{N}(s)+\Gamma(s).
      \end{split}
  \end{align}
Note again that loop-shifting with $\Gamma(s)$ is a mathematical tool for establishing passivity properties and does not involve any hardware or impedance modifications of the actual network or converter system. Having said that, we choose $\Gamma(s)=\mathrm{diag}(\Gamma_1(s),...,\Gamma_n(s))$ as a block-diagonal semi-stable transfer matrix with the $2\times 2$ matrix elements 
\begin{align}\label{eq:Gamma_2x2}
    \Gamma_i(s)=\begin{bmatrix}
        \Gamma_{i}^\mathrm{p}(s)&0\\0&\Gamma_{i}^\mathrm{q}(s)
    \end{bmatrix},
\end{align}
which are selected to ensure passivity of $\mathcal{N}'(s)$. Specifically,  
\begin{align}\label{eq:Gamma_elements}
\begin{split}
    \Gamma_{i}^\mathrm{p}(s) &= \tfrac{1}{s}\left(
        \tfrac{\gamma_{1,i}^\mathrm{p}\omega_0^2s^2}{\omega_0^2+(\rho\omega_0+s)^2}+\tfrac{\gamma_{2,i}^\mathrm{p}\omega_0^2}{\omega_0^2+(\rho\omega_0+s)^2}+\gamma_{3,i}^\mathrm{p}\right)\\
        \Gamma_{i}^\mathrm{q}(s) &=\tfrac{1}{s}\left(\tfrac{\gamma_{1,i}^\mathrm{q}\omega_0^2s^2}{\omega_0^2+(\rho\omega_0+s)^2}+\tfrac{\gamma_{2,i}^\mathrm{q}\omega_0^2}{\omega_0^2+(\rho\omega_0+s)^2}+\gamma_{3,i}^\mathrm{q} \right)
\end{split}
\end{align}
where the parameters are selected as
\begin{align}\label{eq:Gamma_parameters}
\begin{split}
    \gamma_{1,i}^\mathrm{p} &= \gamma_{1,i}^\mathrm{q} =  \tfrac{2|v|_{\mathrm{max}}^2}{\omega_0^2}\textstyle\sum_{j\ne i}^nb_{ij} \\
    \gamma_{2,i}^\mathrm{p} &= \gamma_{2,i}^\mathrm{q} = -3|v|_{\mathrm{max}}^2\textstyle\sum_{j\ne i}^n b_{ij}\\
    \gamma_{3,i}^\mathrm{p} &= -\tfrac{\gamma_{2,i}^\mathrm{p}}{1+\rho^2}\\
    \gamma_{3,i}^\mathrm{q} &= -\tfrac{\gamma_{2,i}^\mathrm{q}}{1+\rho^2}+\tilde{\gamma}_{3,i}^\mathrm{q},
\end{split}
\end{align}
where $\tilde{\gamma}_{3,i}^\mathrm{q} \geq \tfrac{0.8}{1+\rho^2}\textstyle\sum_{j\ne i}^n b_{ij}$. The choice of the transfer function $\Gamma(s)$ is guided by structural insights into the dynamic network model in \cref{eq:full_network_polar_matrix_blocks_level0}, and algebraic considerations to ensure passivity of the loop-shifted system $\mathcal{N}'(s)$, as detailed below.
\begin{figure}[t!]
\centering
\resizebox{0.47\textwidth}{!}{
\begin{tikzpicture}[scale=1,every node/.style={scale=0.9}]
\draw  [top color=backgroundblue!40, bottom color=black!0,color=white](-1.6,3.5) node (v2) {} rectangle (8.2,-0.7);
\draw  [fill=black!0,color=black!0](-1.6,-0.7) rectangle (8.2,-2.5);
\draw [top color=black!0, bottom color=backgroundblue!40,color=white] (-1.6,-2.5) rectangle (8.2,-4.9);
\draw [rounded corners = 3,fill=backgroundblue!40](-1.4,3.25) rectangle (1.4,2.55);
\draw [rounded corners = 3,fill=backgroundblue!40] (-1.4,-3.95) rectangle (1.4,-4.65);
\draw [rounded corners = 3,fill=backgroundblue!40] (-1.4,1.45) rectangle (1.4,0.75);
\draw [rounded corners = 3,fill=black!0] (-1.4,-0.35) rectangle (1.4,-1.05);
\draw [rounded corners = 3,fill=black!0] (-1.4,-2.15) rectangle (1.4,-2.85);

\node at (0,2.9) {$\mathcal{D}_0\#\mathcal{N}_0$};
\node at (1.8,2.9) {\cref{eq:N_0_D_0}};
\node at (0,1.1) {$\mathcal{D}\#\mathcal{N}$};
\node at (1.8,1.1) {\cref{eq:N_D}};
\node at (0,-0.7) {$\mathcal{D}'\#\mathcal{N}'$};
\node at (1.8,-0.7) {\cref{eq:N_dash_D_dash}};
\node at (0,-2.5) {$\mathcal{D}'\hspace{-0.2mm}\#\mathcal{N}'\in\hspace{-0.7mm}\mathcal{RH}_\infty^{4n\times 4n}$};
\node at (0,-4.3) {$\mathcal{D}_0\#\mathcal{N}_0\hspace{-1mm}\in\hspace{-0.7mm}\mathcal{RH}_\infty^{4n\times 4n}$};

\draw[fill=gray!80] (-0.1,-1.3) node (v1) {} -- (-0.1,-1.7) -- (-0.2,-1.7) -- (0,-1.9) -- (0.2,-1.7) -- (0.1,-1.7) -- (0.1,-1.3) --(-0.1,-1.3);
\draw[fill=backgroundblue!200] (-0.1,0.5) node (v1) {} -- (-0.1,0.1) -- (-0.2,0.1) -- (0,-0.1) -- (0.2,0.1) -- (0.1,0.1) -- (0.1,0.5) --(-0.1,0.5);
\draw[fill=backgroundblue!200] (-0.1,-3.1) node (v1) {} -- (-0.1,-3.5) -- (-0.2,-3.5) -- (0,-3.7) -- (0.2,-3.5) -- (0.1,-3.5) -- (0.1,-3.1) --(-0.1,-3.1);
\draw[fill=backgroundblue!200] (-0.1,2.3) node (v1) {} -- (-0.1,1.9) -- (-0.2,1.9) -- (0,1.7) -- (0.2,1.9) -- (0.1,1.9) -- (0.1,2.3) --(-0.1,2.3);

\node at (4.4,2) {I. coordinate transformation};
\node at (4.4,0.2) {II. loop shifting with $\Gamma$};
\node at (4.4,-1.5) {III. passivity check of $\mathcal{N}'$, $\mathcal{D}'$};
\node at (4.4,-1.8) {\& Theorem 1};
\node at (4.4,-3.25) {IV. Final Value Theorem };
\node at (4.4,-3.55) {\& \cref{lemma:rhp_cancellations}};

\node[rotate=-90,color=backgroundblue!200] at (7.9,1.4) {preprocessing};
\node[rotate=-90,color=gray!80] at (7.9,-1.6) {stability};

\node [rotate=-90,color=gray!80] at (7.6,-1.6) {study};
\node[rotate=-90,color=backgroundblue!200] at (7.9,-3.7) {implications};
\draw  [dashed](-0.6,2.4) rectangle (7.3,1.6);
\draw [dashed] (-0.6,0.6) rectangle (7.3,-0.2);
\draw  [dashed](-0.6,-1.2) rectangle (7.3,-2);
\draw  [dashed](-0.6,-3) rectangle (7.3,-3.8);

\end{tikzpicture}
}
\vspace{-9mm}
\caption{Structural overview of the proof of \cref{thm:main_theorem}.}
\label{fig:proof_structure}
\vspace{-4mm}
\end{figure}

\subsubsection*{III. Passivity Checks \& \cref{thm:stabilty_passivity}}
For $\Gamma(s)$ selected as in \cref{eq:Gamma_2x2,eq:Gamma_elements,eq:Gamma_parameters}, we can show that $\mathcal{N}'(s)$ is passive, i.e., it satisfies the conditions (i) to (iii) in \cref{def:passivity}:

\textit{(i) Poles:} The poles of all elements of $\mathcal{N}'(s)$ are $p_1 = \mathrm{j}0$ and $p_{2,3}=-\rho\omega_0\pm \mathrm{j}\omega_0$, i.e., $\text{Re}(p_k)\leq0$ for $k\in\{1,2,3\}$.

\textit{(ii) Positive semi-definiteness:} We can express the Hermitian matrix $\mathcal{S}_{\mathcal{N}'}(\mathrm{j}\omega)\coloneqq\mathcal{N}'(\mathrm{j}\omega)+\mathcal{N}'^\star(\mathrm{j}\omega)$ as
    \begin{align}\label{eq:N_dash_N_dash_pos_semi_def}
    \begin{split}
\mathcal{S}_{\mathcal{N}'}(\mathrm{j}\omega)=\begin{bmatrix}
    \mathcal{S}_{\mathcal{N}',11}(\mathrm{j}\omega)&\dots&\mathcal{S}_{\mathcal{N}',1n}(\mathrm{j}\omega)\\\vdots&\ddots&\vdots\\\mathcal{S}_{\mathcal{N}',n1}(\mathrm{j}\omega)&\dots&\mathcal{S}_{\mathcal{N}',nn}(\mathrm{j}\omega)
\end{bmatrix},
        \end{split}
    \end{align}
    where each $\mathcal{S}_{\mathcal{N}',ij}$ represents a $2\times 2$ transfer matrix block. The diagonal and off-diagonal elements are given by
    \begin{align}
    \begin{split}
\hspace{-2mm}\mathcal{S}_{\mathcal{N}',ii}(\mathrm{j}\omega)&=h_\rho(\omega)\hspace{-0.5mm}\left(\hspace{-0.5mm}\textstyle\sum_{j\ne i}^n|v|_{0,i}^2b_{ij}\hspace{-0.5mm}\begin{bmatrix}
            -1&-\mathrm{j}\omega/\omega_0\\\mathrm{j}\omega/\omega_0&-1
        \end{bmatrix}\hspace{-0.5mm}+\hspace{-0.5mm}\right.\hspace{-0.5mm}\\
        &\quad\quad\quad\quad\quad\,\,\,\left.\begin{bmatrix}
            \omega^2\gamma_{1,i}^\mathrm{p}\hspace{-0.5mm}-\hspace{-0.5mm}\gamma_{2,i}^\mathrm{p} &0\\0&\omega^2\gamma_{1,i}^\mathrm{q}\hspace{-0.5mm} -\hspace{-0.5mm}\gamma_{2,i}^\mathrm{q} 
        \end{bmatrix}\right)\hspace{-2mm}\\
\hspace{-2mm}\mathcal{S}_{\mathcal{N}',ij}(\mathrm{j}\omega)&=h_\rho(\omega)|v|_{0,i}|v|_{0,j}b_{ij}\begin{bmatrix}
            1&\mathrm{j}\omega/\omega_0\\ -\mathrm{j}\omega/\omega_0&1
        \end{bmatrix}\hspace{-0.5mm}.
        \end{split}
    \end{align}
    Here, $h_\rho(\omega)$ is a positive function $\forall \omega\geq0$ defined as
    \begin{align}
        \hspace{-2mm}h_\rho(\omega)={4\rho\omega_0^3}/({(\omega_0^2\hspace{-0.5mm}+\hspace{-0.5mm}\rho^2\omega_0^2\hspace{-0.5mm}-\hspace{-0.5mm}\omega^2)^2\hspace{-0.5mm}+\hspace{-0.5mm}4\rho^2\omega_0^2\omega^2}) \geq 0.
    \end{align}
     We can observe that \cref{eq:N_dash_N_dash_pos_semi_def} is a Hermitian diagonally dominant matrix with real non-negative diagonal entries, i.e., the magnitude of the diagonal entry in a row is greater or equal to the sum of the magnitudes of the off-diagonal entries in that row. Namely, for the odd rows and $\forall \omega \geq 0$, we get
    \begin{align*}
    \begin{split}
        |\hspace{-0.5mm}\textstyle \sum_{j\ne i}^n\hspace{-0.6mm}-|v|_{0,i}^2b_{ij}\hspace{-0.5mm}+\hspace{-0.5mm}\omega^2\gamma_{1,i}^\mathrm{p}\hspace{-0.5mm}-\hspace{-0.5mm}\gamma_{2,i}^\mathrm{p}|\hspace{-0.6mm}&\geq\\
        |\hspace{-0.5mm}\textstyle\sum_{j\ne i}^n\hspace{-0.6mm}|v|_{0,i}^2b_{ij}\tfrac{-\mathrm{j}\omega}{\omega_0}|\hspace{-0.7mm}+\hspace{-0.7mm}\textstyle\sum_{j\ne i}^n\hspace{-0.3mm}|\hspace{-0.3mm}|v|_{0,i}&|v|_{0,j}b_{ij}|\hspace{-0.7mm}+\hspace{-0.7mm}\textstyle\sum_{j\ne i}^n\hspace{-0.6mm}|\hspace{-0.3mm}|v|_{0,i}|v|_{0,j}b_{ij}\tfrac{\mathrm{j}\omega}{\omega_0}|\hspace{-2mm}\\
        \Leftrightarrow\hspace{-0.5mm}
    |\hspace{-0.5mm}\textstyle \sum_{j\ne i}^n\hspace{-0.6mm}-|v|_\mathrm{max}^2b_{ij}\hspace{-0.5mm}+\hspace{-0.5mm}\omega^2\gamma_{1,i}^\mathrm{p}\hspace{-0.5mm}-\hspace{-0.5mm}\gamma_{2,i}^\mathrm{p} |\hspace{-0.6mm}&\geq\hspace{-0.6mm}\\
    2\textstyle\sum_{j\ne i}^n\hspace{-0.6mm}\hspace{-0.5mm}|v|_\mathrm{max}^2&b_{ij}\tfrac{\omega}{\omega_0}\hspace{-0.7mm}+\hspace{-0.7mm}\textstyle\sum_{j\ne i}^n\hspace{-0.6mm}\hspace{-0.5mm}|v|_\mathrm{max}^2b_{ij}\hspace{15mm}\\
   \Leftrightarrow\hspace{-0.5mm}\textstyle \sum_{j\ne i}^n\hspace{-0.6mm}|v|_\mathrm{max}^2b_{ij}(2+\tfrac{2\omega^2}{\omega_0^2})\hspace{-0.6mm}&\geq
    \textstyle\sum_{j\ne i}^n\hspace{-0.6mm}|v|_\mathrm{max}^2b_{ij}\hspace{-0.2mm}(1+\tfrac{2\omega}{\omega_0}),
        \end{split}
    \end{align*}
    and similarly, for each even row and $\forall \omega \geq 0$, we have
\begin{align*}
    \begin{split}
        \hspace{-2mm}|\hspace{-0.3mm}\textstyle \sum_{j\ne i}^n\hspace{-0.6mm}-|v|_{0,i}^2b_{ij}\hspace{-0.5mm}+\hspace{-0.5mm}\omega^2\gamma_{1,i}^\mathrm{q}\hspace{-0.5mm}-\hspace{-0.5mm}\gamma_{2,i}^\mathrm{q} |&\hspace{-0.6mm}\geq\hspace{-0.6mm} \\
        |\hspace{-0.6mm}\textstyle\sum_{j\ne i}^n\hspace{-0.6mm}|v|_{0,i}^2b_{ij}\tfrac{\mathrm{j}\omega}{\omega_0}|\hspace{-0.7mm}+\hspace{-0.7mm}\textstyle\sum_{j\ne i}^n\hspace{-0.3mm}|\hspace{-0.3mm}|v|_{0,i}&|v|_{0,j}b_{ij}\tfrac{-\mathrm{j}\omega}{\omega_0}\hspace{-0.3mm}|\hspace{-0.7mm}+\hspace{-0.7mm}\textstyle\sum_{j\ne i}^n\hspace{-0.6mm}|\hspace{-0.3mm}|v|_{0,i}|v|_{0,j}b_{ij}\hspace{-0.3mm}|\hspace{-0.7mm}\hspace{-2mm}\\
        \hspace{-2mm}\Leftrightarrow
    |\hspace{-0.3mm}\textstyle \sum_{j\ne i}^n\hspace{-0.6mm}-|v|_\mathrm{max}^2b_{ij}\hspace{-0.5mm}+\hspace{-0.5mm}\omega^2\gamma_{1,i}^\mathrm{q}\hspace{-0.5mm}-\hspace{-0.5mm}\gamma_{2,i}^\mathrm{q} |&\hspace{-0.6mm}\geq \hspace{-0.6mm}\\
    2\textstyle\sum_{j\ne i}^n\hspace{-0.6mm}|v|_\mathrm{max}^2&b_{ij}\tfrac{\omega}{\omega_0}\hspace{-0.7mm}+\hspace{-0.7mm}\textstyle\sum_{j\ne i}^n\hspace{-0.6mm}|v|_\mathrm{max}^2b_{ij}\\
    \hspace{-2mm}\Leftrightarrow\textstyle \sum_{j\ne i}^n\hspace{-0.6mm}|v|_\mathrm{max}^2b_{ij}(2+\tfrac{2\omega^2}{\omega_0^2})&\hspace{-0.6mm}\geq \hspace{-0.6mm}
    \textstyle\sum_{j\ne i}^n\hspace{-0.6mm}|v|_\mathrm{max}^2b_{ij}\hspace{-0.2mm}(1+\tfrac{2\omega}{\omega_0}).
        \end{split}
    \end{align*}
    By \cref{lemma:gershgorin}, we can conclude $\mathcal{N}'(\mathrm{j}\omega)+\mathcal{N}'^\star(\mathrm{j}\omega)\succeq 0$.
    
\textit{(iii) Imaginary poles:} For $\rho\ne0$, $\mathcal{N}'(s)$ has one imaginary pole, i.e., $p_1=\mathrm{j}0$, which is a simple pole. We therefore compute the limit $\mathcal{R}_{\mathrm{j}0}^{\mathcal{N}'}\coloneqq\text{lim}_{s\rightarrow \mathrm{j}0}(s-\mathrm{j}0)\mathcal{N}'(s)$, where each $\mathcal{R}_{\mathrm{j}0,ij}^{\mathcal{N}'}$ represents a $2\times2$ transfer matrix block. The diagonal and off-diagonal elements are given by
\begin{align}\label{eq:residue_N}
\begin{split}
    \hspace{-2mm}\mathcal{R}_{\mathrm{j}0,ii}^{\mathcal{N}'} &=\textstyle\sum_{j\ne i}^nb_{ij}\begin{bmatrix}
           \tfrac{|v|_{0,i}|v|_{0,j}}{1+\rho^2}&0\\0&\tfrac{2|v|_{0,i}^2-|v|_{0,i}|v|_{0,j}}{1+\rho^2}
        \end{bmatrix}+\\
        &\quad\quad\quad\quad\quad\quad\quad\quad\,\,\,\begin{bmatrix}
            \tfrac{\gamma_{2,i}^\mathrm{p}}{1+\rho^2}\hspace{-0.5mm}+\hspace{-0.5mm}\gamma_{3,i}^\mathrm{p} &0\\0&\tfrac{\gamma_{2,i}^\mathrm{q}}{1+\rho^2}\hspace{-0.5mm} +\hspace{-0.5mm}\gamma_{3,i}^\mathrm{q} 
        \end{bmatrix}\hspace{-2mm}\\
        \hspace{-2mm}\mathcal{R}_{\mathrm{j}0,ij}^{\mathcal{N}'} &=-b_{ij}\tfrac{|v|_{0,i}|v|_{0,j}}{1+\rho^2}\begin{bmatrix}
           1&0\\0&1
        \end{bmatrix}.
\end{split}
\end{align}
Again, \cref{eq:residue_N} is a Hermitian diagonally dominant matrix with real non-negative diagonal entries. For the odd rows, we get
\begin{align*}
    \begin{split}
         \hspace{-2mm}|\hspace{-0.5mm}\textstyle \sum_{j\ne i}^n\hspace{-0.6mm}b_{ij}\tfrac{|v|_{0,i}|v|_{0,j}}{1+\rho^2}\hspace{-0.5mm}+\hspace{-0.5mm}\tfrac{\gamma_{2,i}^\mathrm{p}}{1+\rho^2}\hspace{-0.5mm}+\hspace{-0.5mm}\gamma_{3,i}^\mathrm{p}|&\hspace{-0.6mm}\geq\hspace{-0.6mm} \textstyle\sum_{j\ne i}^n\hspace{-0.3mm}|\hspace{-0.3mm}-b_{ij}\tfrac{|v|_{0,i}|v|_{0,j}}{1+\rho^2}|\hspace{-0.7mm}\\
         \Leftrightarrow \textstyle \sum_{j\ne i}^n\hspace{-0.6mm}b_{ij}\tfrac{|v|_{0,i}|v|_{0,j}}{1+\rho^2}&\hspace{-0.6mm}\geq\hspace{-0.6mm}\textstyle \sum_{j\ne i}^n\hspace{-0.6mm}b_{ij}\tfrac{|v|_{0,i}|v|_{0,j}}{1+\rho^2}.
    \end{split}
\end{align*}
For the even rows we get (with $|v|_\mathrm{max}\hspace{-0.5mm}=\hspace{-0.5mm}1.1$ and $|v|_\mathrm{min}\hspace{-0.5mm}=\hspace{-0.5mm}0.9$):
\begin{align*}
    \begin{split}
        \hspace{-2mm}|\hspace{-0.5mm}\textstyle \sum_{j\ne i}^n\hspace{-0.6mm}\tfrac{(2|v|_{0,i}^2-|v|_{0,i}|v|_{0,j})b_{ij}}{1+\rho^2}\hspace{-0.5mm}+\hspace{-0.5mm}\tfrac{\gamma_{2,i}^\mathrm{q}}{1+\rho^2}\hspace{-0.5mm}+\hspace{-0.5mm}\gamma_{3,i}^\mathrm{q}|&\hspace{-0.6mm}\geq\hspace{-0.6mm} \textstyle\sum_{j\ne i}^n\hspace{-0.3mm}|b_{ij}\tfrac{-|v|_{0,i}|v|_{0,j}}{1+\rho^2}|\hspace{-0.7mm}\\
    \Leftrightarrow|\hspace{-0.5mm}\textstyle \sum_{j\ne i}^n\hspace{-0.6mm}\tfrac{(2|v|_{0,i}^2-|v|_{0,i}|v|_{0,j})b_{ij}}{1+\rho^2}+\tilde{\gamma}_{3,i}^\mathrm{q}|&\hspace{-0.6mm}\geq\hspace{-0.6mm} \textstyle\sum_{j\ne i}^n\hspace{-0.3mm}|b_{ij}\tfrac{-|v|_{0,i}|v|_{0,j}}{1+\rho^2}|\hspace{-0.7mm}\\
        \Leftrightarrow\textstyle \sum_{j\ne i}^n\hspace{-0.6mm}b_{ij}\tfrac{2|v|_{\mathrm{min}}^2-|v|_{\mathrm{max}}^2}{1+\rho^2}+\tilde{\gamma}_{3,i}^\mathrm{q}&\hspace{-0.6mm}\geq\hspace{-0.6mm} \textstyle\sum_{j\ne i}^n\hspace{-0.3mm}\hspace{-0.3mm}b_{ij}\tfrac{|v|_{\mathrm{max}^2}}{1+\rho^2}\hspace{-0.7mm}
    \end{split}
\end{align*}
By \cref{lemma:gershgorin}, we conclude $\mathcal{R}_{\mathrm{j}0}^{\mathcal{N}'}\succeq 0$.
  \begin{figure}[t!]
    \centering
    \vspace{-1mm}
    \begin{subfigure}{0.21\textwidth}
    \centering
     \resizebox{1.1\textwidth}{!}{
\tikzstyle{roundnode}=[circle,draw=black!60,fill=black!5,scale=0.7]
\begin{tikzpicture}[scale=1,every node/.style={scale=0.8}]
\draw [rounded corners = 3,fill=backgroundblue!30] (-0.3,4.3) rectangle (1.3,3.5);
\node [scale=1.2] at (0.5,3.9) {$\mathcal{D}(s)$};
\draw [rounded corners = 3,fill=backgroundblue!30] (-0.3,2.7) rectangle (1.3,1.9);
\node [scale=1.2] at (0.5,2.3) {$\mathcal{N}(s)$};
\draw [-latex](1.3,3.9) -- (2.5,3.9);
\draw (-0.3,2.3) -- (-1.1,2.3);

\draw[-latex] (1.8,2.3) -- (1.3,2.3);
\draw [-latex](2.6,2.3) -- (2,2.3); 
\draw [-latex](1.9,3.9) -- (1.9,2.4);
\node [roundnode] at (1.9,2.3) {};

\draw [-latex](-1,3.9) -- (-0.3,3.9);
\draw [-latex](-1.8,3.9) -- (-1.2,3.9); 
\draw [-latex](-1.1,2.3) node (v1) {} -- (-1.1,3.8);
\node [roundnode] at (-1.1,3.9) {};
\node at (-0.9,3.6) {$-$};
\node at (-1.8,4.35) {$\begin{bmatrix}\Delta p_\mathrm{d}\\ \Delta q_\mathrm{d}\end{bmatrix}$};
\node at (2.5,4.35) {$\begin{bmatrix} \Delta \omega\\\Delta |\tilde{v}|_\mathrm{n}\end{bmatrix}$};
\node at (-0.65,4.35) {$\begin{bmatrix}\Delta p\\ \Delta q\end{bmatrix}$};
\node at (2.6,2.75) {$\begin{bmatrix} \Delta\omega_\mathrm{d} \\\Delta |\tilde{v}|_\mathrm{nd}\end{bmatrix}$};
\node at (0.5,3.3) {device dynamics};
\node at (0.5,1.7) {network dynamics};
\node at (-1.05,4.35) {$-$};
\draw[-latex](-1.1,2.3) -- (-1.8,2.3);
\node at (-1.8,2.75) {$\begin{bmatrix}\Delta p_\mathrm{e}\\ \Delta q_\mathrm{e}\end{bmatrix}$};
\end{tikzpicture}
}
        \vspace{-8.5mm}
    \caption{Original feedback system.}
    \label{fig:power_sys_model}
\end{subfigure}
\hspace{0.05cm}
\vspace{-1mm}
\begin{subfigure}{0.24\textwidth}
\centering
\resizebox{1.1\textwidth}{!}{
\tikzstyle{roundnode}=[circle,draw=black!60,fill=black!5,scale=0.75]
\begin{tikzpicture}[scale=1,every node/.style={scale=0.8}]
\draw [rounded corners = 3, dashed,fill=backgroundblue!10] (2,3.7) rectangle (-1.1,1.6);
\draw [rounded corners = 3, dashed, fill=backgroundblue!10] (2,4.4) rectangle (-1.1,6.5);
\draw [rounded corners = 3,fill=backgroundblue!30] (-0.3,6.2) rectangle (1.3,5.4);
\draw [rounded corners = 3,fill=backgroundblue!30] (-0.3,2.5) rectangle (1.3,1.7);
\node [scale=1.2] at (0.5,5.8) {$\mathcal{D}(s)$};
\node [backgroundblue!200,scale=1.2] at (1.65,6.15) {$\mathcal{D}'$};
\node [scale=1.2] at (0.5,2.1) {$\mathcal{N}(s)$};
\node [backgroundblue!200,scale=1.2] at (1.65,3.35) {$\mathcal{N}'$};
\draw [-latex](1.3,5.8) -- (3,5.8);
\draw (-0.9,2.1) -- (-1.4,2.1);

\draw[-latex] (2.2,2.1) -- (1.3,2.1);

\draw [-latex](2.3,5.8) -- (2.3,2.2);

\draw [-latex](-1.3,5.8) -- (-0.9,5.8);
\draw [-latex](-2.1,5.8) -- (-1.5,5.8); 
\draw [-latex](-1.4,2.1) node (v2) {} -- (-1.4,5.7);
\node [roundnode] at (-1.4,5.8) {};
\node at (-1.25,5.5) {$-$};

\draw [rounded corners = 3, fill = black!20] (-0.3,3.4) rectangle (1.3,2.6);
\node [scale=1.2] at (0.5,3) {$\Gamma(s)$};
\draw [rounded corners = 3, fill = black!20] (-0.3,5.3) rectangle (1.3,4.5);
\node [scale=1.2] at (0.5,4.9) {$\Gamma(s)$};
\draw [-latex](1.8,2.1) -- (1.8,3) -- (1.3,3);
\draw [-latex](-0.3,2.1) -- (-0.7,2.1); 
\draw [-latex](-0.3,3) -- (-0.8,3) -- (-0.8,2.2);
\node [roundnode] at (-0.8,2.1) {}; 

\node[roundnode] at (-0.8,5.8) {}; 
\draw [-latex](-0.7,5.8) -- (-0.3,5.8);
\draw [-latex](-0.3,4.9) -- (-0.8,4.9) -- (-0.8,5.7);
\draw[-latex] (1.8,5.8) -- (1.8,4.9) -- (1.3,4.9);

\node at (-2.1,6.25) {$\begin{bmatrix}\Delta p_\mathrm{d}\\ \Delta q_\mathrm{d}\end{bmatrix}$};
\node at (3,6.25) {$\begin{bmatrix} \Delta \omega\\\Delta |\tilde{v}|_\mathrm{n}\end{bmatrix}$};
\node at (0.5,4.2) {device dynamics};
\node at (0.5,1.4) {network dynamics};
\node [roundnode] (v1) at (2.3,2.1) {};

\draw [-latex](3,2.1) --  (2.4,2.1);
\draw [-latex](-1.4,2.1) -- (-2.1,2.1);
\node at (-2.1,2.55) {$\begin{bmatrix}\Delta p_\mathrm{e}\\ \Delta q_\mathrm{e}\end{bmatrix}$};
\node at (3,2.55) {$\begin{bmatrix} \Delta\omega_\mathrm{d} \\\Delta |\tilde{v}|_\mathrm{nd}\end{bmatrix}$};
\end{tikzpicture}
}
    \vspace{-8.5mm}
    \caption{Loop-shifting with $\Gamma$.}
    \label{fig:loop_shifted_ND}
    \end{subfigure}
    \caption{\footnotesize Closed-loop feedback interconnection for stability analysis.}
    \label{fig:loop_shifting_ND}
    \vspace{-4mm}
\end{figure}

Next, we derive the decentralized stability conditions in \cref{eq:p_conditions,eq:q_conditions} under which $\mathcal{D}'(s)$ is strictly passive. We consider $\mathcal{D}'(s)=\text{diag}(\mathcal{D}_1'(s),\dots,\mathcal{D}_n'(s))$ with the matrix elements
\begin{align}
    \begin{split}
        \hspace{-1mm}\mathcal{D}_i'(s)\hspace{-0.5mm}=\hspace{-0.5mm}\mathcal{D}_i(s)(I\hspace{-0.5mm}-\hspace{-0.5mm}\Gamma_i(s)\mathcal{D}_i(s))^{-1}\hspace{-0.5mm}=\hspace{-0.5mm}\begin{bmatrix}
            \mathcal{D}_{\mathrm{p},i}'(s)&0\\0&\mathcal{D}'_{\mathrm{q},i}(s)
        \end{bmatrix}\hspace{-1mm},\hspace{-1mm}
    \end{split}
\end{align}
where the diagonal transfer function elements are given as
\begin{subequations}
\begin{align}
\begin{split}\label{eq:D_dash_p}
    \mathcal{D}_{\mathrm{p},i}'(s)&=\tfrac{d_{\mathrm{p},i}}{\tau_{\mathrm{p},i}s+1-\Gamma_i^\mathrm{p}(s)d_{\mathrm{p},i}}
    \end{split}\\
    \begin{split}\label{eq:D_dash_q}
   \mathcal{D}_{\mathrm{q},i}'(s)&= \tfrac{\tfrac{d_{\mathrm{q},i}s}{|v|_{0,i}}}{\tau_{\mathrm{q},i}s+1-\Gamma_i^\mathrm{q}(s)\tfrac{d_{\mathrm{q},i}s}{|v|_{0,i}}},
\end{split}
\end{align}
\end{subequations}
for each of which we check the strict passivity conditions (i) and (ii) in \cref{def:strict_passivity} independently. We start with $\mathcal{D}_{\mathrm{p},i}'(s)$ and insert the expression for $\Gamma_i^\mathrm{p}(s)$ in \cref{eq:Gamma_elements} into \cref{eq:D_dash_p}, i.e.,
\begin{align}
    \mathcal{D}_{\mathrm{p},i}'(s) = d_{\mathrm{p},i}\tfrac{s^2 b_{2,i}+sb_{1,i}+b_{0,i}}{s^3a_{3,i}+s^2a_{2,i}+sa_{1,i}+a_{0,i}},
\end{align}
where the transfer function coefficients are given by
\begin{align}
    \begin{split}
        \hspace{-1mm}a_{0,i} &\hspace{-0.5mm}=\hspace{-0.5mm}\omega_0^2(1\hspace{-0.5mm}+\hspace{-0.5mm}\rho^2)-2\rho\omega_0d_{\mathrm{p},i}\gamma_{3,i}^\mathrm{p}\\
        \hspace{-1mm}a_{1,i} &\hspace{-0.5mm}= 2\rho\omega_0+\tau_{\mathrm{p,i}}\omega_0^2(1\hspace{-0.5mm}+\hspace{-0.5mm}\rho^2)\hspace{-0.5mm}\\&\hspace{-0.5mm}-d_{\mathrm{p},i}\omega_0^2\gamma_{1,i}^\mathrm{p}-d_{\mathrm{p},i}\gamma_{3,i}^\mathrm{p}\\
        \hspace{-1mm}a_{2,i} &\hspace{-0.5mm}= \hspace{-0.5mm}2\omega_0{\tau}_{\mathrm{p},i}\rho\hspace{-0.5mm}+\hspace{-0.5mm}1\\
        \hspace{-1mm}a_{3,i} &\hspace{-0.5mm}=\hspace{-0.5mm}  \tau_{\mathrm{p},i}
        \end{split}
        \begin{split}
        b_{0,i} &\hspace{-0.5mm}=\hspace{-0.5mm} \omega_0^2(1\hspace{-0.5mm}+\hspace{-0.5mm}\rho^2)\hspace{-1mm}\\
        b_{1,i} &\hspace{-0.5mm}=\hspace{-0.5mm} 2\rho\omega_0\hspace{-1mm}\\
        b_{2,i} &\hspace{-0.5mm}=\hspace{-0.5mm} 1.
    \end{split}
\end{align}

\textit{(i) Poles:} To show that the poles of all elements of $\mathcal{D}_{\mathrm{p},i}'(s)$ are in $\text{Re}(s)<0$, we check the Hurwitz criterion for a 3rd order polynomial, and conclude that $a_{3,i}>0$ and $a_{2,i}>0$ are always satisfied, while $a_{1,i}>0$ is ensured by \cref{eq:condition1_p}, $a_{0,i}>0$ by \cref{eq:condition2_p}, and $a_{2,i}a_{1,i}>a_{0,i}a_{3,i}$ by \cref{eq:condition3_p}.

\textit{(ii) Positive definiteness:} To ensure $\mathcal{D}_{\mathrm{p},i}'(\mathrm{j}\omega)+\mathcal{D}'^\star_{\mathrm{p},i}(\mathrm{j}\omega)\succ 0,\,\forall \omega \in (-\infty,\infty)$, we require the additional condition \cref{eq:condition4_p}.

Following the same reasoning, we derive the decentralized stability conditions for the reactive power-voltage droop control in \cref{eq:q_conditions} by checking strict passivity of $\mathcal{D}_{\mathrm{q},i}'(s)$. To do so, we select $\tfrac{d_{\mathrm{q},i}}{|v|_{0,i}}=\tfrac{1}{\tilde{\gamma}_{3,i}^\mathrm{q}}$ to cancel the zero at the origin, and ensure strict passivity of the minimal realization of $\mathcal{D}_{\mathrm{q},i}'(s)$.

Since $\mathcal{N}'(s)$ is passive and $\mathcal{D}'(s)$ is strictly passive, and $\bar{\sigma}(\mathcal{N}'(\mathrm{j}\infty))\bar{\sigma}(\mathcal{D}'(\mathrm{j}\infty))<1$ (because $\mathcal{N}'(\mathrm{j}\infty)=0_{2n\times 2n}$), we can apply \cref{thm:stabilty_passivity} and conclude $\mathcal{D}'\#\mathcal{N}'\in\mathcal{RH}_\infty^{4n\times 4n}$.

\subsubsection*{IV. Final Value Theorem \& \cref{lemma:rhp_cancellations}}
Internal feedback stability of the original system $\mathcal{D}_0\#\mathcal{N}_0$ follows directly from internal feedback stability of the minimal realization of all four closed-loop transfer functions of $\mathcal{D}'\#\mathcal{N}'$:

\begin{lemma}\label{lem:volt_deriv_volt}
  Consider the device and network dynamics $\mathcal{D}'(s)$ and $\mathcal{N}'(s)$ in \cref{eq:N_dash_D_dash}, as well as $\mathcal{D}_0(s)$ and $\mathcal{N}_0(s)$  in \cref{eq:N_0_D_0} with the network model \cref{eq:full_network_polar_matrix_blocks_level0}. Then, internal feedback stability of the minimal realization of all four closed-loop transfer functions of the loop-shifted system $\mathcal{D}'\#\mathcal{N}'$ implies internal feedback stability of $\mathcal{D}_0\#\mathcal{N}_0$.
\end{lemma}

\begin{proof}
We first conclude that $\mathcal{D}'\#\mathcal{N}'\in\mathcal{RH}_\infty^{4n\times 4n}$ implies $(I\hspace{-0.5mm}+\hspace{-0.5mm}\mathcal{D}'\mathcal{N}')^{-1}\mathcal{D}'\hspace{-0.5mm}=\hspace{-0.5mm}(I\hspace{-0.5mm}+\hspace{-0.5mm}\mathcal{D}\mathcal{N})^{-1}\mathcal{D}\hspace{-0.5mm}\in\hspace{-0.5mm}\mathcal{RH}_\infty^{2n\times 2n}$, where the latter closed-loop transfer functions are minimal realizations. Further, since 
\begin{align}
    \begin{split}
        \mathcal{D}(s)&=\mathcal{D}_0(s)\cdot\text{diag}(1,\tfrac{s}{|v|_{0,1}},\dots,1,\tfrac{s}{|v|_{0,n}})\\
        \mathcal{N}(s)&= \mathcal{N}_0(s)\cdot\text{diag}(1,\tfrac{|v|_{0,1}}{s},\dots,1,\tfrac{|v|_{0,n}}{s}),
    \end{split}
\end{align}
$(I\hspace{-0.5mm}+\hspace{-0.5mm}\mathcal{D}\mathcal{N})^{-1}\mathcal{D}\hspace{-0.5mm}=\hspace{-0.5mm}\text{diag}(1,\tfrac{s}{|v|_{0,1}},\dots,1,\tfrac{s}{|v|_{0,n}})(I\hspace{-0.5mm}+\hspace{-0.5mm}\mathcal{D}_0\mathcal{N}_0)^{-1}\mathcal{D}_0$, and we can thus conclude stability of $(I+\mathcal{D}_0\mathcal{N}_0)^{-1}\mathcal{D}_0$ if the step response of the voltage derivatives converges to zero, i.e.,
\begin{align}\label{eq:lim_diag_star_zero}
    \underset{s\rightarrow 0}{\text{lim}}\,\, s \frac{1}{s}(I+\mathcal{D}(s)\mathcal{N}(s))^{-1}\mathcal{D}(s) = I \otimes \begin{bmatrix}
\star & 0 \\
0 & 0
\end{bmatrix}.
\end{align}
To show that \cref{eq:lim_diag_star_zero} holds, we first rewrite the transfer matrix $(I+\mathcal{D}\mathcal{N})^{-1}\mathcal{D} = (I+\mathcal{D}{N}\tfrac{1}{s})^{-1}\mathcal{D}$, and then use the row permutation matrix $\mathcal{P}\in\mathbb{R}^{2n\times 2n}$ with elements
\begin{align}
    \mathcal{P}_{ij} = \begin{cases}
        1 \quad \quad i=k,\, j = 2k-1,\, 1\leq k\leq n\\ 1 \quad \quad i=k+n,\, j = 2k,\, 1\leq k\leq n\\ 0\quad \quad \text{else},
    \end{cases}
\end{align}
to study the decoupled frequency and voltage dynamics, i.e., $\mathcal{P} (I+\mathcal{D}{N}\tfrac{1}{s})^{-1}\mathcal{D}\mathcal{P}^{-1} = (I+\mathcal{D}^\mathcal{P}{N}^\mathcal{P}\tfrac{1}{s})^{-1}\mathcal{D}^\mathcal{P}$. The row-permutated matrices $\mathcal{D}^\mathcal{P}(s)$ and $N^\mathcal{P}(s)$ are given as
\begin{align}
    \hspace{-1mm}\mathcal{D}^\mathcal{P}\hspace{-0.5mm}(s) \hspace{-0.5mm}= \hspace{-1mm}\begin{bmatrix}\mathcal{D}_\mathrm{p}^\mathcal{P}\hspace{-0.5mm}(s)\hspace{-0.5mm} &  0_{n\times n}\\
        0_{n\times n} & \hspace{-0.5mm}\mathcal{D}_\mathrm{q}^\mathcal{P}\hspace{-0.5mm}(s)
    \end{bmatrix}\hspace{-1mm}, \,\, N^\mathcal{P}\hspace{-0.5mm}(s)\hspace{-0.5mm} = \hspace{-1mm}\begin{bmatrix}N_1^\mathcal{P}(s) \hspace{-0.5mm}& \hspace{-0.5mm}N_2^\mathcal{P}\hspace{-0.5mm}(s)  \\
        -N_2^\mathcal{P}\hspace{-0.5mm}(s) \hspace{-0.5mm} & \hspace{-0.5mm}N_3^\mathcal{P}(s) 
    \end{bmatrix}\hspace{-1mm}
\end{align}
where $\mathcal{D}_\mathrm{p}^\mathcal{P}(s)=\text{diag}(\mathcal{D}_{\mathrm{p},1}(s),\dots,\mathcal{D}_{\mathrm{p},n}(s))$ and $\mathcal{D}_\mathrm{q}^\mathcal{P}(s)=\text{diag}(\mathcal{D}_{\mathrm{q},1}(s),\dots,\mathcal{D}_{\mathrm{q},n}(s))$ with elements 
\begin{align}
\mathcal{D}_{\mathrm{p},i}(s) = \tfrac{d_{\mathrm{p},i}}{\tau_{\mathrm{p},i}s+1} \quad\text{and}\quad \mathcal{D}_{\mathrm{q},i}(s) = \tfrac{d_{\mathrm{q},i}s}{\tau_{\mathrm{q},i}s+1}\tfrac{1}{|v|_{0,i}}.
\end{align}
Moreover, we have
\begin{align}
    N_{1,ii}^\mathcal{P}(s) &= \textstyle\sum_{j\ne i}^nb_{ij}\tfrac{|v|_{0,i}^2}{1+(\rho+\tfrac{s}{\omega_0})^2}-\textstyle\sum_{j\ne i}^nb_{ij}\tfrac{|v|_{0,i}^2-|v|_{0,i}|v|_{0,j}}{1+\rho^2}\nonumber\\
    N_{1,ij}^\mathcal{P}(s) &= - b_{ij}\tfrac{|v|_{0,i}|v|_{0,j}}{1+(\rho+\tfrac{s}{\omega_0})^2}\nonumber\\
    \begin{split}
    N_{2,ii}^\mathcal{P}(s) &= \textstyle\sum_{j\ne i}^nb_{ij}\tfrac{|v|_{0,i}^2}{1+(\rho+\tfrac{s}{\omega_0})^2}\tfrac{s}{\omega_0}\\
    N_{2,ij}^\mathcal{P}(s) &= - b_{ij}\tfrac{|v|_{0,i}|v|_{0,j}}{1+(\rho+\tfrac{s}{\omega_0})^2}\tfrac{s}{\omega_0}
    \end{split}\\
    N_{3,ii}^\mathcal{P}(s) &= \textstyle\sum_{j\ne i}^nb_{ij}\tfrac{|v|_{0,i}^2}{1+(\rho+\tfrac{s}{\omega_0})^2}+\textstyle\sum_{j\ne i}^nb_{ij}\tfrac{|v|_{0,i}^2-|v|_{0,i}|v|_{0,j}}{1+\rho^2}\nonumber\\
    N_{3,ij}^\mathcal{P}(s) &= - b_{ij}\tfrac{|v|_{0,i}|v|_{0,j}}{1+(\rho+\tfrac{s}{\omega_0})^2}\nonumber.
\end{align}
Recalling that $(I\hspace{-0.5mm}+\hspace{-0.5mm}\mathcal{D}\mathcal{N})^{-1}\mathcal{D}$ is stable, and thus its permuted version, we can apply the Final Value Theorem (FVT):
\begin{subequations}\label{eq:FVT}
\begin{align}
    &\,\,\underset{s\rightarrow 0}{\text{lim}}\,\,(I+\mathcal{D}^\mathcal{P}(s){N}^\mathcal{P}(s)\tfrac{1}{s})^{-1}\mathcal{D}^\mathcal{P}(s) \\
    =\,&\,\,\underset{s\rightarrow 0}{\text{lim}}\,\,(sI+\mathcal{D}^\mathcal{P}(s){N}^\mathcal{P}(s))^{-1}\mathcal{D}^\mathcal{P}(s) s\\ \nonumber
    = \,&\,\,\underset{s\rightarrow 0}{\text{lim}}\,\left[\begin{array}{c}
    \hspace{-0.15cm}(sI+\mathcal{D}^\mathcal{P}_\mathrm{p}(s){N}^\mathcal{P}_1(s))^{-1}\mathcal{D}^\mathcal{P}_\mathrm{p}(s) s \\
    \hspace{-0.15cm}0_{n\times n}
\end{array}\right.\\\label{eq:FVT3rd}
\quad 
&\hspace{2.5cm}\left.\begin{array}{c}
     0_{n\times n}\hspace{-0.15cm} \\
     (sI+\mathcal{D}^\mathcal{P}_\mathrm{q}(s){N}^\mathcal{P}_3(s))^{-1}\mathcal{D}^\mathcal{P}_\mathrm{q}(s) s\hspace{-0.15cm}
\end{array}\right]\hspace{-0.1cm}\\ \label{eq:FVT4th}
         = \,&\begin{bmatrix}
            \star_{n\times n}&0_{n\times n}\\0_{n\times n}&0_{n\times n}
        \end{bmatrix},
        \end{align}
\end{subequations}
where for \cref{eq:FVT3rd} we have used $N_2^\mathcal{P}(0)=0_{n\times n}$. The last equality in \cref{eq:FVT4th} follows from the fact that $N_1^\mathcal{P}(0)$ is a Laplacian matrix with zero eigenvalue \cite{paganini2019global}, while $N_3^\mathcal{P}(0)$ is a regular matrix which does not have a zero eigenvalue. Given \cref{eq:FVT}, \cref{eq:lim_diag_star_zero} follows, and thus $(I+\mathcal{D}_0\mathcal{N}_0)^{-1}\mathcal{D}_0\in\mathcal{RH}_\infty^{2n\times 2n}$.

Finally, since $\mathcal{D}_0(s)$ does not include a RHP pole or zero, there are no RHP pole-zero cancellations between $\mathcal{D}_0(s)$ and $\mathcal{N}_0(s)$. By \cref{lemma:rhp_cancellations}, we conclude $\mathcal{D}_0\#\mathcal{N}_0\in\mathcal{RH}_\infty^{4n\times 4n}$.
\end{proof}

\begin{remark}
    Recall again that we concluded stability of $\mathcal{D}_0\#\mathcal{N}_0\in\mathcal{RH}_\infty^{4n\times 4n}$ by first proving stability of the minimal realization of all four closed-loop transfer functions of $\mathcal{D}'\#\mathcal{N}'$, which implies stability of $(I\hspace{-0.5mm}+\hspace{-0.5mm}\mathcal{D}'\mathcal{N}')^{-1}\mathcal{D}'\hspace{-0.5mm}=\hspace{-0.5mm}(I\hspace{-0.5mm}+\hspace{-0.5mm}\mathcal{D}\mathcal{N})^{-1}\mathcal{D}\hspace{-0.5mm}\in\hspace{-0.5mm}\mathcal{RH}_\infty^{2n\times 2n}$, where the latter closed-loop transfer functions are minimal realizations. Notice that there are RHP pole-zero cancellations between  $\mathcal{D}(s)$ and $\mathcal{N}(s)$, i.e., we can NOT conclude stability of $\mathcal{D}\#\mathcal{N}$. However, from the stability of $\hspace{-0.5mm}(I\hspace{-0.5mm}+\hspace{-0.5mm}\mathcal{D}\mathcal{N})^{-1}\mathcal{D}\hspace{-0.5mm}\in\hspace{-0.5mm}\mathcal{RH}_\infty^{2n\times 2n}$, we can conclude stability of $\hspace{-0.5mm}(I\hspace{-0.5mm}+\hspace{-0.5mm}\mathcal{D}_0\mathcal{N}_0)^{-1}\mathcal{D}_0\hspace{-0.5mm}\in\hspace{-0.5mm}\mathcal{RH}_\infty^{2n\times 2n}$ (cf. proof of \cref{lem:volt_deriv_volt}), and, given there are no RHP pole-zero cancellations between $\mathcal{D}_0(s)$ and $\mathcal{N}_0(s)$, we conclude stability of $\mathcal{D}_0\#\mathcal{N}_0\in\mathcal{RH}_\infty^{4n\times 4n}$.
\end{remark}

\section{Numerical Case Studies}\label{sec:numerical_case_studies}
To validate the proposed stability guarantees, numerical case studies are conducted in MATLAB/Simulink.
\begin{figure}[t!]
\centering
\resizebox{0.3\textwidth}{!}{
\tikzstyle{roundnode}=[circle,draw=black!60,fill=black!5,scale=0.7]
\begin{tikzpicture}[scale=1,every node/.style={scale=0.7}]

\draw[ultra thick] (-3.2,3.9) -- (-3.2,3.3); 
\draw (-3.2,3.6) node (v1) {} -- (-3.5,3.6);
\draw [fill=gray!30] (-4.1,3.9) rectangle (-3.5,3.3);
\draw (-4.1,3.3) -- (-3.5,3.9);
\node at (-3.9,3.7) {$=$};
\node at (-3.7,3.5) {$\approx$};
\draw (-3.2,3.8) -- (-0.6,3.8); 
\draw[ultra thick] (-0.6,3.9) -- (-0.6,3.3); 
\draw (-0.6,3.6) node (v2) {} -- (-0.3,3.6);
\draw  [fill=gray!30] (-0.3,3.9) rectangle (0.3,3.3);
\draw (-0.3,3.3) -- (0.3,3.9); 
\draw  (-3.2,3.6) -- (-2.9,3.6) -- (-2.1,2.8)-- (-2.1,2.6) ; 
\draw (-0.6,3.6) -- (-0.9,3.6) -- (-1.7,2.8)-- (-1.7,2.6);
\draw [ultra thick](-2.2,2.6) -- (-1.6,2.6);
\draw (-1.9,2.3) -- (-1.9,2.6) node (v3) {};
\draw  [fill=backgroundblue!30] (-2.2,2.3) rectangle (-1.6,1.7);
\draw (-2.2,1.7) -- (-1.6,2.3);
\node at (-2,2.1) {$\approx$};
\node at (-1.8,1.9) {$=$};
\node at (-0.1,3.7) {$\approx$};
\node at (0.1,3.5) {$=$};
\node at (-3.8,3.1) {GFM 1};
\node at (-1.1,2) {GFM 3};
\node at (0,3.1) {GFM 2};
\draw [-latex,thick](-3.2,3.4) -- (-3,3.4) -- (-3,3); 
\draw[-latex,thick] (-0.6,3.4) -- (-0.8,3.4) -- (-0.8,3); 
\draw[-latex,thick](-1.9,2.6) -- (-1.9,3.1); 
\end{tikzpicture}
}
\vspace{-10mm}
\caption{Schematic one-line diagram of the three-bus power system grid topology used in Case Studies I and II.}
\label{fig:3bus}
\vspace{-2mm}
\end{figure}

\renewcommand{\arraystretch}{1}
\begin{table}[t!]\scriptsize
    \centering
           \caption{Network and GFM parameters of the three-bus system}
           \vspace{-1.5mm}
           \begin{subtable}{1\columnwidth}
           \centering
               \caption{Electrical parameters \& set points (in local per unit systems).}
               \vspace{-2mm}
               \begin{tabular}{c||c|c}
     \toprule
         Parameter & Symbol & Value  \\ \hline
         $\hspace{-2mm}$Base power, voltage, frequency$\hspace{-2mm}$ &  $S_\mathrm{b}$, $V_\mathrm{b}$, $f_\mathrm{b}$ & $\hspace{-2mm}$100 MVA, 230 kV, 50 Hz$\hspace{-3mm}$\\
         $RL$ line components & $r_{ij}$, $l_{ij}$ & 0.02, 0.4 pu \\
         Shunt capacitance in $\pi$-model & $c_{ij}$ & 0.1 pu (0.05 pu per end)\\
         Constant power loads & $p_{\mathrm{load},i}$& 1 pu\\\hline
         GFM base power& $S_1$, $S_2$, $S_3$ & 100, 100, 200 MVA\\
         $\hspace{-2mm}$ GFM steady-state voltages$\hspace{-2mm}$ & $\hspace{-1.5mm}$$|v|_{0,1}$,$|v|_{0,2}$,$|v|_{0,3}$$\hspace{-1.5mm}$& 1, 0.9, 1.1 pu\\
         GFM active power set points & $p_{0,1}$, $p_{0,2}$, $p_{0,3}$& 0.703, 1, 0.653 pu\\
         GFM reactive power set points & $q_{0,1}$, $q_{0,2}$, $q_{0,3}$& -0.075, -0.75, 0.348 pu\\
            $\hspace{-2mm}$GFM $RLC$ filter components $\hspace{-2mm}$& $r_{\mathrm{f},i}$, $l_{\mathrm{f},i}$, $c_{\mathrm{f},i}$& 0.01, 0.1, 0.1 pu\\
         \bottomrule
    \end{tabular}
    \label{tab:3bus_parameters}
    \end{subtable}
        \begin{subtable}{1\columnwidth}
    \vspace{1.8mm}
    \centering
    \caption{GFM control parameters (in local per unit systems).}
         \vspace{-2mm}
    \begin{tabular}{c||c|c}
    \toprule
         Parameter & Symbol & Value  \\ \hline
         Droop gains GFM 1 (fixed)& $d_\mathrm{p,1}$, $d_\mathrm{q,1}$&0.003 pu,0.01 pu \\
         Time constants GFM 1 (fixed) & $\tau_\mathrm{p,1}$, $\tau_\mathrm{q,1}$ & 0.1 s, 0.1 s\\\hline
Droop gains GFM 2 (fixed)& $d_\mathrm{p,2}$, $d_\mathrm{q,2}$&0.003 pu,0.01 pu \\
Time constants GFM 2 (fixed) & $\tau_\mathrm{p,2}$, $\tau_\mathrm{q,2}$ & 0.1 s, 0.1 s\\ \hline
Droop gains GFM 3 (no cond.)&$d_\mathrm{p,3}$, $d_\mathrm{q,3}$&0.25 pu, 0.25 pu\\
Time constants GFM 3 (no cond.) &$\tau_\mathrm{p,3}$, $\tau_\mathrm{q,3}$ & 0 s, 0 s\\ \hline
Droop gains GFM 3 (cond. L1)&$d_\mathrm{p,3}$, $d_\mathrm{q,3}$ &0.2 pu, 0.2 pu\\
Time constants GFM 3 (cond. L1) &$\tau_\mathrm{p,3}$, $\tau_\mathrm{q,3}$& 0.01 s, 0.01 s\\ \hline
Droop gains GFM 3 (cond. DYN)&$d_\mathrm{p,3}$, $d_\mathrm{q,3}$ &0.006 pu, 0.02 pu\\
Time constants GFM 3 (cond. DYN) &$\tau_\mathrm{p,3}$, $\tau_\mathrm{q,3}$& 0.1 s, 0.1 s\\ \bottomrule
    \end{tabular}
     \label{tab:converter_parameters}
    \end{subtable}
    \label{tab:case_study_table}
    \vspace{-4mm}
\end{table}
\renewcommand{\arraystretch}{1} \normalsize

The first two case studies consider an idealized three-bus system with three GFM VSCs. This minimal setup enables direct verification and interpretation of the theoretical results without confounding effects from large-scale system complexity or additional components such as transformers, heterogeneous devices, or diverse load models. In particular, it allows a comparison between (i) an analytical model based on ideal (i.e., linearized, reduced) block-diagram dynamics used to derive the stability certificates (\emph{Case Study I}), and (ii) the associated nonlinear electromagnetic transient (EMT) model with detailed network and device dynamics (\emph{Case Study II}).

The \emph{third case study} considers a more realistic system based on a modified IEEE nine-bus network with three GFM VSCs, one synchronous generator, and grid-following converter-based loads. A detailed nonlinear EMT model including inner converter control loops is used. Monte Carlo simulations sample GFM controller parameters within the feasible region defined by the proposed stability conditions, and the resulting responses are evaluated to assess system stability.

\subsection{Case Study I: Linearized \& reduced models}\label{sec:CaseStudyI}
We numerically validate \cref{thm:main_theorem} and \cref{corr:special_cases} using linearized small-signal models of a three-bus system with three GFM VSCs, as shown in \cref{fig:3bus} with parameters in \cref{tab:case_study_table}. We consider a network with uniform resistance-inductance ratio $\rho = 0.05$, a maximum steady-state voltage magnitude $|v|_\mathrm{max}=1.1$ pu, and identical self susceptances $\textstyle\sum_{i\ne j}^3b_{ij}=5$ pu of all VSC nodes $i\in\{1,2,3\}$. The simulations follow the ideal (i.e., linearized and reduced) block-diagram dynamics in \cref{fig:power_sys_model_orig}. GFM 1 and GFM 2 employ fixed controllers that always satisfy the conditions in \cref{eq:p_conditions,eq:q_conditions}, while for GFM 3 we are exploring varying control parameters.

\begin{figure}[t!]
    \centering
\begin{subfigure}{0.23\textwidth}
    \centering
    \vspace{-1mm}
\scalebox{0.52}{\includegraphics[]{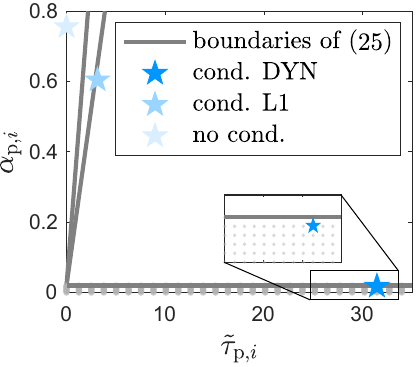}}
\vspace{-2mm}
    \caption{Stability conditions in \cref{eq:p_conditions}.}
    \vspace{-1mm}
    \label{fig:2d_p_sim}
\end{subfigure}
\hspace{0.15cm}
\begin{subfigure}{0.23\textwidth}
    \centering
    \vspace{-1mm}
\scalebox{0.51}{\includegraphics[]{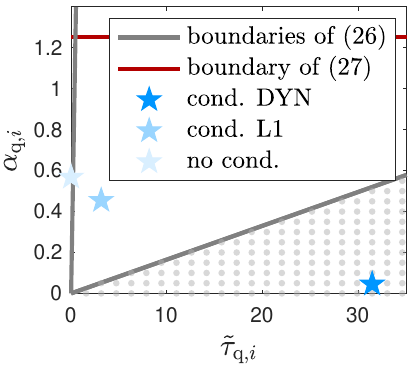}}
\vspace{-2mm}
    \caption{Stability conditions in \cref{eq:q_conditions}.}
    \vspace{-1mm}
    \label{fig:2d_q_sim}
\end{subfigure}
\caption{2D stability region for the three-bus system ($\rho=0.05$). The stars indicate the GFM 3 control parameters in the global per unit system.}
\label{fig:control_params_simulations}
\vspace{-2mm}
\end{figure}
\begin{figure}[t!]
    \centering
      \vspace{-1mm}
    \scalebox{0.5}{\includegraphics[]{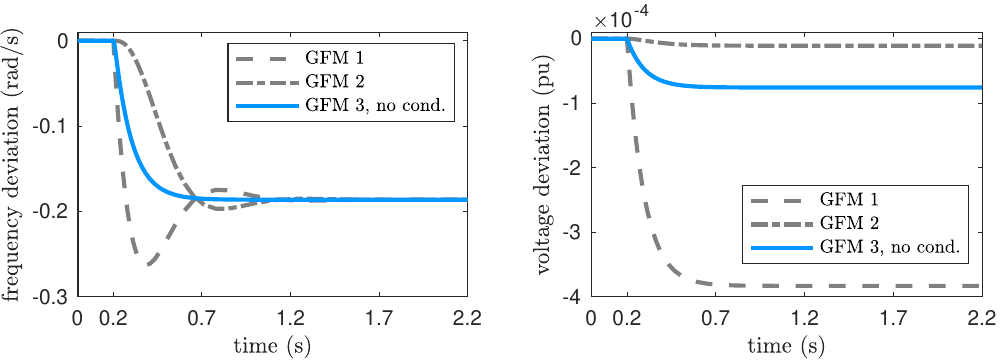}}
    \vspace{-2mm}
    \caption{\emph{Case Study I:} System response of the block diagram in \cref{fig:power_sys_model_orig} with three GFM devices and simplified network model \cref{eq:full_network_polar_matrix_blocks_level2}, where the controller of GFM 3 does not satisfy any of the conditions \cref{eq:p_conditions,eq:q_conditions,eq:level1_conditions} (no cond.).}
    \label{fig:L2_network_sim}
    \vspace{-4mm}
\end{figure}
\begin{figure}[t!]
    \centering
    \scalebox{0.5}{\includegraphics[]{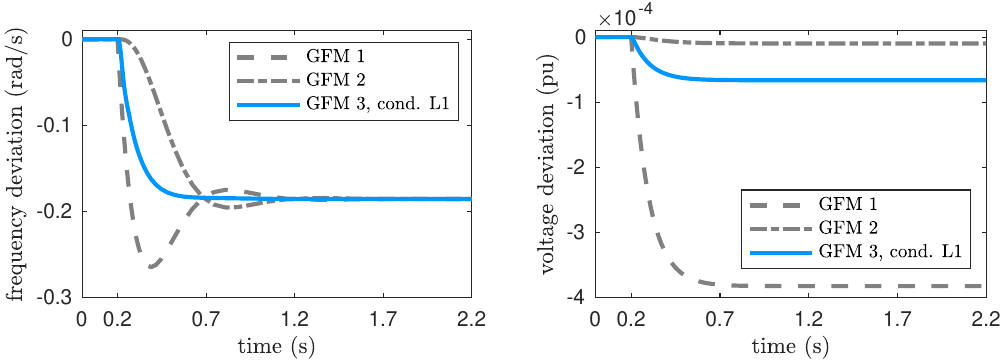}}
    \vspace{-2mm}
    \caption{\emph{Case Study I:} System response of the block diagram in \cref{fig:power_sys_model_orig} with three GFM devices and the simplified network model \cref{eq:full_network_polar_matrix_blocks_level1}, where the GFM 3 controller satisfies the conditions in \cref{eq:level1_conditions} (cond. L1).}
    \label{fig:L1_network_sim}
    \vspace{-4mm}
\end{figure}
\begin{figure}[t!]
    \centering
    \scalebox{0.5}{\includegraphics[]{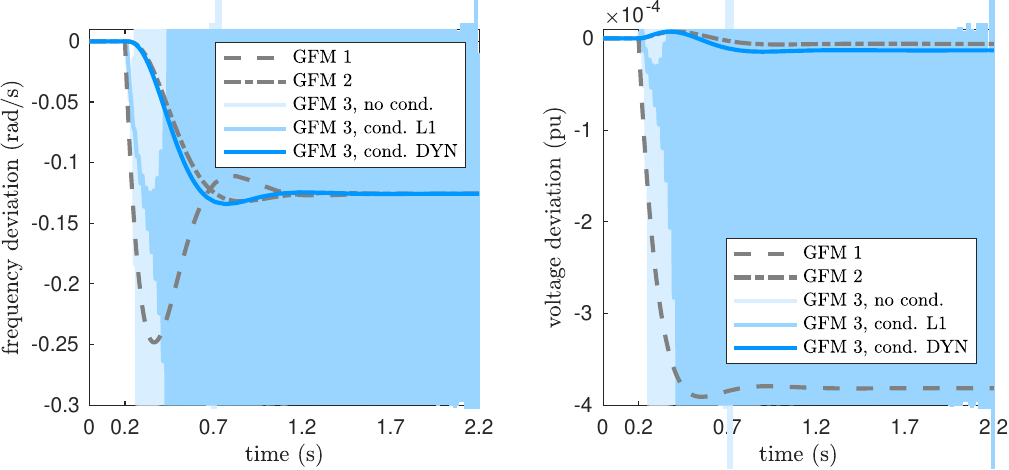}}
    \vspace{-2mm}
    \caption{\emph{Case Study I:} System response of the block diagram in \cref{fig:power_sys_model_orig} with three GFM devices and dynamic network model \cref{eq:full_network_polar_matrix_blocks_level0}, where all GFM controllers satisfy the conditions in \cref{eq:p_conditions,eq:q_conditions} (cond. DYN). We also display the unstable dynamics of GFM 3 when the controller satisfies no conditions (no cond.), and when it satisfies the conditions in \cref{eq:level1_conditions} (cond. L1).}
    \label{fig:Dyn_network_sim}
    \vspace{-4mm}
\end{figure}
We begin by modeling the network dynamics $\mathcal{N}_0(s)$ using Network-Simplification Level 2 in \cref{eq:full_network_polar_matrix_blocks_level2}, without imposing additional conditions on the controller of GFM 3 (see ``no cond.'' in \cref{fig:control_params_simulations}). A small-signal load disturbance at node 1 reveals that the closed-loop system remains stable (\cref{fig:L2_network_sim}), confirming the validity of \cref{corr:special_cases}. Likewise, when modeling $\mathcal{N}_0(s)$ with Network-Simplification Level 1 in \cref{eq:full_network_polar_matrix_blocks_level1} and ensuring that the controller of GFM 3 satisfies the stability conditions in \cref{eq:level1_conditions} (see ``cond. L1'' in \cref{fig:control_params_simulations}), we again observe closed-loop stability (\cref{fig:L1_network_sim}). This further confirms \cref{corr:special_cases}. Our main result, \cref{thm:main_theorem}, is validated in \cref{fig:Dyn_network_sim}, where we model $\mathcal{N}_0(s)$ dynamically as in \cref{eq:full_network_polar_matrix_blocks_level0}, and equip GFM 3 with a controller that meets the stability conditions in \cref{eq:p_conditions,eq:q_conditions} (see ``cond. DYN'' in \cref{fig:control_params_simulations}). Stability is immediately evident. In contrast, using a GFM 3 controller that satisfies only \cref{eq:level1_conditions} (``cond. L1'') or no conditions at all (``no cond.'') leads to instability in the dynamic network model. This underscores the importance of accurate network modeling in control design, which is overlooked in the overly optimistic stability assessment of \cite{pates2019robust,siahaan2024decentralized,watson2020control}.

\subsection{Case Study II: Simple nonlinear circuit model}
For the same three-bus system as in \cref{sec:CaseStudyI}, we perform detailed EMT simulations using a nonlinear three-phase circuit model. In particular, we now consider more general and realistic line models also including shunt capacitors. Additionally, for each GFM VSC, we incorporate the full nonlinear converter models (using average models). Compared to Case Study I, this increases complexity in two key aspects: first, by considering nonlinear models, and second, by accounting for the full network dynamics without the simplifying assumptions in \cref{sec:network}, while also including shunt capacitors. The simulation results for a small-signal load increase in \cref{fig:Nonlinear_circuit_sim} demonstrate that a GFM3 controller satisfying the decentralized stability conditions in \cref{eq:p_conditions,eq:q_conditions} maintains stability even in the presence of nonlinear network and device models.

In contrast, when the GFM 3 controller satisfies only \cref{eq:level1_conditions} or does not meet any stability conditions, interactions with the network dynamics can lead to instability. We thus conclude that our stability conditions in \cref{eq:p_conditions,eq:q_conditions} remain effective in a nonlinear circuit scenario including shunt capacitors, provided the system operates near the nominal point where linearization errors are small. However, the stability conditions, designed for a simplified static network, may fail.
\begin{figure}[t!]
    \centering
        \vspace{-1mm}
    \scalebox{0.5}{\includegraphics[]{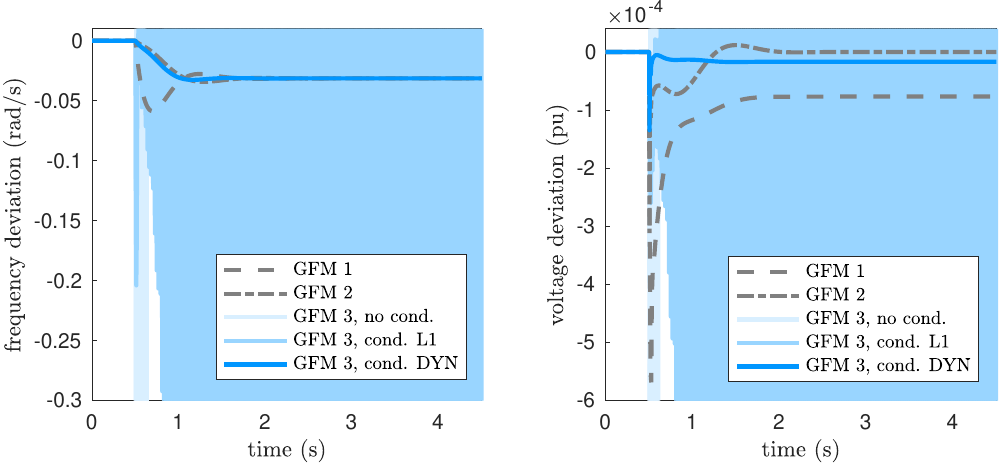}}
    \vspace{-2mm}
    \caption{\emph{Case Study II:} System response of the three-bus converter system with nonlinear circuit and device models (including shunt capacitors) where all GFM controllers satisfy the conditions in \cref{eq:p_conditions,eq:q_conditions} (cond. DYN). We also indicate the unstable dynamics of GFM 3 when the controller satisfies no conditions (no cond.), and the conditions in \cref{eq:level1_conditions} (cond. L1).}
    \label{fig:Nonlinear_circuit_sim}
    \vspace{-4mm}
\end{figure}

\subsection{Case Study III: Validation in IEEE Nine-Bus System}
To validate the stability conditions in \cref{eq:p_conditions,eq:q_conditions} under more realistic conditions, we consider the IEEE nine-bus system, comprising three GFM VSCs, one synchronous generator (SG), two grid-following (GFL) converter-based loads, and one static PQ load, as shown in \cref{fig:9bus}. We perform detailed nonlinear EMT simulations that capture network dynamics and inner converter control loops. We adopt network parameters from \cite{anderson}, and provide device parameters in \cref{tab:9bus_parameters}.

We perform a Monte Carlo analysis by randomly sampling controller parameters of all three GFM VSCs within the feasible region defined by \cref{eq:p_conditions,eq:q_conditions}, where we consider an approximate value of $\rho \approx 0.16$ of the Kron-reduced network. Specifically, we generate 50 distinct system-level parameter combinations, where each combination consists of one feasible realization of the controller $D_i(s)$ for each converter $i=1,2,3$ in \cref{eq:converter_model_final}. For each combination, we jointly sample the droop gains $d_{\mathrm{p},i}$, $d_{\mathrm{q},i}$ and time constants $\tau_{\mathrm{p},i}$, $\tau_{\mathrm{q},i}$ within their admissible ranges, as illustrated in \cref{fig:control_params_monte_carlo}. We then evaluate the small-signal response to load disturbances at bus 6.

For all 50 parameter combinations, stable small-signal behavior is observed, corresponding to a 100\% stability rate. Although the heterogeneous device composition (VSCs and SG) formally violates the homogeneity assumption underlying the theoretical conditions (while still permitting heterogeneous GFM droop controller parameters), the results indicate that the proposed conditions remain predictive when GFM VSC dynamics dominate. A rigorous extension to explicitly incorporate heterogeneous device dynamics is left for future work.

\begin{figure}[t!]
    \centering
    \vspace{-2mm}
    \resizebox{0.5\textwidth}{!}{
\begin{tikzpicture}[scale=0.4, every node/.style={scale=0.65}]
	
	\draw(-13,19.3) -- (6.1,19.4);
	\draw [ultra thick](-3.4,20) -- (-3.4,18.8);
	\draw [ultra thick](-15.2,19.9) -- (-15.2,18.7);
	
	\draw [ultra thick](-12.1,19.9) -- (-12.1,18.7);
	
	\draw [ultra thick](8.3,20) -- (8.3,18.8);
	
	\draw [ultra thick](5.2,20) -- (5.2,18.8);
	\draw[ultra thick] (-11.3,17.8) -- (-9.9,17.8);
	
	\draw [ultra thick](3,17.9) -- (4.4,17.9);
	\draw [ultra thick](-4.1,13.7) -- (-2.7,13.7);
	\draw [ultra thick](-4.1,11.1) -- (-2.7,11.1);
	\fill[black] (-15.2,19.3)circle (0.7 mm); 
	\fill[black]  (-12.1,19.3)circle (0.7 mm); 
	\fill[black] (-3.4,19.4)circle (0.7 mm); 
	\fill[black]  (5.2,19.4)circle (0.7 mm); 
	\fill[black] (8.3,19.4) circle (0.7 mm);

	\fill[black] (4.1,17.9)circle (0.7 mm); 
	\fill[black] (3.3,17.9)circle (0.7 mm); 
	\fill[black] (-10.2,17.8)circle (0.7 mm); 
	\fill[black] (-11,17.8)circle (0.7 mm); 
	\fill[black] (-3.8,13.7)circle (0.7 mm); 
	\fill[black] (-3,13.7)circle (0.7 mm); 
	\fill[black] (-3.4,13.7)circle (0.7 mm); 
	\fill[black]  (-3.4,11.1) circle (0.7 mm); 
	
	\node at (-15.2,20.4) {2};
	\node at (-12.1,20.4) {7};
	\node at (-3.4,20.5) {8};
	\node at (5.2,20.5) {9};
	\node at (8.3,20.5) {3};
	\node at (-11.7,17.8) {5};
	\node at (4.8,17.9) {6};
	\node at (-2.3,13.7) {4};
	\node at (-2.3,11.1) {1};

	\node (v2) at (-12.1,19) {};
	\node at (8.3,19.1) {};
	\draw (-12.1,19) -- (-11,19) -- (-11,17.8);
	\fill[black] (-12.1,19) circle (0.7 mm); 
	\draw (-10.2,17.8) -- (-10.2,17.1) -- (-3.8,14.4) -- (-3.8,13.7);
	\draw (3.3,17.9) -- (3.3,17.2) -- (-3,14.4) -- (-3,13.7);
	
	\draw (-3.4,13.7) -- (-3.4,13) node (v1) {};
	\draw (5.2,19.1) -- (4.1,19.1) -- (4.1,17.9);
	\fill[black] (5.2,19.1) circle (0.7 mm); 

	\draw(-3.4,12.6)  circle (4 mm); 
	\draw(-3.4,12.1)  circle (4 mm);

	\draw [-latex, thick](3.7,17.9) -- (3.7,16.7);
	\fill[black](3.7,17.9)circle (0.7 mm);

	\draw(-5.3,18.1)  circle (4 mm); 
	\draw(-4.8,18.1)  circle (4 mm);

	\draw(-2,18.1)  circle (4 mm); 
	\draw(-1.5,18.1)  circle (4 mm); 

	\draw(-13.9,19.3)  circle (4 mm); 
	\draw(-13.4,19.3)  circle (4 mm); 

	\draw(-10.6,16.2)  circle (4 mm); 
	\draw(-10.6,15.7)  circle (4 mm); 

	\draw(7,19.4)  circle (4 mm); 
	\draw(6.5,19.4)  circle (4 mm); 
	\draw (-15.7,19.3) -- (-14.3,19.3);
	\draw (8.75,19.4) -- (7.4,19.4);
	\draw (-3.4,10.65) -- (-3.4,11.7);

	\node at (-16.7,18.1) {GFM 2};
	\node at (9.8,18.15) {GFM 3};
	\node at (-7.5,16.8) {GFL load 2}; 
	
		\node at (-1.3,9.7) {GFM 1};
	
	\draw [fill=backgroundblue!30] (8.75,20.3) rectangle (10.85,18.5);

	\draw [fill=gray!30] (-8.6,19) rectangle (-6.4,17.2);

\draw[fill=backgroundblue!30]  (-4.45,10.65) rectangle (-2.35,8.85);

\draw[fill=gray!30]  (-11.7,14.8) rectangle (-9.6,13);

\draw[fill=backgroundblue!30]  (-17.8,20.25) rectangle (-15.7,18.45);

\draw (-17.8,18.45) -- (-15.7,20.25); 
\draw (8.75,18.5) -- (10.85,20.3);
\draw (-4.45,8.85) -- (-2.35,10.65);
\draw (-6.4,18.1) -- (-5.7,18.1); 
\draw (-4.4,18.1) -- (-3.8,18.1) -- (-3.8,19.1) -- (-3.4,19.1) node (v3) {};
\draw (-11.7,13) -- (-9.6,14.8);
\draw (-10.6,17.8) -- (-10.6,16.6);
\draw (-10.6,15.3) -- (-10.6,14.8);
\node at (-10.6,12.6) {GFL load 1};
\draw (-8.6,17.2) -- (-6.4,19);

\draw[fill=gray!30] (0.6,18.1) ellipse (1 and 1);
\draw (-3.4,19.1) -- (-3,19.1) -- (-3,18.1) -- (-2.4,18.1);
\node at (0.6,16.7) {SG};
\draw (-0.4,18.1) -- (-1.1,18.1);
\node[scale=1.2] at (-17.2,19.7) {$=$};
\node[scale=1.2] at (-16.4,19) {$\approx$};
\node[scale=1.2] at (9.4,19.8) {$\approx$};
\node[scale=1.2] at (10.1,19) {$=$};
\node[scale=1.2] at (-3.9,10.1) {$\approx$};
\node[scale=1.2] at (-3,9.4) {$=$};
\node[scale=1.2] at (-11.1,14.3) {$\approx$};
\node[scale=1.2] at (-10.2,13.5) {$=$};
\node[scale=1.2] at (-8,18.5) {$=$};
\node[scale=1.2] at (-7.1,17.7) {$\approx$};
\draw (0.1,18.4) -- (0.4,18.4) -- (0.4,17.8) -- (0.1,17.8); 
\draw (1.1,18.4) -- (0.8,18.4) -- (0.8,17.8) -- (1.1,17.8); 
\draw (0.1,18.4) arc (180:0:0.5);
\draw (0.1,17.8) arc (-180:0:0.5); 
\end{tikzpicture}

}
        \vspace{-10mm}
    \caption{IEEE nine-bus test system for validation of the stability conditions in \cref{eq:p_conditions,eq:q_conditions} during Monte Carlo Simulations in Case Study III.}
    \label{fig:9bus}
    \vspace{-1mm}
\end{figure}

\renewcommand{\arraystretch}{1}
\begin{table}[t!]\scriptsize
    \centering
               \caption{Nine-bus system parameters (in local per unit systems)}
               \vspace{-2mm}
               \begin{tabular}{c||c|c}
     \toprule
         Parameter&Symbol&Value\\ \hline
         $\hspace{-1mm}$Base power, voltage, frequency$\hspace{-1mm}$&$S_\mathrm{b}$, $V_\mathrm{b}$, $f_\mathrm{b}$&$\hspace{-1mm}$100 MVA, 230 kV, 50 Hz$\hspace{-1mm}$\\ 
         Bus base voltages&$\hspace{-1mm}$$V_{\mathrm{b},1}$, $V_{\mathrm{b},2}$, $V_{\mathrm{b},3}$&16.5, 18, 13.8 kV\\ 
         Constant power load (bus 6)&$p_{\mathrm{load},6}$&1 pu\\\hline
         GFL base power, voltage& $S_{\mathrm{gfl},i}$, $V_{\mathrm{b,gfl},i}$&100 MVA, 13.8 kV\\
         GFL loads (constant power)&$p_{\mathrm{load},i}$, $q_{\mathrm{load},i}$&1, 0 pu\\
         GFL $RL$ filter components&$r_{\mathrm{f},i}$, $l_{\mathrm{f},i}$&0.01, 0.1 pu\\ 
         GFL transformer components&$r_{\mathrm{t},i}$, $l_{\mathrm{t},i}$&0.01, 0.1 pu\\\hline
         SG base power, voltage&$S_{\mathrm{sg}}$, $V_{\mathrm{b,sg}}$&100 MVA, 13.8 kV\\
         SG power, voltage setpoints&$p_\mathrm{0,sg}$, $|v|_\mathrm{0,sg}$&0.525, 1 pu\\
         SG inertia \& droop constant&$H$, $D$&13 s, 0.04 pu\\ 
         SG transformer components&$r_{\mathrm{t}}$, $l_{\mathrm{t}}$&0.01, 0.1 pu\\\hline
         GFM base power&$S_1$, $S_2$, $S_3$&100, 100, 100 MVA\\
         GFM steady-state voltages&$\hspace{-1.5mm}|v|_\mathrm{0,1}$, $|v|_\mathrm{0,2}$, $|v|_\mathrm{0,3}\hspace{-1.5mm}$&1, 1, 1 pu\\
         GFM active power setpoints&$p_\mathrm{0,1}$, $p_\mathrm{0,2}$, $p_\mathrm{0,3}$&0.9, 0.836, 0.84 pu\\
         GFM reactive power setpoints&$q_\mathrm{0,1}$, $q_\mathrm{0,2}$, $q_\mathrm{0,3}$&$\hspace{-1mm}$-0.192, -0.182, -0.237 pu$\hspace{-1mm}$\\
GFM $RLC$ filter components&$r_{\mathrm{f},i}$, $l_{\mathrm{f},i}$, $c_{\mathrm{f},i}$&0.01, 0.1, 0.1 pu\\
         \bottomrule
    \end{tabular}
    \label{tab:9bus_parameters}
    \end{table}
\renewcommand{\arraystretch}{1}\normalsize

\begin{figure}[t!]
    \centering
\begin{subfigure}{0.23\textwidth}
    \centering
    \vspace{-1mm}
\scalebox{0.48}{\includegraphics[]{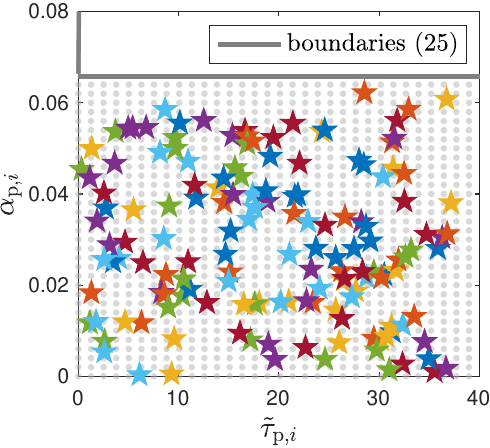}}
\vspace{-2mm}
    \caption{Stability conditions in \cref{eq:p_conditions}.}
    \vspace{-1mm}
    \label{fig:2d_p_sim_monte_carlo}
\end{subfigure}
\hspace{0.15cm}
\begin{subfigure}{0.23\textwidth}
    \centering
    \vspace{-1mm}
\scalebox{0.48}{\includegraphics[]{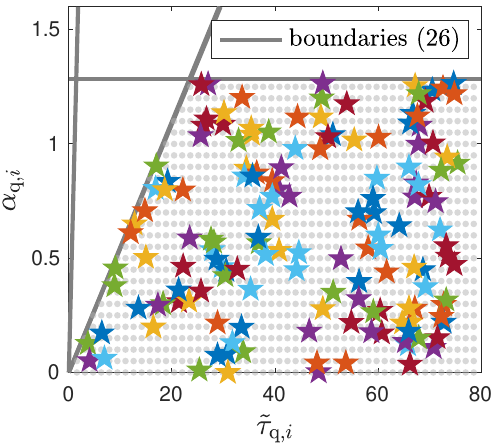}}
\vspace{-2mm}
    \caption{Stability conditions in \cref{eq:q_conditions}.}
    \vspace{-1mm}
    \label{fig:2d_q_sim_monte_carlo}
\end{subfigure}
\caption{2D stability region for the nine-bus system ($\rho=0.16$). Stars denote GFM control parameters in the global per unit system.}
\label{fig:control_params_monte_carlo}
\vspace{-4mm}
\end{figure}

\section{Conclusion}\label{sec:conclusion}
We proposed a decentralized analytic small-signal stability certification framework to mitigate the destabilizing effects of network dynamics on grid-forming converters. Using dynamic loop-shifting techniques and passivity theory, we derived parametric stability conditions that can serve as local tuning rules for device-level controllers, eliminating the need for centralized coordination. 

Future work includes the extension of our framework to a non-diagonal structure of the matrix blocks $\Gamma_i(s)$, thereby eventually reducing conservatism by requesting coupled device controllers. Beyond that, from an application point of view, we envision our stability framework to support the formulation of new grid codes for future power systems.

\renewcommand{\baselinestretch}{0.9}
\bibliographystyle{IEEEtran}
\fontdimen2\font=0.6ex
\bibliography{IEEEabrv,mybibliography}

\appendices
\section{Dynamic Small-Signal Network Model}\label{appendix1}
To derive the dynamic network model $N(s)$ in \cref{eq:full_network_polar,eq:full_network_polar_matrix_blocks}, we linearize \cref{eq:bus_voltage} and \cref{eq:branch_power_injections} around the equilibrium $v_{\mathrm{d0},i}$, $v_{\mathrm{q0},i}$, $i_{\mathrm{d0},i}$, $i_{\mathrm{q0},i}$ and transform them into the frequency domain, i.e.,
\begin{align}\label{eq:linearized_voltage_power}
    \begin{split}
     \Delta |v|_i(s) &\approx \tfrac{v_{\mathrm{d0},i}}{|v|_{0,i}} \Delta v_{\mathrm{d},i}(s)+\tfrac{v_{\mathrm{q0},i}}{|v|_{0,i}} \Delta v_{\mathrm{q},i}(s)\\
    \Delta\delta_i(s) &\approx -\tfrac{v_{\mathrm{q0},i}}{|v|_{0,i}^2}\Delta v_{\mathrm{d},i}(s)+\tfrac{v_{\mathrm{d0},i}}{|v|_{0,i}^2}\Delta v_{\mathrm{q},i}(s)\\
        \Delta p_{i}(s) &\approx v_{\mathrm{d0},i} \Delta i_{\mathrm{d},i}(s) + i_{\mathrm{d0},i} \Delta v_{\mathrm{d},i}(s)\\&\quad\quad\quad+v_{\mathrm{q0},i} \Delta i_{\mathrm{q},i}(s) +i_{\mathrm{q0},i} \Delta v_{\mathrm{q},i}(s)\\
        \Delta q_{i}(s) &\approx -v_{\mathrm{d0},i} \Delta i_{\mathrm{q},i}(s) - i_{\mathrm{q0},i} \Delta v_{\mathrm{d},i}(s)\\&\quad\quad\quad+v_{\mathrm{q0},i} \Delta i_{\mathrm{d},i}(s) +i_{\mathrm{d0},i} \Delta v_{\mathrm{q},i}(s)
           \end{split}
\end{align}
By using \cref{eq:reduced_dynamics_IV}, and inserting the steady-state expressions
\begin{align}
\begin{split}
    \hspace{-0.15cm}i_{\mathrm{d}0,i} \hspace{-0.05cm}&= \hspace{-0.05cm}\textstyle\sum_{i\neq j}^nb_{ij}\tfrac{1}{1+\rho^2}\hspace{-0.05cm}\left[\rho(v_{\mathrm{d0},i}\hspace{-0.05cm}-\hspace{-0.05cm}v_{\mathrm{d0},j}) \hspace{-0.05cm}+ \hspace{-0.05cm}(v_{\mathrm{q0},i}\hspace{-0.05cm}-\hspace{-0.05cm}v_{\mathrm{q0},j})\right]\hspace{-0.05cm}\\
     \hspace{-0.15cm}i_{\mathrm{q}0,i} \hspace{-0.05cm}&= \hspace{-0.05cm}\textstyle\sum_{i\neq j}^nb_{ij}\tfrac{1}{1+\rho^2}\hspace{-0.05cm}\left[-(v_{\mathrm{d0},i}\hspace{-0.05cm}-\hspace{-0.05cm}v_{\mathrm{d0},j}) \hspace{-0.05cm}+ \hspace{-0.05cm}\rho(v_{\mathrm{q0},i}\hspace{-0.05cm}-\hspace{-0.05cm}v_{\mathrm{q0},j})\right]\hspace{-0.05cm}
     \end{split}
\end{align}
and the steady-state bus voltages
\begin{align}
\begin{split}
v_{\mathrm{d0},i} &= |v|_{0,i}\cos{\delta_{0,i}}, \quad v_{\mathrm{d0},j} = |v|_{0,j}\cos{\delta_{0,j}} \\ 
v_{\mathrm{q0},i} &= |v|_{0,i}\sin{\delta_{0,i}}, \quad v_{\mathrm{q0},j} = |v|_{0,j}\sin{\delta_{0,j}}
\end{split}
\end{align}
into \cref{eq:linearized_voltage_power}, we can derive the small-signal dynamics $N(s)$ of the power network in polar coordinates as in \cref{eq:full_network_polar,eq:full_network_polar_matrix_blocks}.

\section{Dynamic Small-Signal Converter Model}\label{appendix2}
To derive the transfer matrix $D_i(s)$, we start by considering the small-signal dynamics of the filter's equations
    \begin{align}\label{eq:filter_equations1}
        \begin{split}
            \Delta v_{\mathrm{cd},i}(s) &= 
            \tfrac{l_{\mathrm{f},i}}{\omega_0}s\Delta i_{\mathrm{cd},i}(s)-l_{\mathrm{f},i}\Delta i_{\mathrm{cq},i}(s)+\Delta v_{\mathrm{d},i}(s)\\
            \Delta v_{\mathrm{cq},i}(s) &= l_{\mathrm{f},i}\Delta i_{\mathrm{cd},i}(s)+\tfrac{l_{\mathrm{f},i}}{\omega_0}s\Delta i_{\mathrm{cq},i}(s)+\Delta v_{\mathrm{q},i}(s),
        \end{split}
    \end{align}
    \begin{align}\label{eq:filter_equations2}
        \begin{split}
            \Delta i_{\mathrm{cd},i}(s) &= 
            \tfrac{c_{\mathrm{f},i}}{\omega_0}s\Delta v_{\mathrm{d},i}(s)-c_{\mathrm{f},i}\Delta v_{\mathrm{q},i}(s)+\Delta i_{\mathrm{d},i}(s)\\
            \Delta i_{\mathrm{cq},i}(s) &= 
            \tfrac{c_{\mathrm{f},i}}{\omega_0}s\Delta v_{\mathrm{q},i}(s)+c_{\mathrm{f},i}\Delta v_{\mathrm{d},i}(s)+\Delta i_{\mathrm{q},i}(s),
        \end{split}
    \end{align}
where the converter's local dq frame in SI units is given by the active power-frequency droop control with small-signal dynamics
\begin{align}\label{eq:pf_droop}
   \Delta \delta_i(s) = \tfrac{\Delta \omega_i(s)}{s} =- \tfrac{1}{s}\tfrac{d_{\mathrm{p},i}}{\tau_{\mathrm{p},i}s+1}\Delta p_i(s),
\end{align}
where $d_{\mathrm{p},i}\in\mathbb{R}$ is the active power droop gain and $\tau_{\mathrm{p},i}\in\mathbb{R}$ the low-pass filter time constant. The small-signal dynamic equations of the current control loop are given by
\begin{align}\label{eq:cc_loop}
\begin{split}
    \Delta v_{\mathrm{cd},i}^{\star}(s)&=\text{PI}_{\mathrm{cc},i}(s)(\Delta i_{\mathrm{cd},i}^\star(s)-\Delta i_{\mathrm{cd},i}(s))\\
    &\quad\quad\quad\quad\quad\quad\,\,\,\,\,+\Delta v_{\mathrm{d},i}(s)-l_{\mathrm{f},i}\Delta i_{\mathrm{cq},i}(s)\\
      \Delta v_{\mathrm{cq},i}^\star(s)&=\text{PI}_{\mathrm{cc},i}(s)(\Delta i_{\mathrm{cq},i}^\star(s)-\Delta i_{\mathrm{cq},i}(s))\\
    &\quad\quad\quad\quad\quad\quad\,\,\,\,\,+\Delta v_{\mathrm{q},i}(s)+l_{\mathrm{f},i}\Delta i_{\mathrm{cd},i}(s),
\end{split}
\end{align}
where $\text{PI}_{\mathrm{cc},i}(s)$ is the transfer function of the PI regulator. The current reference $\Delta i_{\mathrm{dq},i}^\star(s)$ in \cref{eq:cc_loop} comes from the voltage control loop with small-signal dynamics
\begin{align}\label{eq:vc_loop}
\begin{split}
    \Delta i_{\mathrm{cd},i}^\star(s)&=\text{PI}_{\mathrm{vc},i}(s)(\Delta v_{\mathrm{d},i}^\star(s)-\Delta v_{\mathrm{d},i}(s))\\
    &\quad\quad\quad\quad\quad\quad\,\,\,\,\,+\Delta i_{\mathrm{d},i}(s)-c_{\mathrm{f},i}\Delta v_{\mathrm{q},i}(s) \\
    \Delta i_{\mathrm{cq},i}^\star(s)&=\text{PI}_{\mathrm{vc},i}(s)(\Delta v_{\mathrm{q},i}^\star(s)-\Delta v_{\mathrm{q},i}(s))\\
    &\quad\quad\quad\quad\quad\quad\,\,\,\,\,+\Delta i_{\mathrm{q},i}(s)+c_{\mathrm{f},i}\Delta v_{\mathrm{d},i}(s).
\end{split}
\end{align}
The voltage reference in \cref{eq:vc_loop} is given by the reactive power-voltage droop control with the small-signal dynamics
\begin{align}\label{eq:qv_droop}
    \begin{split}
        \Delta v_{\mathrm{d},i}^\star(s)&=-\tfrac{d_{\mathrm{q},i}}{\tau_{\mathrm{q},i}s+1}\Delta q_i(s),\quad\quad\Delta v_{\mathrm{q},i}^\star(s)=0,
    \end{split}
\end{align}
where $d_{\mathrm{q},i}\in\mathbb{R}$ is the reactive power droop gain and $\tau_{\mathrm{q},i}\in\mathbb{R}$ the low-pass filter time constant. Finally, we insert the expressions in \cref{eq:filter_equations1,eq:filter_equations2,eq:pf_droop,eq:cc_loop,eq:vc_loop,eq:qv_droop} into the small-signal power injections
\begin{align}\label{eq:conv_power_injections}
\begin{split}
   \hspace{-1.2mm} \Delta p_{i}(s) &\approx v_{\mathrm{d0},i} \Delta i_{\mathrm{d},i}(s)\hspace{-0.5mm} + \hspace{-0.5mm}i_{\mathrm{d0},i} \Delta v_{\mathrm{d},i}(s)\hspace{-0.5mm}+\hspace{-0.5mm}i_{\mathrm{q0},i} \Delta v_{\mathrm{q},i}(s)\hspace{-1mm}\\
       \hspace{-1.2mm} \Delta q_{i}(s) &\approx -v_{\mathrm{d0},i} \Delta i_{\mathrm{q},i}(s) \hspace{-0.5mm}- \hspace{-0.5mm}i_{\mathrm{q0},i} \Delta v_{\mathrm{d},i}(s)\hspace{-0.5mm}+\hspace{-0.5mm}i_{\mathrm{d0},i} \Delta v_{\mathrm{q},i}(s),\hspace{-1mm}
\end{split}
\end{align}
linearized around the equilibrium $v_{\mathrm{d0},i}$, $v_{\mathrm{q0},i}=0$, $i_{\mathrm{d0},i}$, $i_{\mathrm{q0},i}$, such that the transfer matrix $D_i(s)$ in \cref{eq:local_conv_model} can be obtained as 
\begin{align}\label{eq:full_VSC_dynamics}
    D_i(s) = \begin{bmatrix}
        D_{{11},i}(s)&D_{{12},i}(s)\\D_{{21},i}(s)&D_{{22},i}(s) 
    \end{bmatrix},
\end{align}
with the matrix elements are given as
\begin{align}\label{eq:full_VSC_dynamics_elements}
    \begin{split}
        \hspace{-1.5mm}D_{11}(s)&=\tfrac{d_{\mathrm{p},i}}{\tau_{\mathrm{p},i}s+1}\\
        \hspace{-1.5mm}D_{12}(s)&=0\\
        \hspace{-1.5mm}D_{21}(s)&\hspace{-0.5mm}=\hspace{-0.5mm}\tfrac{1-G_{\mathrm{cc},i}(s)}{(G_{\mathrm{cc},i}(s)-1)(i_{\mathrm{d0},i}+i_{\mathrm{q0},i}-v_{\mathrm{d0},i}c_{\mathrm{f},i})+v_{\mathrm{d0},i}G_{\mathrm{cc},i}(s)\text{PI}_{\mathrm{vc},i}(s)+\tfrac{c_{\mathrm{f},i}s}{\omega_0}}\\
        \hspace{-1.5mm}D_{22}(s)&=\hspace{-0.5mm}\tfrac{v_{\mathrm{d0},i}G_{\mathrm{cc},i}(s)\text{PI}_{\mathrm{vc},i}(s)\tfrac{d_{\mathrm{q},i}}{\tau_{\mathrm{q},i}s+1}}{(G_{\mathrm{cc},i}(s)-1)(i_{\mathrm{d0},i}+i_{\mathrm{q0},i}-v_{\mathrm{d0},i}c_{\mathrm{f},i})+v_{\mathrm{d0},i}G_{\mathrm{cc},i}(s)\text{PI}_{\mathrm{vc},i}(s)+\tfrac{c_{\mathrm{f},i}s}{\omega_0}}\\
        &\quad\quad\quad\quad\quad-\tfrac{(1-G_{\mathrm{cc},i}(s))^2(i_{\mathrm{q0},i}+v_{\mathrm{d0},i}c_{\mathrm{f},i})}{(1-G_{\mathrm{cc},i}(s))i_{\mathrm{d0},i}+v_{\mathrm{d0},i}G_{\mathrm{cc},i}(s)\text{PI}_{\mathrm{vc},i}(s)+\tfrac{c_{\mathrm{f},i}s}{\omega_0}}
    \end{split}
\end{align}
where we have used the linearized expression for the voltage magnitude deviation $\Delta |v|_i(s)\approx\Delta v_{\mathrm{d},i}(s)$, and 
\begin{align}
    G_{\mathrm{cc},i}(s)=\tfrac{\text{PI}_{\mathrm{cc},i}(s)}{\tfrac{sl_{\mathrm{f},i}}{\omega_0}+\text{PI}_{\mathrm{cc},i}(s)}.
\end{align}

We assume that the timescales of the inner current and voltage control loops in \cref{eq:cc_loop,eq:vc_loop} are faster than the outer droop controls in \cref{eq:pf_droop,eq:qv_droop} and the network dynamics in \cref{eq:full_network_polar_matrix_blocks_level0}\cite{gross2022compensating}. We can thus neglect the inner VSC dynamics and thus approximate $\text{PI}_{\mathrm{cc},i}(s) \rightarrow \infty$ and $\text{PI}_{\mathrm{vc},i}(s) \rightarrow \infty$ for small $s$, such that \cref{eq:full_VSC_dynamics_elements} can be reduced as in \cref{eq:converter_model_final}.

\section{Proof of \cref{corr:special_cases}}\label{appendix3}
\subsubsection*{Quasistationary Network Model} For the quasistationary network model \cref{eq:full_network_polar_matrix_blocks_level1} in \cref{def:level1} ($s=0$ and $|v|_{0,i}\ne|v|_{0,j}$) we consider the same coordinate transformation as in \cref{eq:N_D}, followed by the loop-shifting in \cref{eq:N_dash_D_dash}. Since $s=0$ in \cref{eq:full_network_polar_matrix_blocks_level1}, we use a quasistationary version of the $\Gamma_i(s)$ with the diagonal elements 
\begin{align}
\begin{split}
   \Gamma_i^\mathrm{p}(s) = 0, \quad\quad
    \Gamma_i^\mathrm{q}(s) = \tfrac{1}{s}\tilde{\gamma}_{3,i}^\mathrm{q},
\end{split}
\end{align}
where $\tilde{\gamma}_{3,i}^\mathrm{q}\geq \tfrac{0.8}{1+\rho^2}\textstyle\sum_{j\ne i}^n b_{ij}$. With this choice, $\mathcal{N}'(s)$ is passive, i.e., it satisfies the conditions (i) to (iii) in \cref{def:passivity}:

\textit{(i) Poles:} $\mathcal{N}'(s)$ has one pole at $p=\mathrm{j}0$, i.e., $\text{Re}(p)\leq0$.

\textit{(ii) Positive semi-definiteness:} We can show that the Hermitian matrix $\mathcal{S}_{\mathcal{N}'}(\mathrm{j}\omega)\coloneqq \mathcal{N}'(\mathrm{j}\omega)+\mathcal{N}'^\star(\mathrm{j}\omega)$ is zero, i.e.,
\begin{align}\label{eq:N_0_N_0_pos_semi_def_level1}
    \begin{split}
\hspace{-1mm}\mathcal{S}_{\mathcal{N}'}(\mathrm{j}\omega)=\hspace{-1mm}\begin{bmatrix}
    \mathcal{S}_{\mathcal{N}',11}(\mathrm{j}\omega)&\dots&\mathcal{S}_{\mathcal{N}',1n}(\mathrm{j}\omega)\\\vdots&\ddots&\vdots\\\mathcal{S}_{\mathcal{N}',n1}(\mathrm{j}\omega)&\dots&\mathcal{S}_{\mathcal{N}',nn}(\mathrm{j}\omega)
\end{bmatrix} \hspace{-1mm}= 0_{2n\times 2n}
        \end{split}
    \end{align}
and therefore positive semi-definite.

\textit{(iii) Imaginary Poles:} $\mathcal{N}'(s)$ has one imaginary pole $p_1=\mathrm{j}0$, which is a simple pole. We compute $\mathcal{R}_{\mathrm{j}0}^{\mathcal{N}'}\coloneqq\text{lim}_{s\rightarrow \mathrm{j}0}(s-\mathrm{j}0)\mathcal{N}'(s)$, where each $\mathcal{R}_{\mathrm{j}0,ij}^{\mathcal{N}'}$ represents a $2\times2$ transfer matrix block. The diagonal and off-diagonal elements are given by
\begin{subequations}
\begin{align}\label{eq:residue_N_level1}
    \hspace{-2mm}\mathcal{R}_{\mathrm{j}0,ii}^{\mathcal{N}'} &=\textstyle\sum_{j\ne i}^nb_{ij}\begin{bmatrix}
           \tfrac{|v|_{0,i}|v|_{0,j}}{1+\rho^2}&0\\0&\tfrac{2|v|_{0,i}^2-|v|_{0,i}|v|_{0,j}+0.8}{1+\rho^2}
        \end{bmatrix}\\
\hspace{-2mm}\mathcal{R}_{\mathrm{j}0,ij}^{\mathcal{N}'} &=-b_{ij}\tfrac{|v|_{0,i}|v|_{0,j}}{1+\rho^2}\begin{bmatrix}
           1&0\\0&1
        \end{bmatrix},
\end{align}
\end{subequations}
i.e., $\mathcal{R}_{\mathrm{j}0}^{\mathcal{N}'}$ is a Hermitian diagonally dominant matrix with real non-negative diagonal entries. For the odd rows, we get
\begin{align*}
    |\textstyle\sum_{j\ne i}^nb_{ij}\tfrac{|v|_{0,i}|v|_{0,j}}{1+\rho^2}| \geq \textstyle\sum_{j\ne i}^n |-b_{ij}\tfrac{|v|_{0,i}|v|_{0,j}}{1+\rho^2}|.
\end{align*}
For the even rows we get (with $|v|_\mathrm{max}\hspace{-0.5mm}=\hspace{-0.5mm}1.1$ and $|v|_\mathrm{min}\hspace{-0.5mm}=\hspace{-0.5mm}0.9$):
\begin{align*}
    \begin{split}
        \hspace{-2mm}|\hspace{-0.5mm}\textstyle \sum_{j\ne i}^n\hspace{-0.6mm}\tfrac{(2|v|_{0,i}^2-|v|_{0,i}|v|_{0,j})b_{ij}}{1+\rho^2}+\tilde{\gamma}_{3,i}^\mathrm{q}|&\hspace{-0.6mm}\geq\hspace{-0.6mm} \textstyle\sum_{j\ne i}^n\hspace{-0.3mm}|\hspace{-0.3mm}-b_{ij}\tfrac{|v|_{0,i}|v|_{0,j}}{1+\rho^2}|\hspace{-0.7mm}\\
        \Leftrightarrow\textstyle \sum_{j\ne i}^n\hspace{-0.6mm}b_{ij}\tfrac{2|v|_{\mathrm{min}}^2-|v|_{\mathrm{max}}^2}{1+\rho^2}+\tilde{\gamma}_{3,i}^\mathrm{q}&\hspace{-0.6mm}\geq\hspace{-0.6mm} \textstyle\sum_{j\ne i}^n\hspace{-0.3mm}\hspace{-0.3mm}b_{ij}\tfrac{|v|_{\mathrm{max}^2}}{1+\rho^2}.\hspace{-0.7mm}
        \end{split}
\end{align*}
By \cref{lemma:gershgorin}, we conclude $\mathcal{R}_{\mathrm{j}0}^{\mathcal{N}'}\succeq 0$.

Next, we derive conditions under which $\mathcal{D}'(s)$ is strictly passive by verifying conditions (i) and (ii) in \cref{def:strict_passivity}. While doing so, we select $\tfrac{d_{\mathrm{q},i}}{|v|_{0,i}}=\tfrac{1}{\tilde{\gamma}_{3,i}^\mathrm{q}}$ to cancel the zero at the origin, and obtain the $2\times 2$ matrix elements of the form
\begin{align}
    \mathcal{D}_i'(s)=\begin{bmatrix}
        \tfrac{d_{\mathrm{p},i}}{\tau_{\mathrm{p},i}s+1}&0\\0&\tfrac{d_{\mathrm{q},i}}{\tau_{\mathrm{q},i}|v|_{0,i}}
    \end{bmatrix},
\end{align}
where we require
\begin{align}
    \alpha_{\mathrm{q},i}:= \tfrac{d_{\mathrm{q},i}}{|v|_{0,i}} \textstyle \sum_{j\ne i}^n b_{ij}< \tfrac{5(1+\rho^2)}{4} = c_{5,\rho},\quad \tau_{\mathrm{q},i} > 0. 
\end{align}

\textit{(i) Poles:}  The poles of all elements of $\mathcal{D}'(s)$ are in $\text{Re}(s)<0$.

\textit{(ii) Positive definiteness:} We compute 
\begin{align}
\mathcal{D}_i'(\mathrm{j}\omega)+\mathcal{D}_i'^\star(\mathrm{j}\omega)=\begin{bmatrix}
        \tfrac{2d_{\mathrm{p},i}}{1+\omega^2\tau_{\mathrm{p},i}^2} & 0 \\ 0 & 
         \tfrac{2d_{\mathrm{q},i}}{\tau_{\mathrm{q,i}}|v|_{0,i}}\end{bmatrix}\succ 0,
\end{align}
which holds $\forall \omega\in(-\infty,\infty)$. Therefore, since $\mathcal{N}'(s)$ is passive and $\mathcal{D}'(s)$ is strictly passive, and additionally $\bar{\sigma}(\mathcal{N}'(\mathrm{j}\infty))\bar{\sigma}(\mathcal{D}'(\mathrm{j}\infty))<1$ (because $\mathcal{N}'(\mathrm{j}\infty)=0_{2n\times 2n}$), we obtain internal feedback stability of $\mathcal{D}'\#\mathcal{N}'$ by \cref{thm:stabilty_passivity}. Finally, we conclude internal feedback stability of $\mathcal{D}_0\#\mathcal{N}_0$ by following the same arguments as in \cref{sec:proof_main_thm}-IV for $N(s)$ as in \cref{eq:full_network_polar_matrix_blocks_level1}.

\subsubsection*{Zero-Power Flow Network Model} 
For the zero-power flow network model \cref{eq:full_network_polar_matrix_blocks_level2} in \cref{def:level2} ($s=0$ and $|v|_{0,i}=|v|_{0,j}=|v|_{0}$), we can directly apply \cref{thm:stabilty_passivity} by showing passivity of $\mathcal{N}_0(s)$ and deriving conditions for $\mathcal{D}_0(s)$ to be strictly passive. In particular, $\mathcal{N}_0(s)$ is passive, i.e., it satisfies the conditions (i) to (iii) in \cref{def:passivity}:

\textit{(i) Poles:} $\mathcal{N}_0(s)$ has one pole at $p=\mathrm{j}0$, i.e., $\text{Re}(p)\leq0$.

\textit{(ii) Positive semi-definiteness:} We can express the Hermitian matrix $\mathcal{S}_{\mathcal{N}_0}(\mathrm{j}\omega)\coloneqq \mathcal{N}_0(\mathrm{j}\omega)+\mathcal{N}_0^\star(\mathrm{j}\omega)$ as
\begin{align}\label{eq:N_0_N_0_pos_semi_def_level2}
    \begin{split}
\mathcal{S}_{\mathcal{N}_0}(\mathrm{j}\omega)=\begin{bmatrix}
    \mathcal{S}_{\mathcal{N}_0,11}(\mathrm{j}\omega)&\dots&\mathcal{S}_{\mathcal{N}_0,1n}(\mathrm{j}\omega)\\\vdots&\ddots&\vdots\\\mathcal{S}_{\mathcal{N}_0,n1}(\mathrm{j}\omega)&\dots&\mathcal{S}_{\mathcal{N}_0,nn}(\mathrm{j}\omega)
\end{bmatrix},
        \end{split}
    \end{align}
    where each $\mathcal{S}_{\mathcal{N}_0,ij}$ represents a $2\times 2$ transfer matrix block. The diagonal and off-diagonal elements are given by
    \begin{align*}
        \begin{split}
            \mathcal{S}_{\mathcal{N}_0,ii}=\textstyle\sum_{j\ne i}^nb_{ij}\tfrac{|v|_0}{1+\rho^2}\begin{bmatrix}
                0&0\\0&2
            \end{bmatrix},\,\,\,
            \mathcal{S}_{\mathcal{N}_0,ij}=b_{ij}\tfrac{|v|_0}{1+\rho^2}\begin{bmatrix}
                0&0\\0&-2
            \end{bmatrix},
        \end{split}
    \end{align*}
which is a Laplacian matrix, i.e., $\mathcal{N}_0(\mathrm{j}\omega)+\mathcal{N}_0^\star(\mathrm{j}\omega)\succeq0$.

\textit{(iii) Imaginary poles:} $\mathcal{N}_0(s)$ has one imaginary pole, i.e., $p=\mathrm{j}0$, which is a simple pole. We therefore compute the limit $\mathcal{R}_{\mathrm{j0}}^{\mathcal{N}_0}\coloneqq \text{lim}_{s\rightarrow\mathrm{j}0}(s-\mathrm{j}0)\mathcal{N}_0(s)$, where each $\mathcal{R}_{\mathrm{j0},ij}^{\mathcal{N}_0}$ represents a $2\times 2$ transfer matrix block. The diagonal and off-diagonal elements are given by 
\begin{align*}
    \begin{split}
        \mathcal{R}_{\mathrm{j0},ii}^{\mathcal{N}_0} = \textstyle\sum_{j\ne i}^n b_{ij}\tfrac{|v|_0^2}{1+\rho^2}\begin{bmatrix}
            1&0\\0&0
        \end{bmatrix},\quad
         \mathcal{R}_{\mathrm{j0},ij}^{\mathcal{N}_0} = b_{ij}\tfrac{|v|_0^2}{1+\rho^2}\begin{bmatrix}
            -1&0\\0&0
        \end{bmatrix},
    \end{split}
\end{align*}
which is a Laplacian matrix, i.e., $\mathcal{R}_{\mathrm{j0}}^{\mathcal{N}_0}\succeq 0$.

Next, we derive conditions under which $\mathcal{D}_0(s)$ is strictly passive by verifying conditions (i) and (ii) in \cref{def:strict_passivity}.

\textit{(i) Poles:} The poles of all elements of $\mathcal{D}_0(s)$ are in $\text{Re}(s) <0$.

\textit{(ii) Positive-definiteness:} We compute 
\begin{align}
&\mathcal{D}_{0,i}(\mathrm{j}\omega)+\mathcal{D}_{0,i}^\star(\mathrm{j}\omega)=\begin{bmatrix}
        \tfrac{2d_{\mathrm{p},i}}{1+\omega^2\tau_{\mathrm{p},i}^2} & 0 \\ 0 & 
        \tfrac{2d_{\mathrm{q},i}}{1+\omega^2\tau_{\mathrm{q},i}^2} 
    \end{bmatrix}\hspace{-1mm}\succ 0,
\end{align}
which holds $\forall \omega\in(-\infty,\infty)$. Therefore, since $\mathcal{N}_0(s)$ is passive and $\mathcal{D}_0(s)$ is strictly passive, and additionally $\bar{\sigma}(\mathcal{N}_0(\mathrm{j}\infty))\bar{\sigma}(\mathcal{D}_0(\mathrm{j}\infty))<1$ (because $\mathcal{D}_0(\mathrm{j}\infty)=0_{2n\times 2n}$), we obtain internal feedback stability of $\mathcal{D}'\#\mathcal{N}'$ by \cref{thm:stabilty_passivity}.

In total, we conclude stability of $\mathcal{D}_0\#\mathcal{N}_0$ for all tunable local droop control parameters $d_{\mathrm{p},i}, d_{\mathrm{q},i}\in\mathbb{R}_{>0}$ and $\tau_{\mathrm{p},i}, \tau_{\mathrm{q},i}\in\mathbb{R}_{\geq 0}$.

\end{document}